\documentclass[article,onecolumn,doublespacing]{IEEEtran}
\usepackage[T1]{fontenc}
\usepackage{graphicx,colordvi,psfrag}
\usepackage{amsmath,amssymb,algpseudocode}
\usepackage{algorithm}
\usepackage{algorithmicx}
\usepackage{calc,pstricks, pgf, xcolor}
\usepackage{bbding}
\usepackage{bbm}
\usepackage{dsfont}
\usepackage{stfloats}
\usepackage{hyperref}
\usepackage{xparse}
\usepackage{enumerate}
\usepackage{tikz}
\usetikzlibrary{arrows.meta, positioning}

\usepackage{subcaption}

\newcommand{\Ind}{\mathds{1}}

\newcommand{\ZZ}{\mathbb{Z}}
\newcommand{\RR}{\mathbb{R}}

\newcommand{\Var}{\mathrm{Var}}

\newcommand{\trace}{\mathop{\mathrm{tr}}}

\newcommand{\Vol}{\mathrm{Vol}}

\newcommand{\eps}{\varepsilon}

\newcommand{\Unif}{\mathrm{Uniform}}
\newcommand\indep{\protect\mathpalette{\protect\independenT}{\perp}}
\def\independenT#1#2{\mathrel{\rlap{$#1#2$}\mkern2mu{#1#2}}}
\def\mreals{\mathbb{R}}
\def\eqdef{\triangleq}

\def\EE{\mathbb{E}}
\newcommand{\m}{\mathcal}

\newcommand{\eig}{\mathrm{eig}}
\newcommand{\rcov}{r_\mathrm{cov}}
\newcommand{\reff}{r_\mathrm{eff}}
\newcommand{\bones}{\mathbf{1}}

\newcommand{\Uidd}{\hat{U}_{[d],\mathrm{ideal}}}
\newcommand{\Vidd}{\hat{V}_{[d],\mathrm{ideal}}}
\newcommand{\OL}{\mathrm{OL}}
\newcommand{\eid}{e_\mathrm{ideal}}

\NewDocumentCommand{\DIP}{e{^_}}{D^{\mathrm{IP}\IfValueT{#1}{,#1}}_{\IfValueT{#2}{#2}}}

\NewDocumentCommand{\DMM}{e{^_}}{D^{\mathrm{MM}\IfValueT{#1}{,#1}}_{\IfValueT{#2}{#2}}}

\def\Cov{\mathrm{Cov}}

\DeclareMathOperator*{\argmin}{\arg\!\min}

\newtheorem{theorem}{Theorem}

\newtheorem{proposition}{Proposition}

\newtheorem{remark}{Remark}

\newtheorem{definition}{Definition}
\newtheorem{lemma}{Lemma}

\newenvironment{proof}[1][Proof]{\noindent\textbf{#1.} }{\ \rule{0.5em}{0.5em}}

\def\eqdef{\triangleq}

\begin{document}
\title{Optimal Quantization for Matrix Multiplication}

\author{Or Ordentlich 
and	Yury Polyanskiy 
\thanks{O. Ordentlich is with the 
Hebrew University of Jerusalem, Israel (\texttt{or.ordentlich@mail.huji.ac.il}).  Y. Polyanskiy is with the MIT,
USA (\texttt{yp@mit.edu}). The work of OO was supported by the Israel Science 
Foundation (ISF),
grant No. 1641/21. The work of YP was supported in part by the MIT-IBM Watson AI Lab and by the National Science
    Foundation under Grant No CCF-2131115. }
}

\date{}

\maketitle
\begin{abstract}
Recent work in machine learning community proposed multiple methods for performing lossy
compression (quantization) of large matrices. This quantization is important for accelerating
 matrix multiplication (main component of large language models), which is often bottlenecked by
 the speed of loading these matrices from memory. Unlike classical vector quantization and rate-distortion theory, the goal of these
new compression algorithms is to be able to approximate not the matrices themselves, but
their matrix product. Specifically, given a pair of real matrices $A,B$ an encoder (compressor) is
applied to each of them independently producing descriptions with $R$ bits per entry. These
representations subsequently are used by the decoder to estimate matrix product $A^\top B$. In
this work, we provide a non-asymptotic lower bound on the mean squared error  of this
approximation (as a function of rate $R$) for the case of matrices $A,B$ with iid Gaussian
entries.  Algorithmically, we construct a universal quantizer based on nested lattices with an
explicit guarantee of approximation error for any (non-random) pair of matrices $A$, $B$ in terms
of only Frobenius norms $\|\bar{A}\|_F, \|\bar{B}\|_F$ and $\|\bar{A}^\top \bar{B}\|_F$, where
$\bar{A},\bar{B}$ are versions of $A,B$ with zero-centered columns, respectively. For iid Gaussian matrices our
quantizer achieves the lower bound and is, thus, asymptotically optimal. A practical low-complexity
version of our quantizer achieves performance quite close to optimal. In addition, we derive
rate-distortion function for matrix multiplication of iid Gaussian matrices, which exhibits an
interesting phase-transition at $R\approx 0.906$ bit/entry, showing necessity of
Johnson-Lindestrauss dimensionality reduction (sketching) in the low-rate regime.
\end{abstract}

\tableofcontents

\section{Introduction and main results}



Matrix multiplication is a key component of many numerical algorithms, and is often the dominant factor in the runtime of a program. 
With the surge of deep neural nets (DNNs) and large language models (LLMs), finding more efficient ways to perform matrix multiplication 
have become one of the most pressing challenges. Classical work in this field focused
on minimizing the number of required
operations~\cite{strassen1969gaussian,fawzi2022discovering,duan2023faster,WXXZ24}. Specifics of
contemporary problems, however, require rethinking this
classical approach to matrix multiplication. First, in machine learning applications requirements for
precision of computing matrix products are quite lax. Second, modern computational hardware is
often bottlenecked by the memory bandwidth. A natural solution explored by many researchers is to apply lossy
compression to matrices leading to deterioration in precision  but improvement in the amount of data transferred between memory and
computation cores. 

We formalize  this problem as follows. Consider a pair of matrices $A\in\RR^{n\times a}$ and $B \in \RR^{n\times
b}$ which need to be described using $R$
bits per entry (using separate compressors), such that a decoder that obtains bit descriptions of
both matrices can estimate
$\widehat{A^\top  B}$. The metric for gauging quality of approximation that we will use is the
squared error between $ab$ entries of $\widehat{A^\top  B}$ and $A^\top B$. Note that unlike 
classical vector quantization, we are requiring compression algorithms to be tailored to the
special task of matrix multiplication. As a practical motivation, in Section~\ref{sec:i_pract} below we argue that reducing
$R$ down to a few bits/entry is necessary for LLMs to fully leverage modern matrix multiplication hardware.

Our main result shows existence of universal quantizers (based on lattices) which compress $A$ and
$B$ to $R$ bits/entry and come with explicit precision guarantees. Furthermore, we also show that 
these guarantees cannot be generally improved by proving a matching lower bound for the case of
matrices $A$ and $B$ with iid Gaussian entries. We emphasize, though, that quantizers 
\textit{are} universal and do not require Gaussian matrices. Indeed, following theoretical ideas
laid out in this paper a practical implementation of quantizers~\cite{savkin2025nestquant}
established state-of-the-art performance for quantizing weights and activations of various LLMs
ranging from 3B to 70B parameters.

To introduce our main results, let us define the function 
\begin{align}
\Gamma(R)=    \begin{cases}
1-\left(1-(2\cdot2^{-2R^*}-2^{-4R^*})\right)\frac{R}{R^*}   & R<R^* \\
2\cdot 2^{-2R}-2^{-4R}    & R\geq R^*
\end{cases}.
\label{eq:barphidef}
\end{align}
where $R^*\approx 0.906$ is the solution to the fixed-point equation
\begin{align}
R = \frac{1}{2}\log_2(1+4R\ln 2)
\end{align}
It will turn out that $\Gamma(R)$ is distortion-rate function for the matrix multiplication
of iid Gaussian matrices.

We say that a matrix $A\in \RR^{n\times m}$ has ``$M$-bounded entries'' if $|a_{i,j}|\in \{0\}\cup
[M^{-1},M]$ for all $i\in[n],j\in[m]$. Our results require the matrices $A$ and $B$ to have $M$-bounded entries, with $M=e^{o(n)}$. To be more concrete, throughout this paper we take $M=n^{10}2^{2000}$. In particular, this choice of $M$ guarantees that matrices represented in FP64 format have bounded entries. Furthermore, Gaussian iid matrices with finite variance and large $n$ have bounded entries with high probability. This extremely mild condition guarantees that we
can describe the $\ell_2$ norm of each column of $A,B$ with small multiplicative error using $o(n)$
bits (see Section~\ref{sec:arbitraryMM}). Let $\bones=(1,\ldots,1)^\top\in\RR^n$ be the all-ones vector. For a column vector $x\in\RR^n$ we denote by $\bar{x}=x-\left(\frac{1}{n}\bones^\top x\right)\bones$ its zero-centered version. For a matrix $A=[a_1|\cdots|a_a]\in\RR^{n\times a}$ we denote $\bar{A}=[\bar{a}_1|\ldots|\bar{a}_a]$.  Our first result is the following.

\begin{theorem}
For any $\eps>0$ and sufficiently large $n$, there exist randomized encoders $f_1:\RR^{n\times a}\to[2^{na R}]$, $f_2:\RR^{n\times b}\to[2^{nbR}]$, and decoders $g: [2^{na R}]\times [2^{nbR}]\to \RR^{a\times b}$ and $g_{1-\mathrm{sided}}:[2^{na R}]\times \RR^{n\times b}\to \RR^{a\times b}$ such that for any $A\in\RR^{n\times a}$ and $B\in\RR^{n\times b}$ with bounded entries we have
\begin{enumerate}
\item Let $C=A^\top B$, $\tilde{C}=\bar{A}^\top \bar{B}$, and $\hat{C}= g(f_1(A),f_2(B))$. Then, for any $i\in[a],j\in[b]$ we have
\begin{align}
\EE(C_{i,j}-\hat{C}_{i,j})^2\leq  \tilde{C}_{i,j}^2\cdot\left( \Gamma^2(R)+\eps\right)+\frac{\|\bar{a}_i\|^2 \|\bar{b}_j\|^2}{n}\left(\Gamma(R)-\Gamma^2(R)+\eps\right)+n^{-8},   
\end{align}
and, in particular, 
\begin{align}
\EE\|A^\top  B- g(f_1(A),f_2(B))\|_F^2&< \|\bar{A}^\top \bar{B}\|_F^2\cdot\left(\Gamma^2(R)+\eps\right)+\frac{\|\bar{A}\|^2_F \|\bar{B}\|_F^2}{n}\left(\Gamma(R)-\Gamma^2(R)+\eps\right)+a\cdot b\cdot n^{-8}.
\label{eq:generalMatMulUB}
\end{align}
\item  Let $C=A^\top B$, $\tilde{C}=\bar{A}^\top \bar{B}$, and $\hat{C}= g_{1-\mathrm{sided}}(f_1(A),B)$. Then, for any $i\in[a],j\in[b]$ we have
\begin{align}
\EE(C_{i,j}-\hat{C}_{i,j})^2\leq  \tilde{C}_{i,j}^2\cdot \left( 2^{-4R}+\eps\right)+\frac{\|\bar{a}_i\|^2 \|\bar{b}_j\|^2}{n}\left(2^{-2R}-2^{-4R}+\eps\right)+ n^{-8}.   
\label{eq:OneSIdedOpt2}
\end{align}
and, in particular, 
\begin{align}
\EE\|A^\top  B- g_{1-\mathrm{sided}}(f_1(A),B)\|_F^2&< \|\bar{A}^\top \bar{B}\|_F^2\cdot \left( 2^{-4R}+\eps\right)+\frac{\|\bar{A}\|^2_F \|\bar{B}\|_F^2}{n}\left(2^{-2R}-2^{-4R}+\eps\right)+a\cdot b\cdot n^{-8}.
\label{eq:OneSIdedOpt1}
\end{align}
\end{enumerate}
\label{thm:generalMatMul}
\end{theorem}

Note that two parts simply describe the cases, where both or only one matrix needs to be
compressed.\footnote{corresponding to the case of ``weights and attention'' quantization and
``weights-only'' quantization in LLMs.} Our scheme operates by compressing each column
of $A$ and $B$ using the same (randomized) nested-lattice quantizer $f_{\text{col}}:\RR^n \to [2^{nR}]$, which is applied
repeatedly to every column, whereas the decoder $g$ simply estimates each column to get matrices $\hat A$
and $\hat B$ and computes their scaled matrix product; see Figs.~\ref{fig:encoders}
and~\ref{fig:decoder}. The parameter $\kappa$ shown in Figures is used by the encoders for
time-sharing/sparsification and is set to $\kappa=\min\{R/R^*,1\}$ in the Theorem.
In particular, for $R<R^*$ a fraction $1-({R\over R^*})$ of coordinates are
ignored (mapped to 0), corresponding to $\kappa= R/R^*$. As we will see shortly this
dimensionality reduction (\`a la Johnson-Lindenstrauss) turns out to be necessary to achieve asymptotically optimal distortion.

To get a feel for Theorem~\ref{thm:generalMatMul} let us consider independent matrices $A$ and $B$ drawn
iid Gaussian $\mathcal{N}(0,\sigma^2)$. For large $n$, such matrices have bounded entries and
are also arbitrarily close to their centered version, with high probability. For the second part,
where only $A$ needs to be compressed, note that if $B$ is the $n\times n$ identity matrix, the
right hand sides of~\eqref{eq:OneSIdedOpt2} and~\eqref{eq:OneSIdedOpt1} read $\sigma^2
(2^{-2R}+2\eps)$ and  $na \sigma^2 (2^{-2R}+2\eps)$, respectively, which are optimal, as they
correspond to the Gaussian rate-distortion function. For the first part of the Theorem, we have that $\EE\|A^\top  B\|_F^2=\frac{\EE{\|A\|^2_F
\|B\|_F^2}}{n}=\sigma^4\cdot nab$ in this case. Consequently, the two terms corresponding to $\Gamma^2(R)$ in~\eqref{eq:generalMatMulUB} cancel out, and Theorem~\ref{thm:generalMatMul} shows estimate
$$ \EE[\|A^\top B - \widehat{A^\top B}\|_F^2] \le \sigma^4 nab (\Gamma(R) +\epsilon)\,.$$
It turns out that this is the best possible approximation (at this compression rate), as shown in
our next result.
\begin{theorem}
Let $A\in\RR^{n\times a}$ and $B\in\RR^{n\times b}$ be independent random matrices, with iid $\m{N}(0,\sigma^2)$ entries. 
\begin{enumerate}
\item For any $n\geq 1$, and any pair of rate-$R$ encoders $f_1:\RR^{n\times a}\to[2^{na R}]$, $f_2:\RR^{n\times b}\to[2^{nb R}]$ and decoder $g:[2^{naR}]\times [2^{nbR}]\to\RR^{a\times b}$, we have
\begin{align}
\EE\|A^\top  B-g(f_1(A),f_2(B))\|_F^2\geq \sigma^4 \cdot nab \cdot \Gamma(R). 
\label{eq:MMdistortionLowerBoundMain}
\end{align}
\item For any $n\geq 1$, and any rate-$R$ encoder $f:\RR^{n\times a}\to[2^{na R}]$ and decoder $g:[2^{naR}]\times \RR^{n\times b}\to\RR^{a\times b}$, we have
\begin{align}
\EE\|A^\top  B-g(f(A),B)\|_F^2\geq \sigma^4 \cdot nab \cdot 2^{-2R}. 
\label{eq:MMdistortionLowerBoundMainOneSided}
\end{align}
\end{enumerate}
\label{thm:GaussMatrixLowerBound}
\end{theorem}

In other words, the encoders $f_1,f_2,g$ from Theorem~\ref{thm:generalMatMul} attain the lower
bound from Theorem~\ref{thm:GaussMatrixLowerBound}, and are therefore asymptotically optimal for
this class of matrices. 

\medskip

We also show a simpler to use bound, based on our compression scheme applied with no ``MMSE scaling'' and no time-sharing - that is, with $\alpha=\kappa=1$ in Figures~\ref{fig:encoders},~\ref{fig:decoder}. The resulting bound does not meet the lower bound of Theorem~\ref{thm:GaussMatrixLowerBound} for Gaussian iid matrices. However, for moderate $R$ it is never much worse than the bound from Theorem~\ref{thm:generalMatMul}.
For some matrices $A,B$ it is significantly
better than the bound from Theorem~\ref{thm:generalMatMul}. 
\begin{theorem}
For any $\eps>0$ and sufficiently large $n$, there exist randomized encoders $f_1:\RR^{n\times a}\to[2^{na R}]$, $f_2:\RR^{n\times b}\to[2^{nbR}]$, and decoders $g: [2^{na R}]\times [2^{nbR}]\to \RR^{a\times b}$ and $g_{1-\mathrm{sided}}:[2^{na R}]\times \RR^{n\times b}\to \RR^{a\times b}$ such that for any $A\in\RR^{n\times a}$ and $B\in\RR^{n\times b}$ with bounded entries we have
\begin{enumerate}
\item Let $C=A^\top B$, and $\hat{C}= g(f_1(A),f_2(B))$. Then, for any $i\in[a],j\in[b]$ we have
\begin{align}\label{eq:nommse_innerp}
\EE(C_{i,j}-\hat{C}_{i,j})^2\leq  \frac{\|\bar{a}_i\|^2 \|\bar{b}_j\|^2}{n}\left(\frac{2\cdot 2^{2R}-1}{(2^{2R}-1)^2}+\eps\right)+n^{-8},   
\end{align}
and in particular
\begin{align}
\EE\|A^\top  B- g(f_1(A),f_2(B))\|_F^2&< \frac{\|\bar{A}\|^2_F \|\bar{B}\|_F^2}{n}\left(\frac{2\cdot 2^{2R}-1}{(2^{2R}-1)^2}+\eps\right)+a\cdot b\cdot n^{-8}.
\end{align}
\item  Let $C=A^\top B$, and $\hat{C}= g_{1-\mathrm{sided}}(f_1(A),B)$. Then, for any $i\in[a],j\in[b]$ we have
\begin{align}
\EE(C_{i,j}-\hat{C}_{i,j})^2\leq  \frac{\|\bar{a}_i\|^2 \|\bar{b}_j\|^2}{n}\left(\frac{1}{2^{2R}-1}+\eps\right)+n^{-8}.   
\end{align}
and in particular
\begin{align}
\EE\|A^\top  B- g_{1-\mathrm{sided}}(f_1(A),B)\|_F^2&< \frac{\|\bar{A}\|^2_F \|\bar{B}\|_F^2}{n}\left(\frac{1}{2^{2R}-1}+\eps\right)+a\cdot b\cdot n^{-8}.
\label{eq:thm3_a1}
\end{align}
\end{enumerate}
\label{thm:generalMatMulNoMMSE}
\end{theorem}

Note that the term $\|\bar{A}^\top  \bar{B}\|_F^2$ does not appear at all in
Theorem~\ref{thm:generalMatMulNoMMSE}, and whenever $\|\bar{A}^\top  \bar{B}\|^2\gg \frac{\|\bar{A}\|_F^2 \|\bar{B}\|_F^2}{n}$ the error in Theorem~\ref{thm:generalMatMulNoMMSE} is significantly smaller than the error in Theorem~\ref{thm:generalMatMul}.

\begin{figure}
    \begin{center}
    
\begin{tikzpicture}[node distance=1cm, thick, >=Stealth]
    \tikzstyle{block} = [rectangle, draw, minimum width=0.9cm, minimum height=1cm]
     \tikzstyle{blocksmall} = [rectangle, draw, minimum width=0.5cm, minimum height=1cm]
    \tikzstyle{sum} = [circle, draw, inner sep=0mm, minimum size=6mm]
    
    \node (input) [label] {$a_i$};
    \node (sum1) [sum, right=of input] {$-$};
    \node (sum2) [sum, right=of sum1] {$\times$};
    \node (block1) [blocksmall, right=of sum2, xshift=-0.3cm] {$S$};
    \node (block1b) [blocksmall, right=of block1] {$\m{P}_{\kappa n}$};
    \node (sum3) [sum, right=of block1b] {$+$};
    \node (block2) [block, right=of sum3, xshift=-0.3cm] {$Q_{\Lambda_f}(\cdot)$};
    \node (block3) [block, right=of block2, xshift=-0.3cm] {$\bmod Q_{\Lambda_c}$};
    \node (output) [label, right=of block3] {$(W_{a_i},\textcolor{red}{\widehat{\frac{1}{n}\bones^\top a_i}},\widehat{\|\bar{a}_i\|})$};

    \node (input1) [label, below=of sum1] {\textcolor{red}{$\left(\frac{1}{n}\bones^\top a_i\right)\bones$}};
    \node (input2) [label, below=of sum2] {$\frac{\sqrt{n}}{\|\bar{a}_i\|}$};
    \node (input3) [label, below=of sum3] {$Z_1$};

    \draw[->] (input) --(sum1);
    \draw[->] (sum1) -- (sum2) node[anchor=south,inner sep=2pt,midway] {$\bar{a}_i$};
    \draw[->] (sum2) -- (block1);
    \draw[->] (block1) -- (block1b) node[anchor=south,inner sep=2pt,midway] {$U_i$};
    \draw[->] (block1b) -- (sum3) node[anchor=south,inner sep=2pt,midway] {$U_{i,[\kappa n]}$};
    \draw[->] (sum3) -- (block2);
    \draw[->] (block2) -- (block3);
    \draw[->] (block3) -- (output) node[anchor=south,inner sep=2pt,midway] {$\tilde{U}_{i,[\kappa n]}$};;

    \draw[->,color=red] (input1)  -- (sum1);
    \draw[->] (input2) -- (sum2);
    \draw[->] (input3) -- (sum3);
\end{tikzpicture}

\begin{tikzpicture}[node distance=1cm, thick, >=Stealth]
     \tikzstyle{block} = [rectangle, draw, minimum width=0.9cm, minimum height=1cm]
     \tikzstyle{blocksmall} = [rectangle, draw, minimum width=0.5cm, minimum height=1cm]
    \tikzstyle{sum} = [circle, draw, inner sep=0mm, minimum size=6mm]
    
    \node (input) [label] {$b_j$};
    \node (sum1) [sum, right=of input] {$-$};
    \node (sum2) [sum, right=of sum1] {$\times$};
   \node (block1) [blocksmall, right=of sum2, xshift=-0.3cm] {$S$};
    \node (block1b) [blocksmall, right=of block1] {$\m{P}_{\kappa n}$};
    \node (sum3) [sum, right=of block1b] {$+$};
    \node (block2) [block, right=of sum3, xshift=-0.3cm] {$Q_{\Lambda_f}(\cdot)$};
    \node (block3) [block, right=of block2, xshift=-0.3cm] {$\bmod Q_{\Lambda_c}$};
    \node (output) [label,  right=of block3] {$(W_{b_j},\textcolor{red}{\widehat{\frac{1}{n}\bones^\top b_j}},\widehat{\|\bar{b}_j\|})$};

    \node (input1) [label, below=of sum1] {\textcolor{red}{$\left(\frac{1}{n}\bones^\top b_j\right)\bones$}};
    \node (input2) [label, below=of sum2] {$\frac{\sqrt{n}}{\|\bar{b}_j\|}$};
    \node (input3) [label, below=of sum3] {$Z_2$};

    \draw[->] (input) --(sum1);
    \draw[->] (sum1) -- (sum2) node[anchor=south,inner sep=2pt,midway] {$\bar{b}_j$};
    \draw[->] (sum2) -- (block1);
    \draw[->] (block1) -- (block1b) node[anchor=south,inner sep=2pt,midway] {$V_j$};
    \draw[->] (block1b) -- (sum3) node[anchor=south,inner sep=2pt,midway] {$V_{j,[\kappa n]}$};
    \draw[->] (sum3) -- (block2);
    \draw[->] (block2) -- (block3);
    \draw[->] (block3) -- (output) node[anchor=south,inner sep=2pt,midway] {$\tilde{V}_{j,[\kappa n]}$};

    \draw[->,color=red] (input1)  -- (sum1);
    \draw[->] (input2) -- (sum2);
    \draw[->] (input3) -- (sum3);
\end{tikzpicture}
\end{center}
\caption{Encoders for matrix multiplication. Each column of $A$ is encoded by the same encoder, and each column of $B$ is encoded by the same encoder. The encoder used for columns of $A$ and that used for columns of $B$ are also the same, except that for $A$ we use the dither vector $Z_1\in\RR^{\kappa n}$, whereas for $B$ we use the dither vector $Z_2\in\RR^{\kappa n}$. We illustrate the operation of the encoders on the $i$th column of $A$, $a_i\in\RR^n$, and on the $j$th column of $B$, $b_j\in\RR^n$. The block $S$ corresponds to left multiplication by the rotation matrix $S\in\RR^{n\times n}$, and the block $\m{P}_{\kappa n}$ corresponds to projecting the vector $U_i\in\RR^n$ (respectively $V_j\in\RR^n$) to $\RR^{\kappa n}$, $\kappa\in\frac{1}{n}\cdot\{0,1,\ldots,n\}$, by keeping only its first $\kappa n$ coordinates. Here, $\kappa$ is the time-sharing/sparsification parameter, determining the fraction of coordinates in each vector that are actually ``described'' to the decoder. The lattices $\Lambda_c\subset\Lambda_f\subset\RR^{\kappa n}$ are nested. The component $Q_{\Lambda_f}(\cdot)$ is a lattice quantizer which maps a point in $\RR^{\kappa n}$ to the closest lattice point in $\Lambda_f$. The component $\bmod\Lambda_c$ maps a point $x\in\RR^{\kappa n}$ to $x-Q_{\Lambda_c}(x)\in\m{V}_c$, where $\m{V}_c$ is the Voronoi region of $\Lambda_c$. The binary representation $W_{a_i}$ (respectively $W_{b_j}$) is an encoding of $\tilde{U}_{i,[\kappa n]}\in(\Lambda_f\cap \m{V}_c)\cong\Lambda_f/\Lambda_c$ (respectively $\tilde{V}_{j,[\kappa n]}\in\Lambda_f/\Lambda_c$) using $\log|\Lambda_f/\Lambda_c|$ bits. The scalars $\widehat{\frac{1}{n}\bones^\top a_i},\widehat{\|\bar{a}_i\|}$ (respectively, $\widehat{\frac{1}{n}\bones^\top b_j},\widehat{\|\bar{b}_j\|}$) are high-resolution descriptions of $\frac{1}{n}\bones^\top a_i,\|\bar{a}_i\|$ (respectively, $\frac{1}{n}\bones^\top b_j,\|\bar{b}_j\|$), which require only $O(\log n)$ bits.
The dither vectors $Z_1,Z_2$ must be known to the decoder. They can be randomly drawn by the
encoders and decoder and require sharing randomness between them (in practice, we just store random seed with the
matrices). The matrix $S$ need not be known by the decoder. 
The operations marked in red corresponds to zero-centering the column vectors, and may be avoided altogether. The effect of avoiding those operations on the performance is replacing $\bar{A}$ with $A$ and $\bar{B}$ with $B$ in the MSE upper bounds in Theorems~\ref{thm:generalMatMul},~\ref{thm:generalMatMulNoMMSE} and~\ref{thm:MostgeneralMatMul}.}
\label{fig:encoders}
\end{figure}

\begin{figure}
    \begin{center}
    
\begin{tikzpicture}[node distance=1cm, thick, >=Stealth]
    \tikzstyle{block} = [rectangle, draw, minimum width=0.9cm, minimum height=1cm]
     \tikzstyle{blocksmall} = [rectangle, draw, minimum width=0.5cm, minimum height=1cm]
    \tikzstyle{sum} = [circle, draw, inner sep=0mm, minimum size=6mm]
    
    \node (input1) [label] {$W_{a_i}$};
    \node (input2) [label, below=of input1, yshift=-1cm] {$W_{b_j}$};
    \node (block1LD) [block, right=of input1, xshift=-0.7mm] {$\Lambda_f/\Lambda_c$-decoder};
    \node (block2LD) [block, right=of input2, xshift=-0.7mm] {$\Lambda_f/\Lambda_c$-decoder};
    \node (sum1a) [sum, right=of block1LD] {$-$};
    \node (sum2a) [sum, right=of block2LD] {$-$};
    \node (block1mod) [block, right=of sum1a, xshift=-0.3cm] {$\bmod Q_{\Lambda_c}$};
    \node (block2mod) [block, right=of sum2a, xshift=-0.3cm] {$\bmod Q_{\Lambda_c}$};
    \node (prod) [block, right=of block1mod, yshift=-1.25cm] {$\langle\cdot,\cdot\rangle$};
    \node (iterm1) [label,right=of block1mod, xshift=0.5cm]{} ;
    \node (iterm2) [label,right=of block2mod]{};
    \node (prod2) [sum, right=of prod, xshift=-0.3cm] {$\times$};
    \node (prod2) [sum, right=of prod, xshift=-0.3cm] {$\times$};
    \node (prod3) [sum, right=of prod2, xshift=-0.3cm] {$\times$};
    \node (sumexpectations) [sum, right=of prod3] {$+$};
     \node (output) [label,right=of sumexpectations, xshift=-0.3cm]{$\widehat{(A^\top B)_{ij}}$} ;

   
    \node (input3) [label, below=of sum1a] {$Z_1$};
    \node (input4) [label, below=of sum2a] {$Z_2$};
     \node (prod2arg) [label, below=of prod2] {$\frac{1}{n}\widehat{\|\bar{a}_i\|}\widehat{\|\bar{b}_i\|}$};
     \node (prod3arg) [label, below=of prod3] {$\alpha$};
     \node (sumexpectationsarg) [label, below=of sumexpectations] {\textcolor{red}{$n\cdot\widehat{\frac{1}{n}\bones^\top a_i}\cdot \widehat{\frac{1}{n}\bones^\top b_j} $}};

    \draw[->] (input1) --(block1LD);
    \draw[->] (block1LD) -- (sum1a) node[anchor=south,inner sep=2pt,midway] {$\tilde{U}_{i,[\kappa n]}$};
    \draw[->] (sum1a) -- (block1mod);

     \draw[->] (input2) --(block2LD);
    \draw[->] (block2LD) -- (sum2a) node[anchor=south,inner sep=2pt,midway]{$\tilde{V}_{j,[\kappa n]}$};
    \draw[->] (sum2a) -- (block2mod);
    \draw[->] (block1mod)--(prod) node[anchor=south,inner sep=4pt,midway]{$\hat{U}_{i,[\kappa n]}$};;
     \draw[->] (block2mod)--(prod) node[anchor=north,inner sep=8pt,midway]{$\hat{V}_{j,[\kappa n]}$};;
    \draw[->] (prod) -- (prod2);
    \draw[->] (prod2) -- (prod3);
    \draw[->] (prod3) -- (sumexpectations);
    \draw[->] (sumexpectations)-- (output);

    \draw[->] (input3)  -- (sum1a);
    \draw[->] (input4) -- (sum2a);
    \draw[->] (prod2arg) -- (prod2);
    \draw[->] (prod3arg) -- (prod3);
    \draw[->,color=red] (sumexpectationsarg) -- (sumexpectations);
    
\end{tikzpicture}

\end{center}
\caption{Decoder for the matrix multiplication problem. We illustrate the estimation of $(A^\top B)_{ij}$. The component $\Lambda_f/\Lambda_c$-decoder maps $\log|\Lambda_f/\Lambda_c|$ bits to points in
$\Lambda_f\cap\m{V}_c\subset\RR^{\kappa n}$, where $\m{V}_c$ is the Voronoi region of the lattice
$\Lambda_c$. The component $\langle\cdot,\cdot\rangle$ computes the inner product
$\hat{U}_{i,[\kappa n]}^\top \hat{V}_{j,[\kappa n]}$, and $\alpha\in[0,1]$ is a (MMSE-like)
scaling coefficient. The operation marked in red need only be implemented if the encoders
implemented the corresponding zero-centering operations marked in red in Figure~\ref{fig:encoders}. Note that we can estimate
the entire product $A^\top B$ by first decoding $\hat{\tilde{A}}=[\hat{U}_{1,[\kappa n]}|\cdots|\hat{U}_{a,[\kappa n]}]$ and $\hat{\tilde{B}}=[\hat{V}_{1,[\kappa n]}|\cdots|\hat{V}_{b,[\kappa n]}]$, computing the matrix $\alpha \hat{\tilde{A}}^\top \hat{\tilde{B}}$, and then computing its Kronecker product with the rank-1 matrix $N$ whose $ij$th entry is $N_{ij}=\frac{1}{n}\widehat{\|\bar{a}_i\|}\widehat{\|\bar{b}_j\|}$, and adding to it the rank 1 matrix $\mu$ whose $ij$th entry is $\mu_{ij}=n\cdot\widehat{\frac{1}{n}\bones^\top a_i}\cdot \widehat{\frac{1}{n}\bones^\top b_j} $.}
\label{fig:decoder}

\end{figure}

To put Theorem~\ref{thm:generalMatMulNoMMSE} in context, we note that recent work in LLMs suggested to use random rotation of $A$ and $B$ and then quantize each column of the rotated matrices using sub-optimal lattice quantizers~\cite{tseng2024quip,quarot2024}. uch rotations, first practically employed by~\cite{chee2023quip}, make each column vector Gaussian-like, see more in Sections~\ref{sec:sketch} and~\ref{subsec:relatedwork}. A popular choice is to use the scalar quantizer, which is equivalent to quantizing to the lattice $\ZZ^n$ with a cubic shaping region. In practice, to apply the scalar quantizer on a vector $a_i\in\RR^n$, the common approach in the DNN and LLM literature is to store $\|a_i\|_{\infty}$, then normalize to $\tilde{a}_i=a_i/\|a_i\|_{\infty}$ such that all entries of $\tilde{a}_i$ are in $[-1,1]$, then use a $R$-bit scalar quantizer with dynamic range $[-1,1]$, and finally rescale the result by $\|a_i\|_{\infty}$. See, e.g.~\cite{quarot2024}. If the vector $a_i$ is uniform on $\sqrt{n}\mathbb{S}^{n-1}$ ( the sphere with radius $\sqrt{n}$), then for large $n$ we have that $\|{a}_i\|_{\infty}$ concentrates around $\sqrt{2\ln{n}}$. It follows that the expected squared quantization error this quantizer attains per entry is $\left(\frac{2}{3}\ln{n}\right) 2^{-2R}$. Using this quantizer for matrix multiplication (after rotating each matrix by the same random rotation) therefore results in 
\begin{align}\label{eq:factor_distortion}
\EE(C_{i,j}-\hat{C}_{i,j})^2\leq  \frac{\|a_i\|^2 \|b_j\|^2}{n}\left(\frac{2}{3}\ln{n}\right)\left(\frac{2\cdot 2^{2R}+\frac{2}{3}\ln{n}}{(2^{2R})^2}\right),~~\forall i\in[a],j\in[b].   
\end{align}
Thus, replacing the scalar quantizer $\ZZ^n$ with a high-dimensional pair of ``good'' nested lattices, as we do in the proof of Theorem~\ref{thm:generalMatMulNoMMSE} saves a factor of $\frac{2}{3}\ln{n}$ in the expected squared error for moderate $R$.

\medskip

The scheme used for proving Theorem~\ref{thm:generalMatMulNoMMSE} is based on using high-dimensional nested lattices with some asymptotically optimal properties. Unfortunately, such lattices do not lend themselves to efficient implementation. Another key contribution of this paper, described in Section~\ref{subser:practicalLattice}, is a simplified nested-lattice quantization
scheme, based on Conway and Sloane's Voronoi codes~\cite{ConwaySloane83}, that is similar to the one used in the proofs of Theorem~\ref{thm:generalMatMul} and Theorem~\ref{thm:generalMatMulNoMMSE}, but
uses low-dimensional nested lattices. For such lattices, we suggest a fast
implementation, whose computational efficiency does not depend on $R$.
This simplified scheme attains performance fairly close to theoretical
estimates therein. Recent subsequent work~\cite{savkin2025nestquant} suggests that the resulting matrix quantization scheme may be
a good candidate for practical application in LLMs and other algorithms relying on heavy matrix
multiplication operations.

Additional contributions of this work include the following:
\begin{itemize}
    \item We study the inner product case $a=b=1$, in full generality, assuming the entries of $A$ are drawn iid from distribution $P$, the entries of $B$ are drawn iid from distribution $Q$, and the rates $R_1$ and $R_2$ are not necessarily equal. We derive several upper and lower bounds on the smallest attainable distortion in computing the inner product, and prove some results on the structure of the optimal encoders and decoder. 
    \item For the matrix multiplication case, when the entries of $A$ and $B$ are drawn iid from a distribution $P$ with zero mean and variance $\sigma^2$, we show that~\eqref{eq:MMdistortionLowerBoundMain} continues to hold with $\Gamma(R)$ replaced by $\Gamma(R+D(P\|\m{N}(0,\sigma^2))$.  
\end{itemize}

Key ideas and proofs of these results are sketched in Section~\ref{sec:sketch}. We proceed to
motivation and review of the literature.

\subsection{Importance of quantization for modern applications}\label{sec:i_pract}
To set the stage for the problem, let us estimate what level of
quantization (in bits / entry) would be relevant for today's main consumer of matrix
multiplications: the large language models (LLMs). For those, quantization is typically
employed for accelerating inference. During inference LLM is busy computing many products $A^\top
B$ of matrices with sizes $d\times a$ and $d\times b$ respectively. This requires $2abd$ FLOPs and
$ad+bd+ab$  entries to load/store from memory. Ideally, we would want to quantize entries in such a
way that all compute is fully utilized. For that we need to know the ratio $\xi$ of available FLOPs to
available memory bandwidth, a quantity known as ``ops:bytes''
of a processor. It ranges from $\xi = 5\ldots20$ for modern CPUs (FP32 arithmetic via AVX512) to
$\xi\approx 300$ for the fastest GPUs (FP16 on an H100 or B200). The quantization rate saturating compute
should then be bounded (in bits/entry) as
\begin{align}
 R<\frac{16}{\xi}\frac{ab}{a+b + {ab\over d}}\,.
\end{align}
It turns out that there are two stages of running inference with LLMs: the
pre-fill (when the input prompt is processed) and the generation (when response tokens are
sequentially generated).
During the pre-fill LLM we have $a=d$ and $b=L$ ($d$ is the so-called hidden dimension and
$L$ is the sequence length), while during the generation we have $a=L$ and $b=1$ (the $A$ matrix
coming from KV-cache and $B$ matrix being new token's embedding). Thus, to saturate the computation core, we need
	$$ R_{\text{pre-fill}} = {16Ld\over \xi(d + 2L)}\,,\qquad R_{\text{generate}} = {16L\over
	\xi(L+1+L/d)} \approx {16\over \xi}\,.$$
We can see that during generation phase, on CPUs we would want to approach 1-3
bits/entry, while on GPUs we will not be able to ever saturate compute (that is, a decrease in quantization rate 
translates proportionally to decrease in runtime). For
the pre-fill phase, for large LLMs we get $R_{\text{generate}}>16$ bit (that is, just
storing plain FP16 is already good enough). Quantization during pre-fill might still be important
for ``small''
LLMs running on fast GPUs: for example, for BERT~\cite{BERTpaper} we
have $L=512$, $d=768$ and $\xi=300$ (for an H100), resulting in quantization rate $R\approx 11.7$ bit/entry.

Quantization is also important for reducing the power consumption of matrix multiplication. The energy cost of reading two bytes from DRAM is typically 2-3 orders of magnitude higher than that of performing a single INT8 addition/multiplication~\cite{horowitz20141}. Consequently, for memory-limited matrix multiplication, quantization may also have a dramatic effect on the power consumption.

\subsection{Sketch of the proof}\label{sec:sketch}
 This work started with the goal of trying to understand approximate
matrix multiplication for two matrices $A$ and $B$ which are random, with iid Gaussian entries $\mathcal{N}(0,1)$. We
started by trying to solve the case of $a=b=1$ (Sections~\ref{sec:iidIPdefsandBounds}
and~\ref{sec:simmetricIP}), i.e. when $A^\top B$ is simply an inner product of
two iid Gaussian vectors. 

Recall that the Gaussian distortion-rate function is $D(R) = 2^{-2R}$, e.g.~\cite[Section
26.1.2]{PWbook24}. A simple argument
(Thm.~\ref{thm:TSupperbound}) shows that
compressing $A$ to $\hat A$ and $B$ to $\hat B$  via rate-$R$ optimal Gaussian vector quantizer achieves error 
$$ \EE[(\hat A^\top  \hat B-A^\top  B)^2] \le \phi(D(R)), \qquad \phi(x) := 2x-x^2\,.$$
It turned out that the function $\phi(D(R))$ is monotonically decreasing but \textit{not} convex.
Thus, via time-sharing one can achieve a lower convex envelope of $\phi(D(R))$, which turns out to
be the $\Gamma(R)$ function defined in~\eqref{eq:barphidef}. 

We next proceed to lower bounds on distortion or, equivalently, to lower bounds on rate $R$
required for the existence of encoders $f_1,f_2$ and decoder $g$ satisfying
 \begin{equation}\label{eq:iprd_constraint}
 	\EE[(g(f_1(A),f_2(B))-A^\top B)^2] \le nD 
\end{equation} 
A
simple oracle bound (by revealing $B$ to the decoder) shows that rate $R$ cannot be smaller than
the standard Shannon rate-distortion function of $A$, see Theorem~\ref{thm:oracleLB}. However,
this bound leaves a wide gap with the achievability bound given above. Next, by a standard data-processing
argument (and observation that encoders for $A$ and $B$ can be without loss of generality be taken
identical) in Section~\ref{sec:iprd_lb} we deduce that~\eqref{eq:iprd_constraint} requires rate 
\begin{equation}\label{eq:multiletter}
	R\ge \limsup_{n\to \infty} {1\over n} \inf_{\hat A} \{I(A; \hat A): {1\over n} \sum_{i=1}^n
\phi(\lambda_i) \le D\}\,,
\end{equation}
where $A\sim \mathcal{N}(0,I_n)$, infimum is over all $\mathbb{R}^n$-valued random variables $\hat
A$ and $\{\lambda_i\}$ are the eigenvalues of $\Cov(A|\hat A)$. This reduces inner product quantization to
an optimization of a  multi-letter mutual information. Notice that the distortion constraint is no
longer separable, and hence the standard single-letterization (e.g.~\cite[Theorem 24.8]{PWbook24})
does not work and the limit on the right-hand side is not possible to evaluate. For
the special case of Gaussian distribution of entries of $A$ we were able to single-letterize the expression on the
right-hand side of~\eqref{eq:multiletter}, see Theorem~\ref{thm:SLBbasedDistortionLB}, showing
that left-hand side of~\eqref{eq:multiletter} evaluates to $\Gamma^{-1}(D)$.
Putting both upper and lower bounds together, we conclude that optimal compression rate for the iid Gaussian
inner product problem is thus given by $\Gamma^{-1}(D)$, see Theorem~\ref{thm:symgsn_iprd}.

We next proceed to solving the matrix case. Luckily, it turns out that for Gaussian iid matrices,
again, the optimal compression for matrix multiplication of $A^\top B$ is asymptotically achieved
by compressing each column separately via the use of optimal inner product quantizers, see
Theorems~\ref{thm:matrixerrorexpressions} and~\ref{thm:SLBbasedDistortionLBmatrix}.

Having solved the iid Gaussian case, we proceed to analyzing general (non-random) matrices and vectors.  Specifically, for the inner product problem we first normalize each of the two vectors to have norm $\sqrt{n}$  and these norms are compressed using a separate high-resolution scalar quantizer. Next, normalized vectors are multiplied by a common random orthogonal matrix. This makes each resulting vector uniformly distributed on the sphere of radius $\sqrt{n}$, while their inner product is unchanged. As is well known a high-dimensional vector that is uniform on the sphere is very
similar to an iid Gaussian vector (for example, in terms of joint distribution of small
$O(\sqrt{n})$-sized subsets). Thus, we reduce the problem to~\eqref{eq:iprd_constraint} except this time $A_i,B_i \stackrel{iid}{\sim}\mathcal{N}\left(0, \left(\begin{matrix}
    1 & \rho\\\rho & 1
\end{matrix}\right)\right)$, where $\rho=\frac{A^\top B}{\|A\|\cdot\|B\|}$. This slight change creates a crucial complication compared to the previous case of $\rho=0$.

Indeed, suppose we are only tasked with quantizing $B$ and $A$ is given to the decoder undistorted. Because of dependence between two terms in the product $A^\top (B-\hat B)$ we have to recourse to something like Cauchy-Schwarz, yielding
$$ \EE[(A^\top B - A^\top \hat B)^2 \le \EE[\|A\|^2 \|B-\hat B\|^2] = \Omega(n^2)\,.$$
Thus, using ``black box'' quantizers for $A$ and $B$ only yields $n^2$ performance guarantees violating~\eqref{eq:iprd_constraint}.
This is where \textit{lattice quantization} comes in. Specifically, using the idea of dithering we can make a (randomized) quantizer whose quantization error $(B-\hat B)$ becomes independent of $B$ and $A$. 

In order to guarantee finite quantization rate, we also need to ``truncate'' the infinite lattice, for which we use another key idea: a ``good'' nested lattice quantizer as in~\cite{erez2004achieving,ramiBook,oe17}. However, due to the nature of the problem we require construction of nested lattice pairs that satisfy stronger conditions than were known from  prior work  (see Theorem~\ref{thm:good_lattices}, whose proof builds upon heavy-lifting in a recent~\cite{orw22}). Overall, we construct quantizers for inner product problem of non-random vectors with a reconstruction error that depends only on the inner product
between the vectors and their individual $\ell_2$ norms, see Theorem~\ref{thm:universalLattice}. Since the performance bounds coincides with the lower bound
for the iid Gaussian case, it turns out that the resulting quantizers are optimal and generally
cannot be improved (except, possibly, in terms of finite-$n$ performance).
Together these steps complete proof of the main results quoted above.

\begin{remark}[On zero-centering]
One of the components in our quantization scheme is zero-centering, that is, subtracting from each vector its empirical mean prior to quantization. This is possible since the excess rate required for describing the empirical mean in high resolution is $O(\log n)$ bits and is negligible next to the cost of describing the entire vector which is $\Omega(n)$ bits. Describing the empirical mean in high resolution corresponds to describing the projection of the vector on $\frac{1}{n}\bones$. In principle, we can actually choose in advance a subspace of $o\left(\frac{n}{\log n}\right)$ orthogonal directions and describe the projection of the vector to this subspace with high resolution and negligible cost for the quantization rate. However, it is not clear what choice of such subspaces will result in a good rate-distortion tradeoff for arbitrary matrices. In fact, our choice of describing the DC component of the vector is also quite arbitrary, and was only done to allow restricting attention to zero-mean data without loss of generality, as can be done in vector quantization under quadratic distortion.      
\end{remark}

\subsection{Naive constructions: vector quantization, $\epsilon$-nets and sketching}

A crucial part of our construction is usage of common randomness between encoders and decoder. In
order to demonstrate its importance, here we first analyze two obvious constructions and find out
that they only achieve error of order $\Omega(n^2)$, while our
Theorem~\ref{thm:generalMatMulNoMMSE} yields error $O(n)$. 

The first natural question is to simply independently quantize $A$ and $B$ using a good vector
quantizer ($\epsilon$-net), to obtain $\hat A$ and $\hat B$, and then set $\widehat {A^\top B} = \hat A^\top \hat
B$. This idea, unfortunately, does not work even for vectors on a (scaled) sphere
$\mathbb{S}^{n-1}$ as  the next proposition shows. We stress that in this setting the right-hand
side of~\eqref{eq:nommse_innerp} is $O(n)$.

\begin{proposition} Fix rate $R>0$ and consider a special case of $A,B \in
\sqrt{n}\mathbb{S}^{n-1}$ (inner product estimation). Then there exists $\delta=\delta(R)>0$ and
$n_0=n_0(R)$ such
that for any $n\ge n_0$ and any deterministic encoders $f_1,f_2$ (of rate $R$ each) and
decoder $g$, there exist a pair of vectors $A,B$ with $\|A\|=\|B\|=\sqrt{n}$ such that 
	$$ \|A^\top B - \widehat{A^\top B}\|^2 \ge \delta n^2\,.$$ 
\end{proposition}
\begin{proof} Classical results, e.g.~\cite[Corollary 27.4]{PWbook24} show that there exists
$\delta' > 0$ and a
collection $\mathcal{C}$ of at least $2^{nR} + 1$ points on $\sqrt{n}\mathbb{S}^{n-1}$ such that
$$ \|A - \tilde A\| \ge \sqrt{n \delta'} \qquad \forall A\neq \tilde A \in \mathcal{C}\,.$$
By pigeonhole principle there should be such a pair with $f_1(A)=f_1(\tilde A)$. Thus, for such a
choice and any $B \in \sqrt{n}\mathbb{S}^{n-1}$ we have 
$$ \widehat{A^\top B} = \widehat{\tilde A^\top B}\,.$$
Choose $B= \sqrt{n}{A-\tilde A\over \|A-\tilde A\|}$, then by the previous display we have 
\begin{align*}
	\sqrt{n\delta'} &\le \|A-\tilde A\| = {1\over \sqrt{n}} \|A^\top B - \tilde A^\top B\| \\
		&=
{1\over \sqrt{n}} \|A^\top B - \widehat {A^\top B} + \widehat{\tilde A^\top B} - \tilde A^\top B\|
\le 
{1\over \sqrt{n}} \left(\|A^\top B - \widehat {A^\top B}\| + \| \tilde A^\top B - \widehat{\tilde
A^\top B} \|\right)\,. 
\end{align*}
Thus, one of the terms in the parentheses should be at least $n\sqrt{\delta'}/2$, completing the
proof.
\end{proof}

So what makes deterministic quantizers fail and our lattice quantizer work? Note that in the
former case, since $A$ and $B$ can be arbitrary the best we can do is apply  a Cauchy-Schwarz estimate
$$ (A^\top B - \hat A^\top B)^2 \le \|A-\hat A\|^2 \|B\|^2 \asymp n^2\,.$$
Dithered quantizers prevent adversarial alignment of the error vector $A-\hat A$ and $B$. However,
note that if we naively attempt to use dithered uniform quantizer $\epsilon \mathbb{Z}^n$, it would require
rate $\log_2 {\sqrt{n}\over \epsilon}$ due to $\|A\|_\infty$ possibly reaching out value of
$\sqrt{n}$ over the $\sqrt{n}\mathbb{S}^{n-1}$, thus also failing to provide finite-rate and
$O(n)$ MSE guarantee.

\smallskip
The next idea is that of sketching~\cite{alon1996space}, which is often used in the literature on
approximate matrix multiplication (see next section). Let $U\sim \mathcal{N}(0,I_n)$ and suppose
it is shared between the encoders $f_1$ and $f_2$. Let us then compute two scalar quantities $U_A\eqdef U^\top A$ and
$U_B \eqdef U^\top B$. Notice that 
$$ \EE[U_A U_B] = A^\top \EE[ U U^\top] B = A^\top B\,,$$
and hence the product $U_A U_B$ yields an unbiased estimate of the inner product. Now, since these
two scalar quantities are with high probability bounded by $O(\sqrt{n})$ they can be easily
quantized with exponentially small error given the budget of $nR$ bits. Unfortunately, this idea
also does not work, as a simple computation shows
$$ \Var[U_A U_B] = \EE[U_A^2 U_B^2] - (A^\top B)^2 = n^2 \EE[U_1^2 (\rho U_1 + \sqrt{1-\rho^2}U_2)^2]
- n^2 \rho^2 = n^2(1+\rho^2)\,,$$
where $\rho = {1\over n} A^\top B$ is the cosine of the angle between $A$ and $B$. We see that
even if $U_A$ and $U_B$ were provided to the decoder exactly, the error would still be
$\Omega(n^2)$ due to variance of the estimate $U_A U_B$. Of course, one could reduce variance by
working with multiple sketches, however, to drop it all the way to $O(n)$ one would need order $n$
sketches. 

Despite both of these constructions yielding $\Omega(n^2)$ errors, our algorithm combines the two
to obtain $O(n)$ error (and achieves the optimal constant): we first form the right number of sketches and then
convert them to bits via dithered (nested) lattice quantizers.

At this point we stress that even after computing sketches of $A$ and $B$ via random rotation, using vector quantizers that achieve the optimal Gaussian RDF may not suffice. In fact, even a Gaussian RDF achieving dithered nested lattice quantizer may not suffice. The reason for this is that the error in approximating an inner product as the inner product between the two quantized vectors consists also of the inner product between the quantization errors. In order to make sure this term has small variance, one needs to control also the Frobenius norm of the quantization noise autocorrelation matrix, and not just its trace as in the standard Gaussian vector quantization setup. Our analysis show that dithered nested lattice quantizers that posses these required properties do exist.

\subsection{Related work}
\label{subsec:relatedwork}

Randomized linear algebra/sketching, and locality-sensitive hashing (LSH) are techniques widely used in practice for computing approximate inner products and approximate matrix multiplications, as well as other operations, in reduced dimensions. The figure of merit in these fields is typically the tradeoff between the reduced dimension and the approximation error. Since the dimension of the reduced matrix/vector is related to the number of bits required for storing it, this body of work is relevant to our study. However, the tradeoff between the number of bits per dimension and the total approximation error, and its dependence on the properties of $A$, $B$ and $A^\top B$ is often subtle. Thus, there is no immediate translation between the required dimension of a sketch and the number of bits needed for representing it for obtaining the required accuracy.    

Many algorithms have been developed for randomized linear algebra, see~\cite{mahoney2011randomized,martinsson2020randomized} for a survey, and in particular for approximate matrix multiplication. A canonical example is the Monte-Carlo Matrix Multiplication (MCMM) algorithm~\cite{drineas2006fast} which randomly samples (the same) $c$ rows from $A\in\RR^{n\times a}$ and $B\in\RR^{n\times b}$ and estimates $A^\top  B$ as the (scaled) matrix multiplication of the sub-sampled matrices. Thus, each matrix is represented by $ac$ (respectively $bc$ real numbers), and the expected squared Frobenius norm of the approximation error is shown to scale like $O(\|A\|_F^2 \|B\|_F^2/c)$. Similarly, LSH for cosine similarity or $\ell_2$ distance also produce low-dimensional sketches of $u\in\RR^n$ and $v\in\RR^n$, from which the inner product of $u^\top  v$ can be approximated. Specifically, in~\cite{charikar2002similarity} it is proposed to project the two vectors using $c$ random projections (same random projections for both vectors) and only store the sign of each projection. The Hamming distance between the obtained vectors is distributed as $\mathrm{Binomial}\left(c,\frac{\theta(u,v)}{\pi}\right)$ where $\theta(u,v)=\cos\left(\frac{u^\top  v}{\|u\|\cdot \|v\|}\right)$, such that the expected squared error in estimating $\theta(u,v)$ is $O(1/c)$. In~\cite{datar2004locality} it is proposed to estimate $\|u-v\|$ (which is equivalent to estimating $u^\top  v$ for $u$ and $v$ on the sphere) by computing Gaussian linear combinations of each of them, and using a (dithered) scalar quantizer for quantizing each of the entries of the linear combinations. Specifically, for a vector $G\in\RR^{n}$, with iid $\m{N}(0,1)$ entries, we have that $G^\top (u-v)\sim\m{N}(0,\|u-v\|^2)$, and therefore the probability that $G^\top  u$ and $G^\top  v$ are  quantized (after dithering) to the same value is a monotone function of $\|u-v\|$. 

All the schemes mentioned above, as well as many other sketching/LSH schemes suffer from the same
shortcoming: their relative error $\frac{\EE\|\widehat{A^\top  B}-A^\top B\|_F^2}{\|A^\top B\|_F^2}$ scales like
$O\left(\frac{1}{c}\frac{\|A\|_F^2\|B\|_F^2}{\|A^\top  B\|_F^2} \right)$, and typically these schemes
are applied with constant $c$. When $\frac{\|A\|_F^2\|B\|_F^2}{\|A^\top  B\|_F^2}=\Omega(1)$, these
schemes perform remarkably well, despite the fact that $c$ does not grow with $n$. However, when
$\frac{\|A\|_F^2\|B\|_F^2}{\|A^\top  B\|_F^2}=\omega(1)$, as is the case for random iid matrices,
their relative error is very high. A notable exception is the algorithm proposed by Pagh
in~\cite{pagh2013compressed}, which represents each matrix using $n\cdot\min\{m,a\}$ (respectively
$n\cdot \min\{m,b\}$) real numbers, and produces an \emph{unbiased} estimator for $A^\top B$ with
expected error of $\EE\left((\widehat{A^\top  B})_{i,j}-(A^\top B)_{i,j}\right)^2=O\left(\frac{\|A^\top  B\|_F^2}{m} \right)$, for all $i,j$, and
does so with runtime proportional to $n^2+n m$ (ignoring logarithmic factors). When the product $A^\top  B$ is known to be highly sparse, this allows to estimate the sparsity pattern with $m$ proportional to the number of nonzero entries.

The topic of matrix quantization has received much attention in the last decade in the context of
DNNs and LLMs. The goal here is to reduce the memory footprint of the weight matrices, allowing to
load them to the main memory using less IOs, as well as speed up the multiplications and additions
operations by moving from floating point numbers to small integers (and when possible, also
sparsifying the matrices, saving some operations altogether). Roughly speaking, one can
distinguish between two paradigms: \emph{quantization-aware training}, where the training
procedure is designed to output weight matrices with ``cheap''
representation~\cite{guo2016dynamic,hubara2018quantized}, and \emph{post-training quantization},
where the training procedure is performed in high precision, and quantization of the weights is
only performed after training has terminated (perhaps with some fine tuning
afterwards)~\cite{gong2014compressing,jacob2018quantization,dettmers2022gpt3,yao2022zeroquant,xiao2023smoothquant,ma2024era,tseng2024quip}.
In order to further speed up matrix multiplication, and reduce the number of IOs needed for using KV-cache, some works also develop quantizers for the
activations~\cite{jacob2018quantization,yao2022zeroquant,xiao2023smoothquant,ma2024era}, while
other works assume the activations are kept in high
precision~\cite{gong2014compressing,tseng2024quip}. Quantization for DNNs and LLMs are typically
evaluated according to the end-to-end performance of the quantized architecture, but often the
Frobenius norm of the approximation error is considered as the intermediate optimization criterion
for quantizing the weights at each layer~\cite{guo2016dynamic,frantar2023optq}. Some works rely on
specific empirical observations on the distribution of weights and activations in LLMs. For
example~\cite{dettmers2022gpt3,yao2022zeroquant,xiao2023smoothquant} exploit the fact that a small
subset of entries in the activations have significantly larger magnitude than the majority of
entries. Notably, in~\cite{ma2024era} it is observed that for large LLMs, quantizing all
weights to $\{-1,0,1\}$ and the activations to $8$ bits, hardly affects the overall performance.
Among the work described above, the algorithm from~\cite{tseng2024quip} is closest to the scheme we
use in the proof of our Theorem~\ref{thm:generalMatMul} and Theorem~\ref{thm:generalMatMulNoMMSE}, as well as
the practical adaptation of the scheme used in those proofs, which is described in Section~\ref{subser:practicalLattice}. The
work~\cite{tseng2024quip} develops an algorithm for quantizing the weight matrices (keeping the
activations in high precision) that is based on random rotation using the randomized Hadamard
transform (that can be performed in time $n\log n$) and then using the $E_8$ lattice for
quantizing the rotated matrix. The mapping from lattice points to bits that was used in~\cite{tseng2024quip} required access to a lookup table, and was tailor-designed for $R=2$, while using different rates requires to further use successive refinement (residual vector quantization).
While our practical scheme in Section~\ref{subser:practicalLattice}
also uses product-lattice quantization, we use a nested lattice quantizer/Voronoi code~\cite{ConwaySloane83}, which results in a simple
mapping from lattice points to bits, regardless of $R$. Furthermore, we quantize both matrices to be multiplied. In LLMs, when activations/KV-cache data is also compressed, quantization must occur in inference time, and the encoders are required to be fast. On the other hand, when weights-only quantization is performed, the encoding is done offline, and only decoding is required to be efficient. The work~\cite{tseng2024qtip} leverages this asymmetry and builds a complicated encoder based on a trellis with a large number of states, while the decoder, on the other hand, is highly efficient. Such asymmetric schemes are not suitable for quantizing activations/KV-cache, whereas in the scheme we describe in Section~\ref{subser:practicalLattice} both the encoders and the decoder are efficient, and can be both applied in inference time.
In addition to reducing the limitations incurred by the memory bandwidth, an additional benefit of 
quantizing both matrices, is that one can replace the decoder with a lookup table, as
in~\cite{jegou2010product,blalock2021multiplying,stock2019and,guo2016quantization,ko25}, resulting in
very fast decoding in CPUs. 

Following~\cite{tseng2024quip}, the QuaRot~\cite{quarot2024} scheme also uses
randomized Hadamard transform prior to quantization, followed by $4$-bit scalar quantization of
each entry in both rotated matrices. Our implementation in Section~\ref{subser:practicalLattice}
quantizes the entries of the rotated matrices using nested-lattice codes, which come much closer
to the optimal rate-distortion tradeoff than scalar quantizers, with essentially the same
complexity (provided that the base lattice has an efficient nearest-neighbor decoder, as is the
case for essentially all ``good'' lattices in dimensions $2,3,4$ and $8$).

To the best of our knowledge, there was very little work on distributed compression for inner product/matrix multiplication in the information theory literature. Recently, Malak~\cite{malak2024distributed} studied the problem of \emph{lossless} distributed compression of binary random matrices for computing their product, and derived non-trivial bounds under stringent assumptions on the joint distribution. Some prior work considered the problem of distributed compression of random vectors with the goal of approximately computing a linear function of those vectors~\cite{krithivasan2009lattices,wagner2010distributed}. In those works, the goal was to estimate, say, the difference between the two vectors in $\RR^n$, which is itself a vector in $\RR^n$. While the inner product of these vectors, which is a scalar in $\RR$, can be computed from their difference (assuming their individual norms were encoded in high resolution), it seems, a-priory, that distributed compression for inner product computation is an easier task. Our results show that this is, in fact, hardly the case. Another line of related work in the information theory literature, is that of Ingber et al.~\cite{ingber2015compression} that considered the fundamental limits of lossy compression of a database in order to support approximate nearest neighbor search (see also~\cite{Ochoa14} for a practical implementation). We note in passing that much recent work focused on coding for speeding up distributed matrix multiplication by introducing redundancy for mitigating the effect of ``slow workers'' (stragglers), see, e.g.,~\cite{yu2020straggler}. This line of work is not directly related to approximate matrix multiplication via compression, studied in this paper.

Finally, one may wonder if approximating matrix product in mean squared error (MSE) metric is the
right distortion metric. Indeed, it was shown in~\cite{Tang22} that if the high-dimensional vectors to be compressed are
probability distributions and the metric is KL divergence (reasonable assumptions for attention
matrices in LLMs), the optimal quantization algorithms become quite different from the MSE ones.
We leave these extensions for future work.

\medskip

To summarize, the main innovations of this work with respect to prior work are:
\begin{enumerate}[a.]
\item Our work provides, for the first time, the fundamental limits of quantization for matrix multiplication. We derive a non-asymptotic lower bound on the error of any quantization algorithm for the case of Gaussian iid matrices, and develop a ``robust'' quantization algorithm (that makes no assumptions on the matrices $A$, $B$)  that asymptotically attains it. This gives a useful benchmark for evaluating the performance of any robust quantization algorithm.  
\item On the algorithmic side, we introduce several new components that were not used in previous work on quantization for matrix multiplication: sparsification/time-sharing, MMSE scaling, \emph{nested} lattice quantization. Those components, together with randomization in the form of rotation and dithering are required for attaining the optimal performance. For the analysis, we also prove new results on the existence of high-dimensional lattices with the required properties for quantized matrix multiplication.
\item We develop a low-complexity framework for approaching our theoretic lower bounds. Our framework is based on Voronoi codes in low dimensions, but together with an overload avoidance mechanism it nevertheless performs quite close to the asymptotic limits. It allows for fast encoding and decoding, and works for any $R>\log q$, where $q$ is an integer. Unlike prior work on lattice-based weights-only quantization, the same scheme can be used for any such rate, and the encoding/decoding complexity is invariant to the quantization rate $R$. We note that based on these ideas a follow-up paper~\cite{savkin2025nestquant} developed a practical algorithm (NestQuant), which attains state of the art in 4-bit weight+activation quantization of LLMs. 
\item Our lower and upper bounds give a theoretic justification for the widely used idea of applying the same random
rotation to both matrices $A$ and $B$ prior to quantization. In particular, the schemes used in the proofs of Theorem~\ref{thm:generalMatMul} and Theorem~\ref{thm:generalMatMulNoMMSE} are based on random rotation followed by quantizers based on ``good'' high-dimensional
nested lattices. Our analysis reveals that using $\ell_\infty$ normalization followed by quantization to $\ZZ^n$ on the rotated vectors (e.g.~\cite{quarot2024}) is highly sub-optimal, and using ``good'' nested lattices instead, leads to a multiplicative reduction by of factor $O(\log n)$ in the resulting distortion, see~\eqref{eq:factor_distortion}. 

 \end{enumerate}


\medskip
\medskip

\subsection{Paper organization}
 We begin our study with the special case where $a=b=1$, so that matrix
multiplication becomes an inner product. The reason is twofold: First, it is easier to gain
intuition for this problem, and all techniques for proving converse (impossibility) results for
the inner product case, easily extend  to the matrix multiplication case. The second reason is
that our achievability results are based on compression of each column of $A$ separately and compression of each column of $B$ separately, and
estimating each inner product forming $C_{ij}=(A^\top  B)_{i,j}=a_i^\top  b_j$ separately. In
Section~\ref{sec:iidIPdefsandBounds} we formally define the compression for inner product
computation problem, identify the structure of the optimal decoder, and give simple expressions on
the attained distortion as a function of the encoders $f_1$ and $f_2$, as well as a simple lower
bound on the distortion in terms of the ``standard'' distortion-rate function. In
Section~\ref{sec:simmetricIP} we restrict attention to the symmetric case where the two vectors
have the same distribution, and both encoders have the same rate. We prove lower and upper bounds
on the smallest attainable distortion in this case, which coincide in the Gaussian case. In
Section~\ref{sec:iidMM} we generalize the inner product results for matrix multiplication
$A^\top  B$ of $A\in\RR^{n\times a}$ and $B\in\RR^{n\times b}$, for general $a$ and $b$. Building on
the bounds developed for the inner product case of $a=b=1$, we prove lower and upper bound on the
smallest expected squared Frobenius norm of the error. In the special case where all entries in
both matrices are iid Gaussian, the lower bound results in
Theorem~\ref{thm:GaussMatrixLowerBound}. In Section~\ref{sec:arbitraryMM} we develop a
quantization scheme, based on randomly rotating both $A$ and $B$ by the same rotation matrix, and
then using nested-lattice quantizers for separately quantizing each column of the rotated
matrices, for qunatization of arbitrary matrices $A\in\RR^{n\times a}$ and $B\in\RR^{n\times b}$.
The expected squared Frobenius norm of the approximation error attained by this scheme is upper
bounded in Theorem~\ref{thm:MostgeneralMatMul}. Our main results stated above, Theorem~\ref{thm:generalMatMul} and Theorem~\ref{thm:generalMatMulNoMMSE}, are obtained as simple corollaries of Theorem~\ref{thm:MostgeneralMatMul}.
The upper bound depends on $A$ and $B$ only through
$\|\bar{A}\|_F,\|\bar{B}\|_F,\|\bar{A}^\top  \bar{B}\|_F$, and meets the lower bound from
Theorem~\ref{thm:GaussMatrixLowerBound} for the case where $A$ and $B$ have Gaussian iid
entries. The main component in the proof of Theorem~\ref{thm:MostgeneralMatMul} is a nested lattice quantization scheme for inner product computation of two arbitrarily correlated vectors, each uniformly distributed on the sphere. This coding scheme is presented and analyzed in Section~\ref{sec:nestedlatticeIP}. For the analysis we also prove new lattice existence results, stated in Theorem~\ref{thm:good_lattices}. 
Finally, in
Section~\ref{subser:practicalLattice} we introduce a practical implementation of the quantization
scheme from Theorem~\ref{thm:MostgeneralMatMul} for matrix multiplication of arbitrary matrices.
In these scheme, we describe several compromises in the choice of lattices used for quantization,
as well as in the rotation matrix used for rotating both $A$ and $B$. With these compromises the
quantization scheme and the decoder become extremely simple and fast. Some numerical evidence show
that, nevertheless, the resulting approximation error is quite close to the lower bound from
Theorem~\ref{thm:GaussMatrixLowerBound}. We conclude the paper with stating several key open problems in Section~\ref{sec:open}.

\subsection{Notation}

For $x>0$ we denote by $\log(x)$ the logarithm of $x$ taken to base $2$, and by $\ln(x)$ the
natural logarithm. We denote the Euclidean ($\ell_2$) norm of a vector $x\in\RR^n$ by
$\|x\|=\sqrt{\sum_{i=1}^n x_i^2}$ and its $\ell_1$ norm by $\|x\|_1=\sum_{i=1}^n |x_i|$.  For a
matrix $A\in\RR^{n\times m}$ we denote the trace operation as $\trace(A)=\sum_{i=1}^n A_{ii}$, and
the Frobenius norm is $\|A\|_F=\sqrt{\sum_{i=1}^n \sum_{j=1}^m A^2_{ij}}=\sqrt{\trace(A^\top A)}$.
The multiset of eigenvalues of a square matrix $A\in\RR^{n\times n}$ is denote by
$\eig(A)=(\lambda_1,\ldots,\lambda_n)$.  For $y>1$ we denote $[y]=\{1,\ldots,\lfloor y \rfloor\}$.
We denote the all ones vector in $\RR^n$ by $\bones=[1,\ldots,1]^\top$. 
For a column vector $x\in\RR^n$ we denote by $\bar{x}=x-\left(\frac{1}{n}\bones^\top x\right)\bones$ its zero-centered version.
For a matrix $A=[a_1|\cdots|a_a]\in\RR^{n\times a}$ with columns $\{a_i\}$ we denote
$\bar{A}=[\bar{a}_1|\ldots|\bar{a}_a]$ its column-centered version.
The differential entropy,
mutual information, and KL divergence are denoted by $h(\cdot)$, $I(\cdot~;~\cdot)$, and
$D(\cdot\|\cdot)$. The Gaussian distribution in $\RR^d$ with mean $\mu\in\RR^d$ and covariance
matrix $\Sigma$ is denoted $\m{N}(\mu,\Sigma)$. For two random variables $X,Y$ the notation
$X\indep Y$ means that they are statistically independent. For a distribution $P$ on alphabet
$\m{X}$, the distribution $P^{\otimes n}$ is its $n$-product (iid) distribution on $\m{X}^n$.

\section{Compression for Inner Product Computation: General Problem Setup and Simple Bounds}
\label{sec:iidIPdefsandBounds}

Let $P$ and $Q$ be distributions on $\RR$ with zero mean and unit variance, and let $U\sim P^{\otimes n}$ and $V\sim Q^{\otimes n}$ be statistically independent. As we argue below, the unit variance assumption is without loss of generality, and the zero mean assumption is essentially without loss of generality, as the effect of non-zero mean on the performance can be made negligible for large $n$. We consider the problem of quantizing $U$ and $V$ in order to compute their inner product $U^\top  V$. In particular, an $(n,R_1,R_2,D)$ code consists of mappings 
\begin{align}
 f_1&: \RR^n\to[2^{n R_1}]\\
 f_2&: \RR^n\to[2^{n R_2}]\\
 g&:[2^{n R_1}]\times [2^{n R_2}]\to \RR,
\end{align}
with
\begin{align}
D=\DIP=\frac{1}{n}\EE\left(U^\top  V-g(f_1(U),f_2(V))\right)^2.    
\end{align}
We define 
\begin{align}
\DIP^*_n(R_1,R_2)=\DIP^*_n(R_1,R_2,P,Q)=\inf\left\{D~:~\exists (n,R_1,R_2,D)-\text{code} \right\}.    
\end{align}
We further define the asymptotic function 
\begin{align}
\DIP(R_1,R_2)=\DIP(R_1,R_2,P,Q)=\limsup_{n\to\infty} \DIP^*_n(R_1,R_2).  
\end{align}

To see that the assumption that $P$ and $Q$ have unit variance is without loss of generality, assume that $\tilde{U}\sim \tilde{P}^{\otimes n}$ and $\tilde{V}\sim \tilde{Q}^{\otimes n}$, such that $\Var[\tilde{U}]=\sigma^2_1$ and $\Var[\tilde{V}]=\sigma^2_2$, 
and we would like to quantize $\tilde{U}$ and $\tilde{V}$ in order to estimate $\tilde{U}^\top \tilde{V}$. To that end we may define the unit-variance random variables $U=\frac{\tilde{U}}{\sigma_1}$ and $V=\frac{\tilde{V}}{\sigma_2}$ with corresponding distributions $P$ and $Q$, compress them using $f_1(U)$ and $f_2(V)$, and estimate the inner product as
\begin{align}
\widehat{\tilde{U}^\top \tilde{V}}=\sigma_1\sigma_2\cdot g(f_1(U),f_2(V)),
\end{align}
where $f_1,f_2,g$ attain $\DIP^*_n(R_1,R_2)$ for $P$ and $Q$. This scheme will achieve
\begin{align}
\EE\left(\tilde{U}^\top  \tilde{V}-\widehat{\tilde{U}^\top \tilde{V}}\right)^2=\sigma^2_1\sigma^2_2\cdot \EE\left(U^\top  V-g(f_1(U),f_2(V))\right)^2=\sigma^2_1\sigma^2_2\cdot \DIP^*_n(R_1,R_2).    
\end{align}
This must be optimal, since otherwise we could have attained a smaller distortion for $P$ and $Q$ by first scaling $U$ and $V$ by $\sigma_1$ and $\sigma_2$, respectively, feeding them to the better inner product quantization system, and scaling the obtained inner product estimate by $\frac{1}{\sigma_1\sigma_2}$.

Next, let us address the zero-mean assumption. Let $P$ and $Q$ be zero-mean distributions, $U\sim{P}^{\otimes n}$, $V\sim Q^{\otimes n}$, and let $\tilde{U}=U+\mu_U \bones$ and $\tilde{V}=V+\mu_V \bones$ for some $\mu_U,\mu_V\in\RR$. To encode $\tilde{U}$ and $\tilde{V}$ we may use the encoders $f_1,f_2$ designed for $P,Q$ and send 
\begin{align}
 \tilde{f}_1(\tilde{U})=\left[f_1(\tilde{U}-\mu_U \bones),\bones^\top (\tilde{U}-\mu_U \bones)\right]&=\left[f_1(U),\bones^\top U\right],\\
 \tilde{f}_2(\tilde{V})=\left[f_2(\tilde{V}-\mu_V \bones),\bones^\top (\tilde{V}-\mu_V \bones)\right]&=\left[f_2(V),\bones^\top V\right],
\end{align}
and estimate the inner product $\tilde{U}^\top\tilde{V}$ as
\begin{align}
\widehat{\tilde{U}^\top\tilde{V}}&=g\left( f_1(U),f_2(V)\right)+n\cdot\mu_U \mu_V+\mu_U\bones^\top V+\mu_V\bones^\top U,
\end{align}
so that
\begin{align}
\tilde{U}^\top\tilde{V}-\widehat{\tilde{U}^\top\tilde{V}}&=(U+\mu_U\bones)^\top (V+\mu_V\bones)-\left[g(f_1(U),f_2(V))+n\cdot\mu_U \mu_V+\mu_U\bones^\top V+\mu_V\bones^\top U\right]\nonumber\\
&=U^\top V-g(f_1(U),f_2(V)).
\end{align}
Thus, the error in the case of zero mean $U,V$ and non-zero mean $\tilde{U}=U+\mu_U\bones$, $\tilde{V}=V+\mu_V\bones$ can be made the same, at the expense of also sending a description of $\bones^T U$ and $\bones^T V$. As those are scalars, they can be described to high resolution, say $O(n^{-2})$ using $O(\log n)$ bits. Thus, for large $n$ and finite $R_1,R_2>0$, the cost of those descriptions is negligible.

Some of our bounds will rely on the distortion-rate function of $\RR$-valued source under quadratic distortion. An $(n,R,D)$ code for a source $U\sim P^{\otimes n}$ consists of an encoder $f:\RR^n\to[2^{nR}]$ and a decoder $g:[2^{nR}]\to \RR^n$ with $D=\frac{1}{n}\EE\|U-g(f(U))\|^2$. We denote by $D^*_n(R)=D^*_n(R,P)$ the smallest distortion attained by any $(n,R,D)$ code, and we denote the distortion-rate function by~\cite{PWbook24}
\begin{align}
D_P(R)=\lim_{n\to\infty}   D^*_n(R,P)=\min_{P_{Y|U}:I(U;Y)\leq R} \EE(U-Y)^2 .
\end{align}
It is also well-know~\cite{PWbook24}, that $D^*_n(R,P)\geq D_P(R)$ for any $n\geq 1$.

\subsection{Optimal Decoder and Error Expressions}

In the following, we assume $f_1$ and $f_2$ are fixed. We denote $W_U=f_1(U)$ and $W_V=f_2(V)$. Let $\hat{U}=\EE[U|W_U]$ and $\hat{V}=\EE[V|W_V]$.
\begin{proposition}
The optimal choice for $g$ is $g^*(W_U,W_V)=\hat{U}^\top \hat{V}$.
\end{proposition}

\begin{proof}
The minimum mean squared error (MMSE) estimator of a random variable $X$ from another random variable $Y$ is $\hat{X}=\EE[X|Y]$. Thus,
\begin{align}
g^*(W_U,W_V)=\EE[U^\top  V|W_U,W_V]=\EE[U^\top |W_U]\EE[V|W_V]=\hat{U}^\top  \hat{V},
\end{align}
where the second equality follows since $(U,W_U)\indep (V,W_V)$.
\end{proof}

Let $e_U=U-\hat{U}$ and $\Sigma_{e_U}=\EE[(U-\hat{U})(U-\hat{U})^\top ]$. Similarly, let $e_V=V-\hat{V}$ and $\Sigma_{e_V}=\EE[(V-\hat{V})(V-\hat{V})^\top ]$. Recall that by the orthogonality principle\cite[Chapter 4.2]{gersho2012vector}, it holds that $\EE[\hat{U} e_U^\top ]=0$ and $\EE[\hat{V} e_V^\top ]=0$.

\begin{proposition} Assuming that entries of $U$ and $V$ have zero mean and unit variance, we have
that the optimal decoder achieves 
\begin{align}
\EE\left(U^\top  V-g^*(W_U,W_V) \right)^2&=\frac{1}{n}\left[\trace(\Sigma_{e_V})+\trace(\Sigma_{e_U})-\trace(\Sigma_{e_U}\Sigma_{e_V}) \right]     
\end{align}
\label{prop:generalerror}
\end{proposition}

\begin{proof}
We have
\begin{align}
\DIP&=\EE\left((\hat{U}+e_U)^\top  (\hat{V}+e_V)-\hat{U}^\top \hat{V} \right)^2\\
&=\EE\left(\hat{U}^\top  e_V+ \hat{V}^\top  e_U+e_U^\top e_V \right)^2\\
&=\EE\left(\hat{U}^\top  e_V\right)^2+\EE\left( \hat{V}^\top  e_U\right)^2+\EE\left(e_U^\top e_V \right)^2,
\end{align}
where the last transition is due to the orthogonality principle and the statistical independence of $U$ and $V$. We have that
\begin{align}
\EE\left(\hat{U}^\top  e_V\right)^2=\trace\left[\EE[\hat{U}\hat{U}^\top  e_V e_V^\top ]\right]=\trace\left[\EE[\hat{U}\hat{U^\top }]\Sigma_{e_V} \right]
\end{align}
Recalling that $\EE[\hat{U}\hat{U^\top }]= I-\Sigma_{e_U}$, again, by the orthogonality principle, we obtain
\begin{align}
\EE\left(\hat{U}^\top  e_V\right)^2=\trace\left[(I-\Sigma_{e_U})\Sigma_{e_V} \right]= \trace(\Sigma_{e_V})-\trace(\Sigma_{e_V}\Sigma_{e_U}),\label{eq:trace1}
\end{align}
Similarly,
\begin{align}
\EE\left(\hat{V}^\top  e_U\right)^2=\trace(\Sigma_{e_U})-\trace(\Sigma_{e_V}\Sigma_{e_U}).
\end{align}
Finally,
\begin{align}
\EE\left(e_U^\top e_V \right)^2=\trace\left[\EE[e_U e_U^\top  e_V e_V^\top ] \right]=\trace(\Sigma_{e_U}\Sigma_{e_V}).
\end{align}
\end{proof}

\subsection{Simple Lower Bounds}

We show that computing the inner product with mean squared error (MSE) of $nD$ is necessarily
harder than compressing each of the random vectors $U$ and $V$ with $\ell_2$ norm of $nD$. Note that in the inner product quantization problem we are only interested in a single scalar in $\RR$ while in the standard problem of quantizing a random vector we are interested in a vector in $\RR^n$. Yet, the former problem is at least as hard as the latter.

\begin{theorem}
For any $n\geq 1$
\begin{align}
\DIP^*_n(R_1,R_2,P,Q)\geq \max\left\{D_P(R_1),D_Q(R_2) \right\},    
\end{align}    
and in particular
\begin{align}
\DIP(R_1,R_2,P,Q)\geq \max\left\{D_P(R_1),D_Q(R_2) \right\}.    
\end{align}    
\label{thm:oracleLB}
\end{theorem}

\begin{proof}
From Proposition~\ref{prop:generalerror} we have that for any $f_1:\RR^n\to[2^{nR_1}]$ and $f_2:\RR^n\to[2^{nR_2}]$
\begin{align}
\EE\left(U^\top  V-g^*(W_U,W_V) \right)^2&=\frac{1}{n}\left[\trace(\Sigma_{e_V})+\trace(\Sigma_{e_U})-\trace(\Sigma_{e_U}\Sigma_{e_V}) \right]\\
&=\frac{1}{n}\left[\trace(\Sigma_{e_U})+\EE(\hat{U}^\top  e_V)^2 \right]\label{eq:varianceeq}\\
&\geq\frac{1}{n}\trace(\Sigma_{e_U})\label{eq:noquanterror}\\
&\geq D^*_{n}(R_1,P),
\end{align}
where~\eqref{eq:varianceeq} follows from~\eqref{eq:trace1}, and the last inequality follows since $W_U$ is an encoding of $U$ with $2^{nR_1}$ codewords, which must incur distortion at least $D^*_n(R_1,P)$ by definition. Note that the inequality~\eqref{eq:noquanterror} holds with equality in the ``single-sided'' case where only $U$ is quantized while $\hat{V}=V$, so that $e_V=0$.
The bound $D^*_n(R_1,R_2,P,Q)\geq D^*_n(R_2,Q)$ follows similarly. Our statement now follows since $D^*_{n}(R_1,P)\geq D_P(R_1)$ and $D^*_{n}(R_2,Q)\geq D_Q(R_2)$ for any $n\geq 1$. 
\end{proof}

\section{Compression for Inner Product Computation: The Symmetric Case}
\label{sec:simmetricIP}

In this section we assume $P=Q$, $R_1=R_2=R$, and define $\DIP_n^*(R,P)=D_n^*(R,R,P,P)$, and
$\DIP(R,P)=\DIP(R,R,P,P)$. We first develop a simple upper bound based on using the same
encoder for both vectors (that is $f=f_1=f_2$), that time-shares between a ``good'' encoder for
$P$ under quadratic distortion, and a zero-rate encoder. We then develop a lower bound on the
distortion of inner product compression, which shows that for the symmetric case, using the same
encoder $f=f_1=f_2$ for both $U$ and $V$ is optimal, and depends on the spectrum of the covariance
matrix of $e_U=U-\EE[U|f(U)]$. We then give some constraints on the error spectrum that can be
attained by a rate $R$ encoder. Using this characterization we obtain a general lower on
$\DIP(R,P)$. Thus, overall we show in this section that for any iid source $P=Q$ with $\EE_{U_i\sim
P}[U_i] = 0$, $\EE_{U_i\sim
P}[U_i^2] = 1$ we have 
	$$  \Gamma(R+D(P\|\mathcal{N}(0,1))) \le \DIP(R,P) \le \Gamma(R)\,,$$
in particular showing that $\DIP(R,\mathcal{N}(0,1)) = \Gamma(R)$. 


\subsection{Upper Bound}

Define the function 
\begin{align}
\phi(x)=2x-x^2.    
\end{align}
and note that $x\mapsto\phi(x)$ is increasing and concave on $[0,1]$. 
We give a time-sharing upper bound on $\DIP(R,P)$ in terms of $\phi(D_P(R))$.

\begin{theorem}
Assuming that
$P$ has zero mean and unit variance, we have
 \begin{align}
 \DIP(R,P)\leq \min_{0\leq \kappa\leq 1} (1-\kappa)+\kappa\cdot
 \phi\left(D_{P}\left(\frac{R}{\kappa}\right) \right) \le \Gamma(R)\,,
 \end{align} 
where $\Gamma(R)$ is defined in~\eqref{eq:barphidef}. 
 \label{thm:TSupperbound}
\end{theorem}

\begin{proof} Note that it is sufficient to prove the first inequality. Indeed, we know that $D_P(R)\le
D_{\mathcal{N}(0,1)}(R)=2^{-2R}$, e.g.~\cite[Theorem 26.3]{PWbook24}, and in
Appendix~\ref{appendix:convexenvelop} we show that 
	$$ \Gamma(R) = \min_{0\leq \kappa\leq 1} (1-\kappa)+\kappa\cdot \phi\left(2^{-2\frac{R}{\kappa}}\right)\,.$$

In order to show the first inequality, we will prove that
\begin{align}
 \DIP_n^*(R,P)\leq \min_{\kappa\in\frac{1}{n}\left\{0,1,\ldots,n\right\}} (1-\kappa)+\kappa\cdot \phi\left(D^*_{\kappa n}\left(\frac{R}{\kappa} ,P\right) \right)    
 \end{align}   
 from which the statement immediately follows.
Let $\kappa\in\frac{1}{n}\left\{0,1,\ldots,n\right\}$, and consider a compressor for $P^{\otimes \kappa n}$ under quadratic distortion: $f:\RR^{\kappa n}\to [2^{nR}=2^{n\kappa\frac{R}{\kappa}}]$ and  corresponding optimal decoder $g:[2^{nR}=2^{n\kappa\frac{R}{\kappa}}] \to\RR^{\kappa n}$, $g(w)=\hat{U}^{\kappa n}=\EE[U^{\kappa n}|f(U^{\kappa n})=w]$, that attains 
\begin{align}
D =\frac{1}{\kappa n}\EE\|U^{\kappa n}-\hat{U}^{\kappa n}\|^2=\frac{1}{\kappa n}\trace (\Sigma_{e_{U^{\kappa n}}}).
\end{align}
We encode $U$ by applying $f$ on $U^{\kappa n}$ and do not describe the other coordinates. The resulting covariance error matrix is therefore block diagonal of the form
\begin{align}
\Sigma_{e_U}=\left[\begin{array}{cc}
\Sigma_{e_{U^{\kappa n}}} & 0\\
0 & I_{(1-\kappa)n}
\end{array}\right].
\end{align}
Consequently,
\begin{align}
\trace(\Sigma_{e_U})&=	\trace (\Sigma_{e_{U^\kappa n}})+\trace( I_{(1-\kappa)n})=n\kappa D+n(1-\kappa)\\
\trace(\Sigma_{e_U}\Sigma_{e_U})&=\trace (\Sigma_{e_{U^\kappa n}}\Sigma_{e_{U^\kappa n}})+\trace( I_{(1-\kappa)n})=\|\Sigma_{e_{U^\kappa n}}\|_F^2+n(1-\kappa).\label{eq:traceTS}
\end{align}
Recall that for a positive semi-definite matrix $A\in\RR^{m\times m}$ it holds that $\|A\|_F^2\geq \frac{1}{m}(\trace(A))^2$. This follows since the vector $\lambda=\eig(A)$ has non-negative entries, so that $\trace(A)=\|\lambda\|_1$, and therefore $\|A\|_F^2=\|\lambda\|_2^2\geq \frac{1}{m}\|\lambda\|_1^2=\frac{1}{m}(\trace(A))^2$. Thus, 
\begin{align}
\|\Sigma_{e_{U^\kappa n}}\|_F^2\geq \frac{1}{\kappa n}(\trace (\Sigma_{e_{U^\kappa n}}))^2=\kappa n D^2,
\end{align}
and, by~\eqref{eq:traceTS}, we have
\begin{align}
\trace(\Sigma_{e_U}\Sigma_{e_U})&\geq n \kappa D^2 +n(1-\kappa).
\end{align}
We use the same encoder also for encoding $V$, such that $\Sigma_{e_V}=\Sigma_{e_U}$, and use the optimal decoder $g^*$ for estimating $U^\top  V$. Applying Proposition~\ref{prop:generalerror}, we obtain
\begin{align}
\DIP&=\frac{1}{n}\left[\trace(\Sigma_{e_U})+\trace(\Sigma_{e_V})-\trace(\Sigma_{e_U}\Sigma_{e_V})\right]\\
&=\frac{1}{n}\left[2\trace(\Sigma_{e_U})-\|\Sigma_{e_U}\|_F^2\right]\\
&\leq (1-\kappa)+\kappa\cdot (2D-D^2)\\
&=(1-\kappa)+\kappa \phi(D).
\end{align}
Taking the compressor $f$ that attains $D^*_{\kappa n}\left(\frac{R}{\kappa} ,P\right)$, we obtain the claimed result.
\end{proof}

\subsection{Lower Bound}\label{sec:iprd_lb}

\begin{lemma}
For the symmetric case, assuming that
$P$ has zero mean and unit variance,  there is no loss of optimality in taking $f_1=f_2=f$, and
\begin{align}
\DIP_n^*(R,P)=\frac{1}{n}\inf_f\left[ 2 \|\lambda(f)\|_1-\|\lambda(f)\|_2^2\right]=\frac{1}{n}\inf_f \sum_{i=1}^n \phi\left(\lambda_i(f)\right),    
\end{align}
where the infimum runs over all encoders $f:\RR^n\to [2^{nR}]$, and 
\begin{align}
\lambda(f)=\eig\left(\Sigma_{e^f_U} \right),\label{eq:eigfdef}
\end{align}
where $e^f_U=U-\EE[U|f(U)]$, $\Sigma_{e^f_U}=\EE[e^f_U e^{f,\top}_U]$.
\label{lem:symmetricisoptimal}
\end{lemma}

\begin{proof}
By Proposition~\ref{prop:generalerror}, we have that for any two encoders $f_1:\RR^n\to [2^{nR}]$ and $f_2:\RR^n\to [2^{nR}]$, when the optimal decoder is used, it holds that
\begin{align}
\DIP&=\frac{1}{n}\left[\trace(\Sigma_{e^{f_1}_U})+\trace(\Sigma_{e^{f_2}_V})-\trace(\Sigma_{e^{f_1}_U}\Sigma_{e^{f_2}_V}) \right]\\
&=\frac{1}{n}\left[\trace(\Sigma_{e^{f_1}_U})+\trace(\Sigma_{e^{f_2}_U})-\trace(\Sigma_{e^{f_1}_U}\Sigma_{e^{f_2}_U}) \right],\label{eq:Derrorf1f2}
\end{align}
where the last equality follows since $P=Q$, and therefore $U$ and $V$ have the same distribution. Rearranging~\eqref{eq:Derrorf1f2}, we obtain
\begin{align}
n(1-\DIP) &=
\trace((I_n-\Sigma_{e^{f_1}_U})(I_n-\Sigma_{e^{f_2}_U}))\\
&\leq \sqrt{\trace((I_n-\Sigma_{e^{f_1}_U})^2)\trace((I_n-\Sigma_{e^{f_2}_U})^2)}\label{eq:CSforTrace}\\
&\leq \max_{i\in\{1,2\}}\trace((I-\Sigma_{e^{f_i}_U})^2)
\end{align}
where~\eqref{eq:CSforTrace} follows since $(C,D) \mapsto \trace(CD)$ defines inner product on
symmetric matrices $C,D\in \RR^{n\times n}$, and therefore the Cauchy-Schwartz inequality holds. We have therefore obtained
\begin{align}
  \DIP&\geq  \frac{1}{n}   \min_{i\in\{1,2\}}(2\trace\Sigma_{e^{f_i}_U} - \trace \Sigma_{e^{f_i}_U}^2)\\
  &\geq \frac{1}{n}\inf_f\left[ 2 \|\lambda(f)\|_1-\|\lambda(f)\|_2^2\right].
\end{align}
Note that all inequalities in the derivation hold with equality if $f_1=f_2=f^*$, where $f^*$ attains infimum above (or is a sequence of functions approaching this infimum).
\end{proof}

The following Shannon-lower-bound-type lemma constrains the eigenvalues of an MSE matrix for estimating $U$ from a $2^{nR}$-level quantizer $f:\RR^n\to[2^{nR}]$.

\begin{lemma}
Assume $P$ has zero mean and unit variance. Let $f:\RR^n\to [2^{nR}]$ be a $2^{nR}$-level quantizer, and define $\lambda(f)=(\lambda_1,\ldots,\lambda_n)\in[0,1]^n$ as in~\eqref{eq:eigfdef}. Then
\begin{align}
\frac{1}{n}\sum_{i=1}^n \frac{1}{2}\log \frac{1}{\lambda_i}\leq R+D(P\|\m{N}(0,1)).
\end{align}
\label{lem:SLB}
\end{lemma}

\begin{proof}
We may assume without loss of generality that $h(P)>-\infty$, as otherwise $D(P\|\m{N}(0,1))=\infty$ and the statement trivially holds.
Let $e^f_U=U-\EE[U|f(U)]$, $\Sigma_{e^f_U}=\EE[e^f_U e^{f,\top}_U]$. Since the Gaussian distribution maximizes differential entropy under second moment constraints, we have that
\begin{align}
h(U|f(U))\leq \frac{1}{2}\log\det \left((2\pi e)\Sigma_{e^f_U}\right)=n\cdot \frac{1}{n}\sum_{i=1}^n \frac{1}{2}\log (2\pi e \lambda_i).
\end{align}
Consequently,
\begin{align}
nR&\geq I(U;f(U))=h(U)-h(U|f(U))\geq h(U)-n\cdot \frac{1}{n}\sum_{i=1}^n \frac{1}{2}\log (2\pi e \lambda_i)\\
&=h(\m{N}^{\otimes n}(0,1))-n\cdot \frac{1}{n}\sum_{i=1}^n \frac{1}{2}\log (2\pi e \lambda_i)+h(P^{\otimes n})-h(\m{N}^{\otimes n}(0,1))\\
&=n\left(\frac{1}{n}\sum_{i=1}^n \frac{1}{2}\log\frac{1}{\lambda_i} -D(P\|\m{N}(0,1))\right),
\end{align}
which yields the claimed result.
\end{proof}

\begin{theorem}
Assuming $P$ has zero mean and unit variance, for any $n\geq 1$
\begin{align}
\DIP^*_n(R,P)\geq \Gamma\left(R+D(P\|\m{N}(0,1))\right),
\end{align}

where $\Gamma(R)$ is defined in~\eqref{eq:barphidef}, and in particular
\begin{align}
\DIP(R,P)\geq \Gamma\left(R+D(P\|\m{N}(0,1))\right).
\end{align}
\label{thm:SLBbasedDistortionLB}
\end{theorem}

\begin{proof}
Let $f:\RR^n\to [2^{nR}]$ be a $2^{nR}$-level quantizer, and define $\lambda(f)=(\lambda_1,\ldots,\lambda_n)\in[0,1]^n$ as in~\eqref{eq:eigfdef}. Denote by $K=K_f$ the uniform distribution over (the multiset) $\lambda(f)$. By Lemma~\ref{lem:symmetricisoptimal}, we have that
\begin{align}
\DIP_n^*(R,P)=\inf_f \EE_{\lambda\sim K_f}\phi(\lambda)=\inf_f \EE_{\lambda\sim K_f}\phi\left(2^{-2 R_{\m{N}}(\lambda)}\right),
\end{align}
where $R_{\m{N}}(\lambda)=\frac{1}{2}\log\frac{1}{\lambda}$. Denote $\Gamma_1(R)=\phi(2^{-2R})$. In Appendix~\ref{appendix:convexenvelop} we show that
\begin{align}
\text{convex envelope of } \Gamma_1(R)=\Gamma(R).    
\end{align}
It therefore follows that
\begin{align}
\DIP_n^*(R,P)&=\inf_f \EE_{\lambda\sim K_f} \Gamma_1(R_{\m{N}}(\lambda))\\
&\geq \inf_f \EE_{\lambda\sim K_f} \Gamma\left(R_{\m{N}}(\lambda)\right)\\
&\geq \inf_f  \Gamma\left(\EE_{\lambda\sim K_f}R_{\m{N}}(\lambda)\right)\\
&\geq \Gamma\left(R+D(P\|\m{N}(0,1))\right),
\end{align}
where we have used Lemma~\ref{lem:SLB} in the last inequality.
\end{proof}

\subsection{The Symmetric Gaussian case}

Combining Theorem~\ref{thm:TSupperbound} and Theorem~\ref{thm:SLBbasedDistortionLB}, we obtain a complete characterization for the Gaussian case.

\begin{theorem}\label{thm:symgsn_iprd}
\begin{align}
\DIP(R,\m{N}(0,1))=\Gamma(R)=\begin{cases}
1-\left(1-\phi(2^{-2R^*})\right)\frac{R}{R^*} & R<R^*\\
\phi(2^{-2R}) & R\geq R^*
\end{cases}.
\end{align}
\end{theorem}

\begin{proof}
The upper bound follows from applying Theorem~\ref{thm:TSupperbound} with $\kappa=\min\{R/R^*,1\}$, and recalling that $D_{\m{N}(0,1)}(R)=2^{-2R}$. The lower bound follows directly from Theorem~\ref{thm:SLBbasedDistortionLB}.
\end{proof}

\section{Compression for Matrix Multiplication}
\label{sec:iidMM}

\subsection{Setup}

Let $A\in\RR^{n\times a}$ be a matrix whose entries are drawn iid from the distribution $P$ and $B\in\RR^{n\times b}$ be a matrix, statistically independent of $A$, whose entries are drawn iid from the distribution $Q$. We assume both $P$ and $Q$ are distributions with zero mean and unit variance. We consider the problem of quantizing $A$ and $B$ in order to compute their matrix multiplication $A^\top  B$. In particular, an $(n,a,b,R_1,R_2,D)$ code consists of mappings
\begin{align}
 f_1&: \RR^{n\times a}\to[2^{n a R_1}]\\
 f_2&: \RR^{n\times b}\to[2^{n b R_2}]\\
 g&:[2^{n a R_1}]\times [2^{n b R_2}]\to \RR^{a\times b},
\end{align}
with
\begin{align}
D=\DMM=\frac{1}{n\cdot a\cdot b}\EE\|A^\top  B-g(f_1(A),f_2(B))\|_F^2.    
\end{align}
We define 
\begin{align}
\DMM^*_{n,a,b}(R_1,R_2)=\DMM^*_{n,a,b}(R_1,R_2,P,Q)=\inf\left\{D~:~\exists (n,a,b,R_1,R_2,D)-\text{code} \right\}. 
\end{align}
We further define the asymptotic function
\begin{align}
\DMM_{a,b}(R_1,R_2)=\DMM_{a,b}(R_1,R_2,P,Q)=\limsup_{n\to\infty} D^*_{n,a,b}(R_1,R_2),  
\end{align}

\subsection{Basic Properties and Bounds}

Denote $W_A=f_1(A)$ and $W_B=f_2(B)$ and further denote $\hat{A}=\EE[A|W_A]$ and $\hat{B}=\EE[B|W_B]$. Define $\Sigma_A=\EE[(A-\hat A) (A-\hat A)^\top ] \in \mreals^{n\times n}$ and $\bar M_A = \EE[\hat A \hat A^\top ] \in \mreals^{n\times n}$. Similarly, $\Sigma_B=\EE[(B-\hat B) (B-\hat B)^\top ] \in \mreals^{n\times n}$ and $\bar M_B = \EE[\hat B \hat B^\top ] \in \mreals^{n\times n}$.
As in the scalar case, we still have the identities:
\begin{align}
 \Sigma_A + \bar M_A &= a  I_n\\
 \Sigma_B + \bar M_B &= b  I_n.
\end{align}
The next theorem generalizes the basic bounds we derived above for the inner product case, to the matrix multiplication case. The proofs are similar to the statements above, and are therefore omitted.
\begin{theorem}
Assume $P$ and $Q$ have zero mean and unit variance. The following hold:
\begin{enumerate}
\item For fixed $f_1,f_2$, the optimal choice for $g$ is $g^*(W_A,W_B)=\hat{A}^\top \hat{B}$, and the distortion is given by
\begin{align*}
   \DMM &=\frac{1}{n\cdot a\cdot b}\left[ \trace(\Sigma_A \bar M_B) + \trace(\Sigma_B \bar M_A) + \trace(\Sigma_A \Sigma_B)\right]\\
   &= \frac{1}{n}\left[\frac{1}{a} \trace(\Sigma_A) + \frac{1}{b} \trace(\Sigma_B) - \frac{1}{a\cdot b}\trace(\Sigma_A \Sigma_B)\right].
\end{align*}
\item The oracle lower bound (taking $\hat{B}=B$ or $\hat{A}=A)$ gives
\begin{align*}
\DMM\geq \max\left\{\frac{1}{n\cdot a}\trace \Sigma_A,\frac{1}{n\cdot b}\trace \Sigma_B  \right\},  
\end{align*}
and consequently for any $n\geq 1$
\begin{align*}
\DMM^*_{n,a,b}(R_1,R_2,P,Q)\geq  \max\left\{D_{P}(R_1),D_{Q}(R_2) \right\},
\end{align*}
and in particular
\begin{align*}
\DMM_{a,b}(R_1,R_2,P,Q)\geq  \max\left\{D_{P}(R_1),D_{Q}(R_2) \right\}.
\end{align*}
\item \label{item:matroxTS} For the symmetric case, where  $R_1=R_2=R$ and $P=Q$, we have 
 \begin{align*}
 \DMM_{a,b}(R,P)\leq \min_{0\leq \kappa\leq 1} (1-\kappa)+\kappa\cdot \phi\left(D_{P}\left(\frac{R}{\kappa}\right) \right)    
 \end{align*}   
This is asymptotically attained by quantizing only the first $\kappa n$ coordinates of each column of $A$ and each column of $B$.
\item For the symmetric case, where  $R_1=R_2=R$ and $P=Q$, for any $n\geq 1$ we have
\begin{align}
\DMM_{n,a,b}^*(R,P)&\geq\frac{1}{n}\min\left\{\inf_{f_a}\left[ 2 \|\lambda(f_a)\|_1-\|\lambda(f_a)\|_2^2\right],\inf_{f_b}\left[ 2 \|\lambda(f_b)\|_1-\|\lambda(f_b)\|_2^2\right]\right\}\nonumber\\
&=\frac{1}{n}\min\left\{\inf_{f_a} \sum_{i=1}^n \phi\left(\lambda_i(f_a)\right),\inf_{f_b} \sum_{i=1}^n \phi\left(\lambda_i(f_b)\right)\right\}, 
\label{eq:matrixSymmetricGeneralLB}
\end{align}
where the infima runs over all encoders $f_a:\RR^{n\times a}\to [2^{na R}]$, $f_b:\RR^{n\times b}\to [2^{n bR}]$, and 
\begin{align}
\lambda(f_a)=\eig\left(\frac{1}{a}\Sigma_{e^{f_a}_A} \right),~~~\lambda(f_b)=\eig\left(\frac{1}{b}\Sigma_{e^{f_b}_B} \right)
\label{eq:eigendefmat}
\end{align}
where $e^{f_a}_A=A-\EE[A|f_a(A)]$, $\Sigma_{e^{f_a}_A}=\EE[e^{f_a}_A e^{f_a,\top}_A]$, and $e^{f_b}_B=B-\EE[B|f_b(B)]$, $\Sigma_{e^{f_b}_B}=\EE[e^{f_b}_B e^{f_b,\top}_B]$.
\end{enumerate}
\label{thm:matrixerrorexpressions}
\end{theorem}

\subsection{Maximum Entropy Matrices}

The fact that the Gaussian distribution maximizes the differential entropy of a vector, under second moment constraints, played a pivotal role in the derivation of our bounds for inner product quantization. For matrix multiplication quantization, the following lemma will play a similar role.
\begin{lemma}
    Let $M\in \mreals^{n\times a}$ be a random matrix with $\EE[M]=0$, and $\EE[M M^\top ] = \Sigma$. Then
    \begin{align} 
    h(M) \leq \frac{a}{2}\log\det \left(2\pi e\frac{1}{a}\Sigma\right)
    =a\cdot\sum_{i=1}^n\frac{1}{2}\log(2\pi e\lambda_i),
    \end{align}
    where $\lambda=\eig\left(\frac{1}{a}\Sigma\right)$, and this is attained with equality if the columns of $M$ are independent $\m{N}\left(0,\frac{1}{a}\Sigma\right)$ random vectors.
    \label{lem:matmaxent}
\end{lemma}

\begin{proof}
Write $M=[m_1|m_2|\cdots|m_a]$, where $m_1,\ldots,m_a$ are zero-mean random vectors in $\RR^n$. Denote the marginal distribution of $m_i$ by $P_i$. Let $\Sigma_i=\EE[m_i m_i^\top ]$, and recall that 
\begin{align}
\Sigma=\EE[M M^\top ]=\sum_{i=1}^a \EE[m_i m_i^\top ]=\sum_{i=1}^a \Sigma_i.
\end{align}
We further have that
\begin{align}
h(M)=h(m_1,\ldots,m_a)\leq \sum_{i=1}^a h(m_i)\leq a\cdot  h\left(\frac{1}{a}\sum P_i \right),   
\end{align}
where we have used sub-additivity and concavity of differential entropy in the inequalities above. Noting that the covariance matrix corresponding to the distribution $\frac{1}{a}\sum_{i=1}^a P_i$ is $\frac{1}{a}\sum_{i=1}^a \Sigma_i=\frac{1}{a}\Sigma$, we have
\begin{align}
h(M)\leq a\cdot h\left(\m{N}\left(0,\frac{1}{a}\Sigma\right)\right) =\frac{a}{2}\log\det \left(2\pi e\frac{1}{a}\Sigma\right).
\end{align}
All inequalities are attained with equality when $m_i\stackrel{iid}{\sim}\m{N}\left(0,\frac{1}{a}\Sigma\right)$, for $i=1,\ldots,a$.
\end{proof}

This immediately gives the following generalization of Lemma~\ref{lem:SLB}
\begin{lemma}
Assume the distribution $P$ has zero mean and unit variance. Let $f_a:\RR^{n\times a}\to [2^{naR}]$ be a $2^{naR}$-level quantizer, and define $\lambda(f_a)=(\lambda_1,\ldots,\lambda_n)\in[0,1]^n$ as in~\eqref{eq:eigendefmat}. Then
\begin{align}
\frac{1}{n}\sum_{i=1}^n \frac{1}{2}\log \frac{1}{\lambda_i}\leq R+D(P\|\m{N}(0,1)).
\end{align}
\label{lem:SLBmatrix}
\end{lemma}

\begin{proof}
Without loss of generality, we may assume $h(P)>-\infty$, as otherwise $D(P\|\m{N}(0,1))=\infty$ and the statement is trivial.
Let $e^{f_a}_A=A-\EE[A|f_a(A)]$, $\Sigma_{e^{f_a}_A}=\EE[e^{f_a}_A e^{f_a,\top}_A]$. By Lemma~\ref{lem:matmaxent}, we have that
\begin{align}
h(A|f_a(A))\leq \frac{a}{2}\log\det \left((2\pi e)\frac{1}{a}\Sigma_{e^{f_a}_A}\right)=na\cdot \frac{1}{n}\sum_{i=1}^n \frac{1}{2}\log (2\pi e \lambda_i).
\end{align}
Consequently,
\begin{align}
naR&\geq I(A;f_a(A))=h(A)-h(A|f_a(A))\geq h(A)-na\cdot \frac{1}{n}\sum_{i=1}^n \frac{1}{2}\log (2\pi e \lambda_i)\\
&=h(\m{N}^{\otimes na}(0,1))-na\cdot \frac{1}{n}\sum_{i=1}^n \frac{1}{2}\log (2\pi e \lambda_i)+h(P^{\otimes na})-h(\m{N}^{\otimes na}(0,1))\\
&=na\left(\frac{1}{n}\sum_{i=1}^n \frac{1}{2}\log\frac{1}{\lambda_i} -D(P\|\m{N}(0,1))\right),
\end{align}
which yields the claimed result.
\end{proof}

\subsection{Fundamental Limits}

Using Theorem~\ref{thm:matrixerrorexpressions} and Lemma~\ref{lem:SLBmatrix}, we prove the following result for the symmetric matrix multiplication case.
\begin{theorem}
Assuming the distribution $P$ has zero mean and unit variance, for any $n\geq 1$
\begin{align}
\DMM_{n,a,b}^*(R,P)\geq \Gamma\left(R+D(P\|\m{N}(0,1))\right),
\end{align}
where $\Gamma(R)$ is defined in~\eqref{eq:barphidef}, and in particular
\begin{align}
\DMM_{a,b}(R,P)\geq \Gamma\left(R+D(P\|\m{N}(0,1))\right).
\end{align}
\label{thm:SLBbasedDistortionLBmatrix}
\end{theorem}

\begin{proof}
Let $f_a:\RR^{n\times a}\to [2^{na R}]$ be a $2^{naR}$-level quantizer, and define $\lambda(f_a)=(\lambda_1,\ldots,\lambda_n)\in[0,1]^n$ as in~\eqref{eq:eigendefmat}. Denote by $K=K_{f_a}$ the uniform distribution over (the multiset) $\lambda(f_a)$, and $R_{\m{N}}(\lambda)=\frac{1}{2}\log\frac{1}{\lambda}$, as in the proof of Theorem~\ref{thm:SLBbasedDistortionLB}, we have that
\begin{align}
\EE_{\lambda\sim K_{f_a}}\phi(\lambda)=\EE_{\lambda\sim \Gamma_{f_a}}\phi\left(2^{-2 R_{\m{N}}(\lambda)}\right).
\end{align}
Recalling from the proof of Theorem~\ref{thm:SLBbasedDistortionLB} that the $\Gamma_1(R)=\phi(2^{-2R})\geq \Gamma(R)$ and the function $R\mapsto \Gamma(R)$ is convex and non-increasing, it follows that
\begin{align}
\EE_{\lambda\sim K_{f_a}}\phi\left(2^{-2 R_{\m{N}}(\lambda)}\right)&\geq \EE_{\lambda\sim K_f} \Gamma\left(R_{\m{N}}(\lambda)\right)\\
&\geq  \Gamma\left(\EE_{\lambda\sim K_f}R_{\m{N}}(\lambda)\right)\\
&\geq \Gamma\left(R+D(P\|\m{N}(0,1))\right),
\end{align}
where we have used Lemma~\ref{lem:SLBmatrix} in the last inequality. Thus,
\begin{align}
\frac{1}{n}\sum_{i=1}^n\phi(\lambda_i(f_a))=\EE_{\lambda\sim K_{f_a}}\phi\left(\lambda\right)\geq \Gamma\left(R+D(P\|\m{N}(0,1))\right).
\end{align}
Similarly, for any $f_b:\RR^{n\times b}\to\left[2^{nbR} \right]$ we have
\begin{align}
\frac{1}{n}\sum_{i=1}^n\phi(\lambda_i(f_b))=\EE_{\lambda\sim K_{f_b}}\phi\left(\lambda\right)\geq \Gamma\left(R+D(P\|\m{N}(0,1))\right).
\end{align}
Thus, by~\eqref{eq:matrixSymmetricGeneralLB} in Theorem~\ref{thm:matrixerrorexpressions}, for any $n\geq 1$
\begin{align}
\DMM_{n,a,b}(R,P)&\geq \min\left\{\min_{f_a}\frac{1}{n}\sum_{i=1}^n\phi(\lambda_i(f_a)),\min_{f_b}\frac{1}{n}\sum_{i=1}^n\phi(\lambda_i(f_b)) \right\}\\
&\geq   \Gamma\left(R+D(P\|\m{N}(0,1))\right),  
\end{align}
as claimed.
\end{proof}

\begin{proof}[Proof of Theorem~\ref{thm:GaussMatrixLowerBound}]
Part 1 follows immediately from Theorem~\ref{thm:SLBbasedDistortionLBmatrix}. Part 2 follows from part 2 of Theorem~\ref{thm:matrixerrorexpressions}, and recalling that $D_{\m{N}(0,1)}(R)=2^{-2R}$.  
\end{proof}

\subsection{The Symmetric Gaussian case}

Combining Theorem~\ref{thm:matrixerrorexpressions} and Theorem~\ref{thm:SLBbasedDistortionLBmatrix}, we obtain a complete characterization for the Gaussian case.

\begin{theorem}
\label{thm:GaussDMM}
\begin{align}
\DMM_{a,b}(R,\m{N}(0,1))=\Gamma(R).
\end{align}
\end{theorem}

\begin{proof}
The upper bound follows applying Part~\ref{item:matroxTS} of Theorem~\ref{thm:matrixerrorexpressions} with $\kappa=\min\{R/R^*,1\}$, and recalling that $D_{\m{N}(0,1)}(R)=2^{-2R}$. The lower bound follows directly from Theorem~\ref{thm:SLBbasedDistortionLBmatrix}.
\end{proof}

\section{Lattice Quantization Scheme for Matrix Multiplication of Arbitrary Matrices}
\label{sec:arbitraryMM}

Our theoretical analysis in Sections~\ref{sec:iidIPdefsandBounds}
-\ref{sec:iidMM} assumed the entries in the vectors/matrices to be multiplied are drawn iid from some known distribution. In this section, we drop this assumption, and, building on the observations from the analysis above, develop a robust scheme for compression for matrix multiplication. Our scheme is designed to attain the optimal distortion in the case where $A$ and $B$ have iid Gaussian entries, but the error it attains for arbitrary matrices can also be upper bounded.

We first develop encoders $f_1,f_2:\RR^n\to[2^{nR}]$ and a decoder $g:[2^{nR}]\times [2^{nR}]\to
\RR$ for estimating the inner product of $U,V\in \sqrt{n}\mathbb{S}^{n-1}$ where
$\mathbb{S}^{n-1}=\left\{x\in\RR^n~:~\|x\|=1 \right\}$ is the unit sphere. We then show how these
encoders and decoder can be leveraged for compression for matrix multiplication. Let
$\mathrm{O}_n(\RR)$ be the orthogonal group, consisting of all orthonormal matrices in
$\RR^{n\times n}$. It will be useful to analyze the performance of $f_1,f_2,g$ with respect to the following distribution on $U,V$. 
\begin{definition}[$\rho$-correlated spherically uniform random vectors]
Let $S=[S_1|S_2|\cdots|S_n]\sim\Unif(\mathrm{O}_n(\RR))$ be a random matrix uniformly distributed over the group of orthogonal matrices in $\RR^{n\times n}$ (that is, $S$ is drawn from the Haar measure on $\mathrm{O}_n(\RR)$). We say that the random vectors $U\in\RR^n$ and $V\in\RR^n$ are $\rho$-correlated spherically uniform random vectors if $U=\sqrt{n}S_1$, $Z=\sqrt{n} S_2$ and
\begin{align}
V=\rho U+\sqrt{1-\rho^2}Z.
\label{eq:VUdef_jointGauss}
\end{align}
\end{definition}

\begin{theorem}
For any $\eps>0$, $0<\kappa\leq 1$ and sufficiently large $n$,  there exist randomized encoders $f_1,f_2:\RR^n\to[2^{nR}]$ and decoders $g:[2^{nR}]\times [2^{nR}]\to \RR$, and $g_{1-\mathrm{sided}}:[2^{nR}]\times \RR^n \to \RR$, such that if $U,V$ are $\rho$-correlated spherically uniform
\begin{enumerate}
\item  for every $-1\leq \rho\leq 1$ and $0\leq \alpha\leq 1$
\begin{align}
\frac{1}{n}\EE(U^\top  V-\alpha g(f_1(U),f_2(V))^2< 
\rho^2 n\left(1-\kappa\alpha \right)^2+\kappa\alpha^2\left(\frac{1-\kappa+\kappa\phi\left(2^{-2\frac{R}{\kappa}}\right)}{1-\phi\left(2^{-2\frac{R}{\kappa}}\right)}\right)+\eps(1+\rho^2 n),  
\label{eq:nestedlatticespherical}
\end{align}
where $\phi(t)=2t-t^2$.   
\item for every $-1\leq \rho\leq 1$ and $0\leq \alpha\leq 1$
\begin{align}
\frac{1}{n}\EE(U^\top  V-\alpha g_{1-\mathrm{sided}}(f_1(U),V))^2< \rho^2 n\left(1-\kappa\alpha \right)^2+\kappa \alpha^2\left(1-\kappa+\frac{1}{2^{2\frac{R}{\kappa}}-1}\right)+\eps(1+\rho^2 n).
\label{eq:nestedlatticesphericalOneSided}
\end{align}
\end{enumerate}
\label{thm:universalLattice}
\end{theorem}

The proof is based on dithered nested lattice quantization, and is brought in Section~\ref{sec:nestedlatticeIP}.

\begin{remark}
The randomization required by the encoders and decoders above is in the form of \emph{dithering}, as will become clear in the proof of Theorem~\ref{thm:universalLattice}. For the special case of $\kappa=\alpha=1$, the MSE does not involve $\rho$, and therefore, there must exist fixed (deterministic) values for the dither vectors which attain the same MSE as above, or smaller. We believe that there also exist fixed values for the dithers, for which Theorem~\ref{thm:universalLattice} holds, and that randomness is not required at all here. The technical challenge is that when $\kappa,\alpha\neq 1$ the MSE is bounded as a weighted sum of expectations, where the weights depend on $\rho$ and $\alpha$. Thus, showing the existence of ``good'' fixed dithers, requires showing that there are dither values for which all involved expectations are small. This requires establishing some (very weak form of) of concentration of involved random variables, and is left for future work. Note that the assumptions of Theorem~\ref{thm:universalLattice} are that $U,V$ are random vectors on the sphere. This assumption is enforced in the next Theorem by applying a random rotation on arbitrary matrices. While random dithering is likely to not be needed, we do believe that the random rotation is necessary for Theorem~\ref{thm:MostgeneralMatMul} to hold. 
\end{remark}

\begin{remark}
Our proof of Theorem~\ref{thm:universalLattice} is based on using dithered ``good'' nested lattice quantizers. With these quantizers, the quantization noise is statistically independent of the input and is approximately white, a property that is crucial for the analysis. Using unstructured random coding techniques for proving Theorem~\ref{thm:universalLattice} does not seem straightforward. The results from~\cite{kipnis2021gaussian} may be a good starting point for pursuing this direction.
\end{remark}

\medskip

Equipped with Theorem~\ref{thm:universalLattice}, we can now easily prove Theorem~\ref{thm:MostgeneralMatMul}. Recall that for a column vector $x\in\RR^n$ we denote by $\bar{x}=x-\left(\frac{1}{n}\bones^\top x\right)\bones$ its zero-centered version. For a matrix $A=[a_1|\cdots|a_a]\in\RR^{n\times a}$ we denote $\bar{A}=[\bar{a}_1|\ldots|\bar{a}_a]$.

\begin{theorem}
For any $\eps>0$, $0<\kappa\leq 1$ and sufficiently large $n$, there exist randomized encoders $f_1:\RR^{n\times a}\to[2^{na R}]$, $f_2:\RR^{n\times b}\to[2^{nbR}]$, and decoders $g:[0,1]\times [2^{na R}]\times [2^{nbR}]\to \RR^{a\times b}$ and $g_{1-\mathrm{sided}}:[0,1]\times[2^{na R}]\times \RR^{n\times b}\to \RR^{a\times b}$ such that for any $A\in\RR^{n\times a}$ and $B\in\RR^{n\times b}$ with bounded entries we have
\begin{enumerate}
\item Let $C=A^\top B$, $\tilde{C}=\bar{A}^\top \bar{B}$, and $\hat{C}= g(\alpha,f_1(A),f_2(B))$ for  $0<\alpha\leq 1$. Then, for any $i\in[a],j\in[b]$ we have
\begin{align}
\EE(C_{i,j}-\hat{C}_{i,j})^2\leq  \tilde{C}_{i,j}^2\cdot\left( (1-\kappa\alpha)^2+\eps\right)+\frac{\|\bar{a}_i\|^2 \|\bar{b}_j\|^2}{n}\left(\kappa\alpha^2\left(\frac{1-\kappa+\kappa\phi\left(2^{-2\frac{R}{\kappa}}\right)}{1-\phi\left(2^{-2\frac{R}{\kappa}}\right)}\right)+\eps\right)+n^{-8},   
\end{align}
and in particular
\begin{align}
\EE\|A^\top  B- g(\alpha,f_1(A),f_2(B))\|_F^2&< \|\bar{A}^\top \bar{B}\|_F^2\cdot\left( (1-\kappa\alpha)^2+\eps\right)\nonumber\\
&+\frac{\|\bar{A}\|^2_F \|\bar{B}\|_F^2}{n}\left(\kappa\alpha^2\left(\frac{1-\kappa+\kappa\phi\left(2^{-2\frac{R}{\kappa}}\right)}{1-\phi\left(2^{-2\frac{R}{\kappa}}\right)}\right)+\eps\right)+a\cdot b\cdot n^{-8}.
\end{align}
\item  Let $C=A^\top B$, $\tilde{C}=\bar{A}^\top \bar{B}$, and $\hat{C}= g_{1-\mathrm{sided}}(\alpha,f_1(A),B)$ for $0<\alpha\leq 1$. Then, for any $i\in[a],j\in[b]$ we have
\begin{align}
\EE(C_{i,j}-\hat{C}_{i,j})^2\leq  \tilde{C}_{i,j}^2\cdot \left( (1-\kappa\alpha)^2+\eps\right)+\frac{\|\bar{a}_i\|^2 \|\bar{b}_j\|^2}{n}\left(\kappa \alpha^2\left(1-\kappa+\frac{1}{2^{2\frac{R}{\kappa}}-1}\right)+\eps\right)+n^{-8}.   
\end{align}
and in particular
\begin{align}
\EE\|A^\top  B- g_{1-\mathrm{sided}}(\alpha,f_1(A),B)\|_F^2&< \|\bar{A}^\top \bar{B}\|_F^2\cdot \left( (1-\kappa\alpha)^2+\eps\right)\nonumber\\
&+\frac{\|\bar{A}\|^2_F \|\bar{B}\|_F^2}{n}\left(\kappa \alpha^2\left(1-\kappa+\frac{1}{2^{2\frac{R}{\kappa}}-1}\right)+\eps\right)+a\cdot b\cdot n^{-8}.
\end{align}
\end{enumerate}
\label{thm:MostgeneralMatMul}
\end{theorem}

\begin{proof}[Proof of Theorem~\ref{thm:MostgeneralMatMul}]
We only prove part 1. The proofs for part 2 is nearly identical, and we specify the required modifications in the end of the proof.

Recall that $M=n^{10} 2^{2000}$ and let $\delta=M^{-5}$.
Let $f_1,f_2,g$ be the encoders and decoder from  Theorem~\ref{thm:universalLattice}. Based on those $f_1,f_2,g$, we propose the following rate-$R$ quantization scheme for quantization of matrices $A\in\RR^{n\times a}$ and $B\in\RR^{n\times b}$ in order to estimate $C=A^\top  B$:
\begin{enumerate}
\item Let $\mu_{a_i}=\frac{1}{n}\bones^\top a_i$ for $i\in[a]$ (similarly, $\mu_{b_j}=\frac{1}{n}\bones^\top b_j$ for $j\in[b]$). Since the matrices $A$ and $B$ have bounded entries, we have that $\mu_{a_i}\in [-M,M]$, $\forall i\in[a]$, and similarly for $\mu_{b_j}$, $j\in[b]$.
For each $i\in[a]$ we quantize $\mu_{a_i}$ to the nearest point in $\left\{k\cdot 2\delta\right\}_{k=-M/(2\delta)}^{k=M/(2\delta)}$, such that the quantized value $\hat{\mu}_{a_i}$ satisfies $\hat{\mu}_{a_i}=\mu_{a_i}\pm\delta$. This requires a total of $a\log (M/\delta)$ bits. Similarly, we quantize $\mu_{b_j}$ to $\hat{\mu}_{b_j}$ for each $j\in[b]$, which requires a total of $b\log (M/\delta)$ bits.
\item Let $\bar{a}_i=a_i-\mu_{a_i}\bones$ for $i\in[a]$ (similarly, $\bar{b}_j=b_j-\mu_{b_j}\bones$ for $j\in[b]$). Since the matrices $A$ and $B$ have bounded entries, we have that $\|\bar{a}_i\|\leq \|a_i\|\leq \sqrt{n}M$, $\forall i\in[a]$, and similarly for $\|\bar{b}_j\|$, $j\in[b]$.
We quantize the each $\|\bar{a}_i\|$, $i\in[a]$, to the nearest point in the grid $\{0\}\cup \{M^{-4}(1+\delta)^k\}_{k=0}^{T}$, where $T=\frac{\lceil \log (\sqrt{n}M^5) \rceil}{\log(1+\delta)}$. This requires a total of $a\log (T+2)<a\left(\log(3+\log(\sqrt{n}M^5))-\log \log(1+\delta) \right)$ bits.
Note that if $\|\bar{a}_i\|\in \{0\}\cup [M^{-4},\sqrt{n}M]$ we have that $\widehat{\|\bar{a}_i\|}=\|\bar{a}_i\|(1\pm \delta)$, and if $0<\|\bar{a}_i\|<M^{-4}$, then $|\|\bar{a}_i\|-\widehat{\|\bar{a}_i\|}|<M^{-4}$. We quantize each $\|\bar{b}_j\|$, $j\in[b]$ to $\widehat{\|\bar{b}_j\|}$, in a similar manner, requiring a total of $b\log (T+2)<b\left(\log(3+\log(\sqrt{n}M^5))-\log \log(1+\delta) \right)$ bits.
\item Draw $S\sim\Unif(\mathrm{O}_n(\RR))$ at both encoders (using common randomness), and compute $\tilde{A}=[\tilde{a}_1|\cdots|\tilde{a}_a]=S \bar{A}$, and $\tilde{B}=[\tilde{b}_1|\cdots|\tilde{b}_b]=S \bar{B}$, where $\bar{A}=[\bar{a}_1|\cdots|\bar{a}_a]$ and $\bar{B}=[\bar{b}_1|\cdots|\bar{b}_b]$.
\item Let
\begin{align}
 U_i&=\sqrt{n}\frac{\tilde{a}_i}{\|\bar{a}_i\|}=\sqrt{n} S \frac{\bar{a}_i}{\|\bar{a}_i\|},~~ i=1,\ldots,a\label{eq:RandRotA}\\
 V_j&=\sqrt{n}\frac{\tilde{b}_j}{\|\bar{b}_i\|}=\sqrt{n} S \frac{\bar{b}_j}{\|\bar{b}_j\|},~~ j=1,\ldots,b\label{eq:RandRotB}.
\end{align}
Let
\begin{align}
\eps_0=\frac{1}{n}\left[\log(M/\delta)+ \log(3+\log(\sqrt{n}M^5))-\log \log(1+\delta)\right]
\end{align}
and note that $\eps_0$ can be made arbitrarily small for $n$ large enough. 
Apply $f_1:\RR^n\to [2^{n(R-\eps_0)}]$ on $U_i$, for $i=1,\ldots,a$, and $f_2:\RR^n\to [2^{n(R-\eps_0)}]$ on $V_j$, for $j=1,\ldots,b$.
\item Use $g:[2^{n(R-\eps_0)}]\times [2^{n(R-\eps_0)}]\to\RR$, to estimate each entry of $C=A^\top  B$ as
\begin{align}
\hat{C}_{ij}=\alpha\frac{\widehat{\|\bar{a}_i\|}~\widehat{\|\bar{b}_j\|}}{n}g(f_1(U_i),f_2(V_j))+n\hat{\mu}_{a_i}\hat{\mu}_{b_j},~~~ i=1,\ldots,a,~~j=1,\ldots,b.    
\end{align}
\end{enumerate}
To analyze the mean squared error $\EE(C_{ij}-\hat{C}_{i,j})^2$, first note that
\begin{align}
\bar{a}_i^\top \bar{b}_j=\left(a_i-\mu_{a_i}\bones \right)^\top\left(b_j-\mu_{b_j}\bones \right)=a_i^\top b_j-n\mu_{a_i}\mu_{b_j},
\end{align}
so that
\begin{align}
C_{ij}&=a_i^\top b_j=\bar{a}_i^\top  \bar{b}_j+n\mu_{a_i}\mu_{b_j}=\bar{a}^\top _i S^\top  S \bar{b}_j+n\mu_{a_i}\mu_{b_j}=\tilde{a}_i^\top \tilde{b}_j+n\mu_{a_i}\mu_{b_j}\nonumber\\
&=\frac{\|\bar{a}_i\|~\|\bar{b}_j\|}{n} U_i^\top  V_j+n\mu_{a_i}\mu_{b_j}.
\end{align}
We therefore have that
\begin{align}
C_{ij}-\hat{C}_{ij}&=\frac{\|\bar{a}_i\|~\|\bar{b}_j\|}{n} U_i^\top  V_j-\alpha\frac{\widehat{\|\bar{a}_i\|}~\widehat{\|\bar{b}_j\|}}{n}g(f_1(U_i),f_2(V_j))+ n\mu_{a_i}\mu_{b_j}-n\hat{\mu}_{a_i}\hat{\mu}_{b_j}\\
&=e_{ij}+\Delta,
\end{align}
where 
\begin{align}
e_{ij}= \frac{\|\bar{a}_i\|~\|\bar{b}_j\|}{n}\left(U_i^\top  V_j-\alpha g(f_1(U_i),f_2(V_j)) \right)   
\end{align}
and $\Delta=\Delta_1+\Delta_2$, where
\begin{align}
\Delta_1=n\mu_{a_i}\mu_{b_j}-n\hat{\mu}_{a_i}\hat{\mu}_{b_j}
\end{align}
and
\begin{align}
\Delta_2=\frac{\alpha}{n}g(f_1(U_i),f_2(V_j))\left(\widehat{\|\bar{a}_i\|}~\widehat{\|\bar{b}_j\|}-\|\bar{a}_i\|\|\bar{b_j}\| \right).    
\end{align}
We have that
\begin{align}
|\Delta_1|\leq n\delta(|\mu_a|+|\mu_b|)+n\delta^2\leq 3n M^{-4}.  
\end{align}
To upper bound $|\Delta_2|$, first note that without loss of generality we can assume $|\alpha g(f_1(U_i),f_2(V_j)|\leq n$ because the quantity $U_i^\top V_j$ it estimates is in $[-n,n]$. Furthermore, 
\begin{align}
\left|\widehat{\|\bar{a}_i\|}~\widehat{\|\bar{b}_j\|}-\|\bar{a}_i\|\|\bar{b_j}\| \right|\leq 
\begin{cases}
\|\bar{a}_i\|\cdot \|\bar{b}_j\|\cdot 3\delta    & \|\bar{a}_i\|,\|\bar{b}_j\|\in \{0\}\cup [M^{-4},\sqrt{n}M] \\
M^{-8} &  \|\bar{a}_i\|,\|\bar{b}_j\|<M^{-4}\\
2\sqrt{n}M^{-3} & \text{otherwise}
\end{cases}.
\end{align}
Thus (for $n\geq 4)$,
\begin{align}
|\Delta_2|\leq  3M^{-5}\|\bar{a}_i\|\cdot \|\bar{b}_j\|+2\sqrt{n}M^{-3}\leq 4n M^{-3}.   
\end{align}
We consequently have that 
\begin{align}
|\Delta|<\eps_1=7 n M^{-3},    
\end{align}
with probability $1$. We have therefore obtained
\begin{align}
\EE(C_{ij}-\hat{C}_{ij})^2\leq \EE(e_{ij}^2)+\eps_1^2+2\eps_1\EE|e_{ij}|\leq \EE(e_{ij}^2)+\eps_1^2+4\eps_1 \|\bar{a}_i\|\cdot \|\bar{b}_j\|,
\end{align}
where in the last inequality we have used the fact that both $|U_i^\top V_j|\leq n$ and $|\alpha g(f_1(U_i),f_2(V_j))|\leq n$.

We are therefore left with the task of upper bounding $\EE(e_{ij}^2)$. To that end, let $\rho_{ij}=\frac{\bar{a}_i^\top  \bar{b}_j}{\|\bar{a}_i\|~\|\bar{b}_j\|}$. We claim that $U_i,V_j$ are $\rho_{ij}$-correlated spherically uniform random vectors.
To see this, note that due to the random rotation matrix $S$, we may assume without loss of generality that 
\begin{align}
\frac{\bar{a}_i}{\|\bar{a}_i\|}&=[1|0|0|\cdots|0]^\top ,\\
\frac{\bar{b}_j}{\|\bar{b}_j\|}&=[\rho_{ij}|\sqrt{1-\rho_{ij}^2}|0|\cdots|0]^\top ,
\end{align}
and this assumption will have no affect on the joint distribution of $U_i,V_j$. Writing $S=[S_1|S_2|\cdots|S_n]$, we therefore have that $U_i=\sqrt{n}S_1$ and $V_j=\rho_{ij} U_i+\sqrt{1-\rho_{ij}^2}Z$, with $Z=\sqrt{n}S_2$. Thus, if $f_1,f_2,g$ are the encoders and decoder from Theorem~\ref{thm:universalLattice}, we therefore have that for any $\eps'>0$ and $n$ large enough
\begin{align}
\frac{1}{n}\EE(U_i^\top  V_j-\alpha g(f_1(U_i),f_2(V_j)))^2<\rho_{ij}^2 n\left(1-\kappa\alpha \right)^2+\kappa\alpha^2\left(\frac{1-\kappa+\kappa\phi\left(2^{-2\frac{(R-\eps_0)}{\kappa}}\right)}{1-\phi\left(2^{-2\frac{(R-\eps_0)}{\kappa}}\right)}\right)+\eps'(1+\rho_{ij}^2 n).
\end{align}
Consequently,
\begin{align}
\EE(e^2_{ij})&=\EE\left(\frac{\|\bar{a}_i\|~\|\bar{b}_j\|}{n}\left( U_i^\top  V_j-\alpha g(f_1(U_i),f_2(V_j))\right) \right)^2\nonumber\\
&<\frac{\|\bar{a}_i\|^2~\|\bar{b}_j\|^2}{n}\rho_{ij}^2 n\left(1-\kappa\alpha \right)^2+\kappa\alpha^2\frac{\|\bar{a}_i\|^2~\|\bar{b}_j\|^2}{n}\left(\frac{1-\kappa+\kappa\phi\left(2^{-2\frac{(R-\eps_0)}{\kappa}}\right)}{1-\phi\left(2^{-2\frac{(R-\eps_0)}{\kappa}}\right)}\right)+\frac{\|\bar{a}_i\|^2~\|\bar{b}_j\|^2}{n}\eps'(1+\rho_{ij}^2 n)\\
&=(\bar{a}_i^\top \bar{v}_j)^2 \left((1-\kappa\alpha)^2 +\eps'\right)+\frac{\|\bar{a}_i\|^2~\|\bar{b}_j\|^2}{n}\left(\kappa\alpha^2\left(\frac{1-\kappa+\kappa\phi\left(2^{-2\frac{(R-\eps_0)}{\kappa}}\right)}{1-\phi\left(2^{-2\frac{(R-\eps_0)}{\kappa}}\right)}\right)+\eps'\right).
\end{align}
Therefore,
\begin{align}
\EE(C_{ij}-\hat{C}_{ij})^2&\leq   (\bar{a}_i^\top \bar{v}_j)^2 \left((1-\kappa\alpha)^2 +\eps'\right)+\frac{\|\bar{a}_i\|^2~\|\bar{b}_j\|^2}{n}\left(\kappa\alpha^2\left(\frac{1-\kappa+\kappa\phi\left(2^{-2\frac{(R-\eps_0)}{\kappa}}\right)}{1-\phi\left(2^{-2\frac{(R-\eps_0)}{\kappa}}\right)}\right)+\eps'\right)\nonumber\\
&+4\|\bar{a}_i\|~ \|\bar{b}_j\|\eps_1 +\eps^2_1.
\end{align}
Thus, recalling that $\|
\bar{a}_i\|~\|\bar{b}_j\|\leq nM^2$, for any $\eps>0$ and $n$ large enough
\begin{align}
\EE(C_{ij}-\hat{C}_{ij})^2&\leq   (\bar{a}_i^\top \bar{v}_j)^2 \left((1-\kappa\alpha)^2 +\eps\right)+\frac{\|\bar{a}_i\|^2~\|\bar{b}_j\|^2}{n}\left(\kappa\alpha^2\left(\frac{1-\kappa+\kappa\phi\left(2^{-2\frac{R}{\kappa}}\right)}{1-\phi\left(2^{-2\frac{R}{\kappa}}\right)}\right)+\eps\right)+n^{-8}
\end{align}
The proof of part 1 is complete, by noting that $\tilde{C}_{ij}=\bar{a}_i^\top \bar{b}_j$ and that
\begin{align}
\frac{1}{n}\sum_{i,j}\|\bar{a}_i\|^2 \|\bar{b}_j\|^2&=\frac{1}{n}\sum_{i=1}^a \|\bar{a}_i\|^2 \sum_{j=1}^b \|\bar{b}_j\|^2=\frac{\|\bar{A}\|_F^2 ~ \|\bar{B}\|^2_F}{n},\nonumber\\
\sum_{i,j}\tilde{C}^2_{ij}&=\|\tilde{C}\|^2_{F}.
\end{align}    

The proof for part 2 follows identically from part 2 of Theorem~\ref{thm:universalLattice}.
\end{proof}

With Theorem~\ref{thm:MostgeneralMatMul} at hand, we easily obtain Theorem~\ref{thm:generalMatMul} and Theorem~\ref{thm:generalMatMulNoMMSE} as simple corollaries.

\begin{proof}[Proof of Theorem~\ref{thm:generalMatMul}]
For part 1, let $\alpha=\alpha_{\kappa}=\left(1-\phi\left(2^{-2\frac{R}{\kappa}}\right)\right)$. Applying part 1 of Theorem~\ref{thm:MostgeneralMatMul} gives
\begin{align}
\EE(C_{i,j}-\hat{C}_{i,j})^2\leq  \tilde{C}_{i,j}^2\cdot\left(G^2(R,\kappa)+\eps\right)+\frac{\|\bar{a}_i\|^2 \|\bar{b}_j\|^2}{n}\left((1-G(R,\kappa))G(R,\kappa)+\eps\right)+n^{-8},   
\end{align}
where 
\begin{align}
G(R,\kappa)=1-\kappa+\kappa \phi\left(2^{-2\frac{R}{\kappa}}\right).
\end{align}
Choosing $\kappa=\min\{R/R^*,1\}$, we get $G(R,\kappa)=\Gamma(R)$, establishing the claim.

For part 2, let $\alpha=\alpha_{\kappa}=1-2^{-2\frac{R}{\kappa}}=\frac{2^{2\frac{R}{\kappa}}-1}{2^{2\frac{R}{\kappa}}}$. Applying part 2 of Theorem~\ref{thm:MostgeneralMatMul} gives
\begin{align}
\EE(C_{i,j}-\hat{C}_{i,j})^2\leq  \tilde{C}_{i,j}^2\cdot\left(\tilde{G}^2(R,\kappa)+\eps\right)+\frac{\|\bar{a}_i\|^2 \|\bar{b}_j\|^2}{n}\left((1-\tilde{G}(R,\kappa))\tilde{G}(R,\kappa)+\eps\right)+n^{-8}, 
\end{align}
where 
\begin{align}
\tilde{G}(R,\kappa)=1-\kappa+\kappa 2^{-2\frac{R}{\kappa}}.
\end{align}
Choosing $\kappa=1$, we get $\tilde{G}(R,\kappa)=2^{-2R}$, establishing the claim.
\end{proof}

\begin{proof}[Proof of Theorem~\ref{thm:generalMatMulNoMMSE}]
Follows by applying part 1 and part 2 of Theorem~\ref{thm:MostgeneralMatMul} with $\alpha=\kappa=1$ (that is, no time-sharing and no MMSE scaling). Note that a straightforward application of Theorem~\ref{thm:MostgeneralMatMul} with $\alpha=\kappa=1$ leaves an $\|\bar{A}^\top \bar{B}\|_F^2\cdot\eps$ term. However, a careful inspection of the proof of Theorem~\ref{thm:universalLattice} shows that this term is not needed for the special case of $\alpha=\kappa=1$.
\end{proof}
 
\section{Nested Lattice Quantization for Inner Product Computation}
\label{sec:nestedlatticeIP}

\subsection{Lattices}
\label{subsec:latticebasics}

We review some basic lattice definitions. See~\cite{ramiBook} for a comprehensive treatment of lattices in information theory. For a lattice $L\subset\RR^d$ we define the nearest neighbor quantizer $Q_L:\RR^d\to L$ as
\begin{align}
Q_{L}(x)=\argmin_{\lambda\in L}\|x-\lambda\|, 
\end{align}
where ties are broken arbitrarily, but in systematic manner. The Voronoi region $\m{V}_L$ is defined as the set of all points in $\RR^n$ that are closer to $0$ than to any other lattice point
\begin{align}
\m{V}_L=\left\{x\in\RR^d~:~Q_L(x)=0\right\}.    
\end{align}
Any lattice $L\subset \RR^d$ has a (non-unique) generating matrix $G\in\RR^{d\times d}$ such that $L=G\ZZ^d$. The covolume of the lattice $L$, denoted $\mathrm{covol}(L)$, is the volume of its Voronoi region (or any other fundamental cell of $L$), which is also equal to $|\det G|$. Let $\m{B}=\{x\in\RR^d~~:~~\|x\|\leq 1\}$ be the unit $\ell_2$ ball in $\RR^d$, whose volume is 
\begin{align}
V_d=\frac{\pi^{d/2}}{\Gamma\left(1+\frac{d}{2} \right)},  
\end{align}
where here $\Gamma$ is Euler's Gamma function, not to be confused with $\Gamma(R)$ defined in~\eqref{eq:barphidef}.
We denote by $\reff(L)=(\mathrm{covol}(L)/V_d)^{\frac{1}{d}}$ the effective radius of $L$, that is, the radius of a $\ell_2$ ball in $\RR^d$ whose volume $V_d \reff^d(L)$ equals $\mathrm{covol}(L)$. The covering radius of $L$ is defined as
\begin{align}
 \rcov(L)=\min\left\{r>0~:~L+r\m{B}=\RR^d\right\}=\max\left\{\|x\|~:~x\in\m{V}_L\right\} .
\end{align}
Clearly, $\reff(L)\leq \rcov(L)$. Let $Z\sim\Unif(\m{V}_L)$ be a random vector uniformly distributed over the Voronoi region of $L$. We define the second moment of the lattice $L$ as
\begin{align}
\sigma^2(L)=\frac{1}{d}\EE\|Z\|^2,    
\end{align}
and the covariance matrix of $L$ as
\begin{align}
R(L)=\EE[Z Z^\top].    
\end{align}
The modulo operation with respect to the lattice $L$, is defined in this paper as
\begin{align}
[x]\bmod L=x-Q_{L}(x). 
\end{align}
Note that $[x]\bmod L\in\m{V}_L$.

The proof of Theorem~\ref{thm:universalLattice} uses a nested-lattice quantizer~\cite{ramiBook}, based on a pair of nested lattices $\Lambda_c\subset \Lambda_f$ in $\RR^d$. A quantizer is constructed from such a pair by first quantizing each point in $\RR^d$ to $Q_{\Lambda_f}(x)$, the nearest point in the lattice $\Lambda_f$. Since there is an infinite number of points in $\Lambda_f$, the encoder cannot describe $Q_{\Lambda_f}(x)$ using $dR$ bits. Instead, it describes the \emph{coset} of $\Lambda_c$ in which $Q_{\Lambda_f}(x)$ lies. There are $|\Lambda_f/\Lambda_c|=\mathrm{covol}(\Lambda_c)/\mathrm{covol}(\Lambda_f)$ such cosets, and therefore, if $|\Lambda_f/\Lambda_c|\leq 2^{dR}$ the encoder can indeed send that information with $dR$ bits. When the decoder gets this information, it knows that $Q_{\Lambda_f}(x)\in Q_{\Lambda_f}(x)+\Lambda_c$, but does not know which point within this coset is $Q_{\Lambda_f}(x)$. Typically, the decoder will output the most likely member from the coset. For the case where $X$ is an iid Gaussian vector in $\RR^d$, this (approximately) corresponds to selecting $\hat{x}$ as the member with the smallest energy in $Q_{\Lambda_f}(x)+\Lambda_c$. Under this paradigm, the reconstruction points are $\Lambda_f\cap \m{V}_{\Lambda_c}$.

Many works have established the existence of nested lattices that simultaneously posses many desired properties, namely, relatively large packing radius, small covering radius, small second moment, and resilience to noise~\cite{erez2004achieving,erez2005lattices,krithivasan2007proof,kudryashov2007random,ling2014achieving,ling2014semantically,ramiBook,ncng16,campello2016algebraic,oe17,HuangNarayanan17,di2017lda,orw22,ordentlich2023bounds,sadeghi2024simpler,liu2024quantization}. In the problem of quantization for inner product computation a new ingredient enters the picture, that was not previously needed. The inner product reconstruction error includes a term that consists of the inner product of the quantization errors of each one of the vectors. To control this term, we need the variance of the inner product $Z^\top \bar{Z}$ between two independent dither vectors $Z,\bar{Z}\sim\Unif(\m{V}_L)$ to be small. This in turn, requires the spectrum of the quantization error to be ``nearly-white'' in the Frobenius norm sense. Namely, a good lattice quantizer $L$ for the inner product problem needs to satisfy $\frac{1}{d}\|R(L)\|_F^2\approx (\sigma^2(L))^2$. While the optimal lattice quantizer in dimension $d$ always satisfies $R(L)=\sigma^2(L)\cdot I_d$~\cite{ramiBook,zamir1996lattice}, we do not know whether it also has the additional required properties, e.g., resilience to noise. We must therefore resort to analyzing a random ensemble of nested lattices, and show that in addition to all other required properties, they also typically have small $\frac{1}{d}\|R(L)\|_F^2$. Our proof that a random lattice has small $\frac{1}{d}\|R(L)\|_F^2$ relies on the fact that the covering density of a random lattice is only polynomial in the dimension, which was recently proved in~\cite{orw22}. We prove the following result in Appendix~\ref{appendix:latticeproofs}. Except for item~\ref{itemL1:FrobCov} which required new ideas, the proofs for all other items follow the techniques developed in the papers on lattice goodness that were cited above.

\begin{theorem}\label{thm:good_lattices}
There are universal constants $C_1,C_2$ such that for any distribution $P_U$ on $\RR^d$, any $r_U>0$, $D>0$, $\alpha,\beta>0$, $0<\eps\leq \frac{1}{\sqrt{2}}$, and 
\begin{align}
R\geq \frac{1}{2}\log\left(\beta^2+\alpha^2 \frac{r_U}{ D}\right)+C_1\left(\eps+\frac{\log d}{d}\right)   
\label{eq:Rconstraint}
\end{align}
there exists a pair of nested lattices $\Lambda_c\subset\Lambda_f$ in $\RR^d$ satisfying the following
\begin{enumerate}
\item $\frac{1}{2}2^{dR}\leq |\Lambda_f/\Lambda_c|\leq 2^{dR}$\label{itemL1:rate} 
\item $\rcov(\Lambda_f)\leq \sqrt{dD}$ and $\rcov(\Lambda_c)\leq 2^R\sqrt{dD}$; \label{itemL1:rcov}
\item $\sigma^2(\Lambda_f)\leq D$;\label{itemL1:D}
\item $\frac{1}{d}\|R(\Lambda_f)\|_F^2\leq D^2(1+\frac{C_2\log^3{d}}{d})$;\label{itemL1:FrobCov}
\item For $U\sim P_U$, $Z~\sim\Unif(\m{V}_{\Lambda_f})$, $Z\indep U$, we have
\begin{align}
\Pr(\alpha U+\beta Z \notin \m{V}_{\Lambda_c})\leq \Pr(\|U\|^2> d\cdot r_U)+6e^{-d \frac{\eps^2}{2}}.
\end{align}
\label{item:L1pe}
\end{enumerate}
\end{theorem}

\begin{remark}
If we further require that $R=\log q$ for some integer $q\geq 2$, there exists a pair of self-similar nested lattices $\Lambda_c=q \Lambda \subset \Lambda=\Lambda_f$ satisfying the statements in Theorem~\ref{thm:good_lattices}. The proof is essentially the same.    
\end{remark}

\begin{remark}
While our proof for Theorem~\ref{thm:good_lattices} does not impose any particular structure on the lattices $\Lambda_c\subset\Lambda_f$, it is possible to prove the existence of Construction A lattices $\Lambda_c\subset\Lambda_f$ satisfying Theorem~\ref{thm:good_lattices}. This follows from~\cite[Corollary 1.5]{ordentlich2023bounds} that shows that the covering density of a typical Construction A lattice (with judiciously chosen parameters) is also polynomial in the dimension.
\end{remark}

\medskip

\subsection{Proof of Theorem~\ref{thm:universalLattice}}

\subsubsection{Dithered Nested Lattice Quantization for Inner Product}
\label{subsec:nestedlatticescheme}

Let $d=\lfloor \kappa n \rfloor$, and denote $U_{[d]}=(U_1,\ldots,U_d)^\top$ and similarly $V_{[d]}=\sqrt{\rho} U_{[d]}+\sqrt{1-\rho^2} Z_{[d]}$. Let $\tilde{R}=\frac{n}{d}R\geq \frac{R}{\kappa}$, and
let $\Lambda_c\subset\Lambda_f$ be a pair of nested lattices in $\RR^d$, with $|\Lambda_f/\Lambda_c|\leq 2^{d\tilde{R}}$. Let $\tilde{Z}_1,\tilde{Z}_2\sim\Unif(\m{V}_{\Lambda_f})$ be statistically independent dither vectors.
Our encoders $f_1,f_2:\RR^n\to[2^{nR}]$ compute
\begin{align}
\tilde{U}_{[d]}&=\left[Q_{\Lambda_f}\left( U_{[d]}+\tilde{Z}_1\right)\right]\bmod \Lambda_c\\
\tilde{V}_{[d]}&=\left[Q_{\Lambda_f}\left(V_{[d]}+\tilde{Z}_2\right)\right]\bmod \Lambda_c,
\end{align}
and each of them maps the result to $nR=d\tilde{R}$ bits (which is possible since $|\Lambda_f/\Lambda_c|\leq 2^{d\tilde{R}}$).

The decoder $g(f_1(U),f_2(V))$ computes
\begin{align}
\hat{U}_{[d]}&=\left[\tilde{U}_{[d]}-\tilde{Z}_1\right]\bmod \Lambda_c \label{eq:decU}\\
\hat{V}_{[d]}&=\left[\tilde{V}_{[d]}-\tilde{Z}_2\right]\bmod \Lambda_c \label{eq:decV},
\end{align}
and estimates the inner product as 
\begin{align}
g(f_1(U),f_2(V))=\hat{U}_{[d]}^\top \hat{V}_{[d]}.    
\end{align}

\medskip
\medskip

\subsubsection{Analysis}
\label{subsec:nestedlatticeanalysis}

We now analyze the performance of this scheme. First, note that
\begin{align}
\hat{U}_{[d]}=\left[\tilde{U}_{[d]}-\tilde{Z}_1\right]\bmod \Lambda_c&=\left[\left[Q_{\Lambda_f}\left( U_{[d]}+\tilde{Z}_1\right)\right]\bmod \Lambda_c-\tilde{Z}_1\right]\bmod \Lambda_c\nonumber\\
&=\left[Q_{\Lambda_f}\left(U_{[d]}+\tilde{Z}_1\right)-\tilde{Z}_1\right]\bmod \Lambda_c\nonumber\\
&=\left[U_{[d]}+\left(Q_{\Lambda_f}\left(U_{[d]}+\tilde{Z}_1\right)-(U_{[d]}+\tilde{Z}_1)\right)\right]\bmod \Lambda_c\nonumber\\
&=\left[ U_{[d]}+Z_1\right]\bmod \Lambda_c,
\label{eq:modLambdac1}
\end{align}
where
\begin{align}
Z_1=Q_{\Lambda_f}\left(U_{[d]}+\tilde{Z}_1\right)-\left(U_{[d]}+\tilde{Z}_1\right)
\end{align}
is uniform over $-\m{V}_{\Lambda_f}=\m{V}_{\Lambda_f}$ and statistically independent of $U$ (and everything else), by the Crypto Lemma~\cite{erez2004achieving,ramiBook}. Similarly, we obtain
\begin{align}
\hat{V}_{[d]}=\left[\tilde{V}_{[d]}-\tilde{Z}_2\right]\bmod \Lambda_c=\left[V_{[d]}+Z_2\right]\bmod \Lambda_c,  
\label{eq:modLambdac2}
\end{align}
where $Z_2\sim\Unif(\m{V}_f)$ is statistically independent of $V$ (and everything else). 

Let
\begin{align}
\Uidd&=U_{[d]}+Z_1\\
\Vidd&=V_{[d]}+Z_2\\
\end{align}
and define the overload events
\begin{align}
\OL_1=\{U_{[d]}+Z_1\notin\m{V}_{\Lambda_c}\},~~\OL_2=\{V_{[d]}+Z_2\notin\m{V}_{\Lambda_c}\},~~\OL=\OL_1\cup \OL_2.\label{eq:overloadDef}    
\end{align}
In particular, if $\mathrm{OL}$ did not occur, the $\bmod\Lambda_c$ operation in~\eqref{eq:modLambdac1} and in~\eqref{eq:modLambdac2} is inactive, and $\hat{U}_{[d]}=\Uidd$ and $\hat{V}_{[d]}=\Vidd$.  
Let
\begin{align}
e=\alpha g(f_1(U),f_2(V))-U^\top V=\alpha\hat{U}_{[d]}^\top \hat{V}_{[d]}-U^\top V,    
\end{align}
and
\begin{align}
\eid=\alpha\Uidd^\top \Vidd -U^\top V.
\end{align}
If $\Pr(\OL)$ is very small,  then intuitively $\EE[e^2]$ should be close to $\EE[\eid^2]$. Indeed, we prove the following in Appendix~\ref{appendix:fromidealtoreal}
\begin{proposition}
\begin{align}
\EE(e^2)\leq \EE[\eid^2]+75\Pr(\OL)\cdot M^4(\rcov(\Lambda_c)) 
\end{align}
where 
\begin{align}
M(\rcov(\Lambda_c))=\max\{\sqrt{n},\rcov(\Lambda_c)\}.
\end{align}
\label{prop:fromidealtoreal}
\end{proposition}

We therefore proceed to compute $\EE(\eid^2)$. Note that
\begin{align}
\eid=\alpha\left(U_{[d]}^\top V_{[d]}+U_{[d]}^\top Z_2+V_{[d]}^\top Z_1+Z_1^\top Z_2\right)-\rho n .
\end{align}
Since $P_{UV Z_1 Z_2}=P_{UV} P_{Z_1} P_{Z_2}$ and all random vectors $U,V,Z_1,Z_2$ have zero mean, we have
\begin{align}
\EE[U_{[d]}^\top V_{[d]}]=\rho d,~\EE[U_{[d]}^\top Z_2]=0,~\EE[V_{[d]}^\top Z_1]=0,~\EE[Z_1^\top Z_2]=0,    
\end{align}
and therefore
\begin{align}
\EE[\eid^2]&=\rho^2 n^2-2\alpha \rho^2 nd +\alpha^2\left(\EE[(U_{[d]}^\top V_{[d]})^2]+\EE[(U_{[d]}^\top Z_2)^2]+\EE[(V_{[d]}^\top Z_1)^2]  +\EE[(Z_1^\top Z_2)^2]\right).
\end{align}
In Appendix~\ref{appendix:projections} we show that 
\begin{align}
\EE[(U_{[d]}^\top V_{[d]})^2]\leq \rho^2 n\frac{d(d+1)}{n}+\frac{d(n-d)}{n},
\end{align}
Furthermore,
\begin{align}
\EE[(V_{[d]}^\top Z_1)^2]&=\EE[(U_{[d]}^\top Z_2)^2]=\trace \EE[U_{[d]} U_{[d]}^\top]\EE[Z_2 Z_2^\top]=\trace \EE[Z_2 Z_2^\top]=\EE\|Z_2\|^2=d\cdot\sigma^2(\Lambda_f),\label{eq:traceU}\\
\EE[(Z_1^\top Z_2)^2]&=\trace \EE[Z_1 Z_1^\top] \EE[Z_2 Z_2^\top]=\trace R^2(\Lambda_f)= \|R(\Lambda_f)\|_F^2,  
\end{align}
where~\eqref{eq:traceU} follows since $\EE[U_{[d]} U_{[d]}^\top]=I_d$. 
Thus,
\begin{align}
\EE[\eid^2]&\leq\rho^2 n^2-2\alpha \rho^2 nd +\alpha^2\left(\rho^2 n\frac{d(d+1)}{n}+\frac{d(n-d)}{n}+2d \sigma^2(\Lambda_f)+d\frac{1}{d}\|R(\Lambda_f)\|_F^2\right)\\
&=n\left[\rho^2 n\left(1-2\alpha \frac{d}{n}+\alpha^2 \frac{d(d+1)}{n^2} \right)+\alpha^2\left(\frac{d(n-d)}{n^2}+2\frac{d}{n}\sigma^2(\Lambda_f)+\frac{d}{n}\frac{1}{d}\|R(\Lambda_f)\|_F^2\right) \right]\label{eq:eidmono}\\
&=n \xi(\Lambda_f)
\end{align}
where
\begin{align}
\xi(\Lambda_f)=\rho^2 n\left(\left(1-\frac{d}{n}\alpha\right)^2 +\alpha^2 \frac{d}{n^2}\right)+\alpha^2\frac{d}{n}\left(1-\frac{d}{n}+\psi\left(\sigma^2(\Lambda_f)\right)\right)+\alpha^2\frac{d}{n}\left(\frac{1}{d}\|R(\Lambda_f)\|_F^2-\sigma^4(\Lambda_f)\right) 
\end{align}
and $\psi(t)=2t+t^2$. Note that $\xi(\Lambda_f)$ is monotonically increasing in both $\sigma^2(\Lambda_f)$ and $\frac{1}{d}\|R(\Lambda_f)\|_F^2$. We have therefore obtained that
\begin{align}
\frac{1}{n}\EE(e^2)\leq   \xi(\Lambda_f) +\frac{75
M^4(\rcov(\Lambda_c))}{n} \Pr(\OL).\label{eq:distortionGeneralLattices} 
\end{align}
The expression in~\eqref{eq:distortionGeneralLattices}, holds for any pair of nested lattices $\Lambda_c\subset\Lambda_f$. We now evaluate it for ``good'' nested lattices, whose existence is guaranteed by Theorem~\ref{thm:good_lattices}. Recall that $R>0$ is fixed. Applying this theorem with $P_{U_{[d]}}$ taken as the uniform (Haar) distribution over $\sqrt{n}\mathbb{S}^{n-1}$ projected to the first $d\leq n$ coordinates,  $\alpha=1$, $\beta=1$, $r_u=1+\eps_0$ and some $0< \eps_0\leq\frac{1}{\sqrt{2}}$, we have that for
\begin{align}
D=\frac{1}{2^{2\left(\frac{n}{d}R-\delta\right)}-1};~~~\delta=C_1\left(2\eps_0+\frac{\log d}{d} \right)    
\end{align}
we can find a pair of nested lattice $\Lambda_c\subset\Lambda_f$ satisfying Items~\ref{itemL1:rate}-\ref{item:L1pe}. In particular, for such lattices we have that
\begin{align}
\frac{1}{n}\xi(\Lambda_f)&\leq\rho^2 n\left(\left(1-\frac{d}{n}\alpha\right)^2 +\alpha^2 \frac{d}{n^2}\right)+\alpha^2\frac{d}{n}\left(1-\frac{d}{n}+\psi\left(D\right)\right)+\alpha^2\frac{d}{n}D^2\frac{C_2 \log^3 d}{d}
\\
\frac{1}{n}M^4(\Lambda_c)&\leq n \max\{1,D^2 2^{4\frac{n}{d}R}\}\\
\Pr(\mathrm{OL})&\leq 2 \Pr(\mathrm{OL}_1)=2\Pr\left(U_{[d]}+Z_1\notin\m{V}_{\Lambda_c}\right)=2\left(\Pr\left(\|U_{[d]}\|^2>(1+\eps_0) d \right)+6 e^{-d\frac{\eps_0^2}{2}}\right)\leq 16 e^{-d\frac{\eps_0^2}{96}},
\end{align}
where in the last inequality we have used Proposition~\ref{prop:projUtail}, proved in Appendix~\ref{appendix:projUtail}, which shows that
\begin{align}
\Pr\left(\|U_{[d]}\|^2>(1+\eps_0) d \right)<2 e^{-\frac{\eps_0^2}{96}d},    
\end{align}
for all $0<\eps_0<1$. Plugging these into~\eqref{eq:distortionGeneralLattices} we get
\begin{align}
\frac{1}{n}\EE(e^2)&\leq \rho^2 n\left(\left(1-\frac{d}{n}\alpha\right)^2 +\alpha^2 \frac{d}{n^2} \right)+\alpha^2\frac{d}{n}\left(1-\frac{d}{n}+\psi\left(D\right)\right)\nonumber\\
&+D^2\frac{C_2\log^3 d}{d}+ 1200 n  \max\{1,D^2 2^{4\frac{n}{d}R}\}e^{-d\frac{\eps_0^2}{96}}.  
\end{align}
It can be verified that
\begin{align}
\psi(D)=\frac{\phi\left(\frac{D}{D+1} \right)}{1-\phi\left(\frac{D}{D+1} \right)},  
\end{align}
which implies that
\begin{align}
\psi\left(\frac{1}{2^{{2t}}-1} \right)=\frac{\phi(2^{-2t})}{1-\phi(2^{-2t})}.    
\end{align}
Thus, for any $\kappa>0$ and $\eps>0$ we can take $\eps_0>0$ small enough and $n$ large enough, we have that
\begin{align}
\frac{1}{n}\EE\left(U^\top V-\alpha g(f_1(U),f_2(V)) \right)^2\leq \rho^2 n\left(1-\kappa\alpha \right)^2+\kappa\alpha^2\left(1-\kappa+\frac{\phi\left(2^{-2\frac{R}{\kappa}}\right)}{1-\phi\left(2^{-2\frac{R}{\kappa}}\right)}\right)+\eps(1+\rho^2 n),    
\end{align}
such that~\eqref{eq:nestedlatticespherical} holds. This establishes the first part of the theorem.

The second part of the theorem follows from the same nested lattice coding scheme for encoding $U$, setting $\hat{V_{[d]}}=V_{[d]}$, and applying the same decoder. The analysis is identical, but with $Z_2=0$.

\section{Practical Implementation of Nested Lattice Quantizers}
\label{subser:practicalLattice}

In the proof of Theorem~\ref{thm:universalLattice} we used a pair of nested lattices $\Lambda_c\subset\Lambda_f\subset\RR^d$, with $|\Lambda_f/\Lambda_c|=2^{dR}$. Given such a pair of lattices in $\RR^d$, in order to implement the coding scheme described above, we need to implement the following procedures:
\begin{enumerate}
    \item $Q_{\Lambda_f}(x)=\argmin_{\lambda_f\in\Lambda_f}\|x-\lambda_f\|$
    \item $Q_{\Lambda_c}(x)=\argmin_{\lambda_c\in\Lambda_c}\|x-\lambda_c\|$
    \item Mapping from $\Lambda_f/\Lambda_c$ to $dR$  bits
    \item Mapping from $dR$  bits to the coset representatives $\Lambda_f\cap\m{V}_c$ of $\Lambda_f/\Lambda_c$
    \item Generating a random dither $Z\sim \Unif(\m{V}_{\Lambda_f})$, where $\m{V}_{\Lambda_f}$ is the Voronoi cell of $\Lambda_f$
\end{enumerate}

\underline{\textbf{Self-similar nested lattice codebooks/Voronoi codes:}} Let $\Lambda\subset \RR^d$ be a lattice with generating matrix $G\in\RR^{d\times d}$, such that $\Lambda= G\ZZ^d$. Assume that we have access to a procedure that implements the lattice quantizer $Q_{\Lambda}(x)$ efficiently, and that there is some $\tau>0$ such that $\tau\ZZ^d\subset \Lambda$. The assumption that $\ZZ^d$ is nested in $\Lambda$ (up to scaling) is not very important, but also not restrictive, since the majority of lattices for which efficient lattice quantizers are known do satisfy it.

Using the lattice $\Lambda$, we can construct a pair of nested lattices $\Lambda_c\subset\Lambda_f\subset\RR^d$, with $|\Lambda_f/\Lambda_c|=2^{dR}$, that induce an efficiently implementable coding scheme. In particular, let $\beta>0$ and set
$\Lambda_f=\beta\Lambda$, $\Lambda_c=q\Lambda_f=\beta \cdot q\Lambda$, where $q=2^R$ is an integer. In~\cite{ConwaySloane83}, Conway and Sloane propose simple implementation of the encoders and decoder for the induced nested lattice quantizer, which they referred to as Voronoi codes.
Algorithm~\ref{alg:nestedlatenc} below provides the pseudo code for implementing $f_1,f_2$ from Subsection~\ref{subsec:nestedlatticescheme} for such a nested lattice codebook. Note that the output $\mathrm{OverloadError}$ of Algorithm~\ref{alg:nestedlatenc} specifies whether or not the overload event $\mathrm{OL}_i$, $i=1,2$, defined in~\eqref{eq:overloadDef} have occurred.
In order to implement the decoder $g$ from Subsection~\ref{subsec:nestedlatticescheme}, one implements~\eqref{eq:decU} by applying Algorithm~\ref{alg:nestedlatdec} on the output of $f_1$, implements~\eqref{eq:decV} by applying Algorithm~\ref{alg:nestedlatdec} on the output of $f_2$, and computes the inner product of the two vectors. In order to generate the random dithers $\tilde{Z}_1,\tilde{Z}_2$, one applies Algorithm~\ref{alg:dithergeneration}.

\begin{algorithm}[h!]
\caption{NestedLatticeEncoder}
\label{alg:nestedlatenc}
\begin{algorithmic}
\State \textbf{Inputs:} vector to be encoded $x\in\RR^{d'}$, lattice $\Lambda\subset\RR^{d'}$ with generating matrix $G\in\RR^{d'\times d'}$, nesting ratio $q\in\mathbb{N}$, dither vector $z\in\m{V}_{\Lambda}\subset\RR^{d'}$, scaling factor $\beta>0$
\State \textbf{Outputs:} $\mathrm{Enc}(x)\in [q]^{d'}$ (can be represented using $\lceil d'\log q \rceil$ bits), $\mathrm{OverloadError}$ that indicates if a modulo error occurred
\vspace{2mm}
\State $t\gets Q_\Lambda\left(\frac{x}{\beta}+z \right)$
\State $y\gets G^{-1}t$
\State $\mathrm{Enc}(x) \gets [y]\bmod q$ (elementwise modulo $q$ reduction)

\vspace{2mm}

\State \% check whether a modulo error occurred:

\vspace{2mm}

\State $\tilde{x}\gets t-z$
\State $\lambda_c=q\cdot Q_{\Lambda}\left(\frac{\tilde{x}}{q} \right)$
\State $\mathrm{OverloadError}=\Ind\left\{\lambda_c\neq 0 \right\}$
\end{algorithmic}
\end{algorithm}

\begin{algorithm}[h!]
\caption{NestedLatticeDecoder}
\label{alg:nestedlatdec}
\begin{algorithmic}
\State \textbf{Inputs:} The encoding $\mathrm{Enc}(x)\in[q]^{d'}$ of $x\in\RR^{d'}$, lattice $\Lambda\subset\RR^{d'}$ with generating matrix $G\in\RR^{d'\times d'}$, nesting ratio $q\in\mathbb{N}$, dither vector $z\in\m{V}_{\Lambda}\subset\RR^{d'}$, scaling factor $\beta>0$
\State \textbf{Outputs:} $\hat{x}\in \RR^{d'}$ 
\vspace{2mm}
\State $\tilde{y}\gets G\cdot \mathrm{Enc}(x) - z$
\State $\hat{x}\gets \beta\left(\tilde{y}-q\cdot Q_{\Lambda}\left(\frac{\tilde{y}}{q} \right) \right) $
\end{algorithmic}
\end{algorithm}

\begin{algorithm}[h!]
\caption{GenerateRandomDither}
\label{alg:dithergeneration}
\begin{algorithmic}
\State \textbf{Inputs:} Lattice $\Lambda\subset\RR^{d'}$ and a number $\tau>0$ such that $\tau\ZZ^{d'}\subset \Lambda$
\State \textbf{Outputs:} $Z\sim\Unif(\m{V}_{\Lambda})$
\vspace{2mm}
\State $U\gets \Unif\left([0,\tau)^{d'}\right)$
\State $Z \gets U-Q_{\Lambda}(U)$
\end{algorithmic}
\end{algorithm}

\medskip

\underline{\textbf{Choice of the parameter $\beta$:}} Using this scheme, we have that 
\begin{align}
D=\sigma^2(\Lambda_f)=\beta^2\sigma^2(\Lambda).    
\end{align}
Thus, since the base lattice $\Lambda$ is given, the parameter $\beta$ controls $D$. We also have that 
\begin{align}
\sigma^2(\Lambda_c)=q^2\sigma^2(\Lambda_f)=2^{2R}D.
\end{align}
The ``no-overload'' event is equivalent to $U_{[d]}+Z_1\in\m{V}_c$ (and similarly, $V_{[d]}+Z_2\in\m{V}_c$). If $\Lambda$ is a ``good'' high-dimensional ($d\gg 1$) lattice, that is $\Lambda$ is such that $\beta q\Lambda=\Lambda_c\subset\Lambda_f=\beta\Lambda$ satisfy all items in Theorem~\ref{thm:good_lattices}, the ``no-overload'' event happens with high probability provided that $1+D=\frac{1}{d}\EE\|U+Z_1\|^2<2^{2R}D$, which is equivalent to $D>D^*(R)=\frac{1}{2^{2R}-1}$. In practice, we will usually work with a base lattice $\Lambda$ whose second moment and coding goodness are sub-optimal. For this reason, we take $D=\gamma D^*(R)=\frac{\gamma}{2^{2R}-1}$, for some $\gamma>0$ (where $\gamma$ is not necessarily close to $1$), which is done by setting 
\begin{align}
 \beta=\left(\frac{\gamma}{2^{2R}-1}\cdot\frac{1}{\sigma^2(\Lambda)} \right)^{1/2}.\label{eq:BetaEq}
\end{align}

\medskip

\underline{\textbf{Overload avoidance mechanism:}} Recall that Algorithm~\ref{alg:nestedlatenc} also indicates, through the variable $\mathrm{OverloadError}$, whether or not a modulo error occurred, that is, whether or not $U_{[d]}+Z_1\in\m{V}_c$ (respectively, $V_{[d]}+Z_2\in\m{V}_c$). Whenever a modulo error does occur, one can increase the value of $\gamma$ further to a large enough value, such that a modulo error does not occur with the new value, and inform the decoder on what value of $\gamma$ was chosen. In practice, we may choose a bank of $M$ values sorted in increasing order $\gamma\in\{\gamma_1,\ldots,\gamma_{M}\}$. The encoder first uses $\gamma_1$. If $\mathrm{OverloadError}=1$ it tries again with $\gamma_2$, and keeps increasing $\gamma$ to the next value until $\mathrm{OverloadError}=0$. If $\gamma_1$ is chosen such that overload error is already not too common, and the values of $\gamma_i$ increase sufficiently fast with $i$, say $\gamma_i=i\cdot\gamma_1$, the entropy of the first value of $\gamma$ that returned $\mathrm{OverloadError}=0$ will be small. Since we only have to report this index to the decoder once for $d$ symbols, the effect on the quantization rate is not significant. We note that the newest GPUs released by Nvidia use a format called \emph{micro-block scaling} for low-precision data types such as FP4~\cite{NvidiaNVFP4blog}. For example, the NVFP4 data type assigns a different scale $\beta$, represented in FP8 format, to every 16 consecutive entries, where $\beta$ is chosen to ensure that all 16 entries are within the dynamic range of the FP4 format, or in other words, that overload does not occur. Thus, such formats can be seen as a special case of our overload avoidance mechanism.

Next, we develop a heuristic for choosing $\gamma_1$. Recall the definition of $r_{\mathrm{eff}}(\Lambda)=\left(\frac{\mathrm{covol}(\Lambda)}{V_d}\right)^{1/d}$  from Section~\ref{subsec:latticebasics}. The normalized second moment (NSM) of a lattice $\Lambda$ is defined as
\begin{align}
N(\Lambda)=\frac{\sigma^2(\Lambda)}{(\mathrm{covol}(\Lambda))^{2/d}}=\frac{\sigma^2(\Lambda)}{V_d^{2/d} r^2_{\mathrm{eff}}(\Lambda)}.    
\end{align}
If $U_{[d]}+Z_1$ were Gaussian, the probability that it stays within $\m{V}_{\Lambda_c}$ would have been upper bounded by the probability that it stays within a ball with the same volume, that is, within a ball with radius $r_{\mathrm{eff}}(\Lambda_c)$. Thus, we need $r^2_{\mathrm{eff}}(\Lambda_c)$ to be greater than $\EE\|U_{[d]}+Z_1\|^2$. This corresponds to
\begin{align}
1<\frac{\frac{1}{d}r^2_{\mathrm{eff}}(\Lambda_c)}{\frac{1}{d}\EE\|U+Z_1\|^2}=\frac{1}{d}\frac{r^2_{\mathrm{eff}}(\Lambda_c)}{\sigma^2(\Lambda_c)}\frac{\sigma^2(\Lambda_c)}{\frac{1}{d}\EE\|U+Z_1\|^2}=\frac{1}{d V_d^{2/d} N(\Lambda)}\frac{2^{2R}D}{1+D}=\frac{1}{d V_d^{2/d} N(\Lambda)}\frac{\gamma\cdot 2^{2R}}{2^{2R}+\gamma-1}\approx\frac{\gamma}{d V_d^{2/d} N(\Lambda)},
\end{align}
where the last approximation assumes that $2^{2R}+\gamma-1\approx 2^{2R}$.
Thus, we will take
\begin{align}
\gamma_1\gtrapprox d V_d^{2/d} N(\Lambda)=\frac{d\sigma^2(\Lambda)}{\reff^2(\Lambda)}.\label{eq:gamma1ruleofthumb}  
\end{align}
For a measurable set $\m{K}\subset\RR^d$ let $U_{\m{K}}\sim\Unif(\m{K})$ and $\sigma^2(\m{K})=\frac{1}{d}\EE\|U_{\m{K}}\|^2$. For all measurable sets $\m{K}$ with volume $V_d\reff^d(\Lambda)$, we have that $\sigma^2(\m{K})\geq\frac{\reff^2(\Lambda)}{d+2}$, and this is attained by $\m{K}=\reff(\Lambda)\m{B}$~\cite{ramiBook}. It therefore follows that the right hand side of~\eqref{eq:gamma1ruleofthumb} is at least $\frac{d}{d+2}$.


\medskip

\underline{\textbf{Product lattices/Product quantization:}} In order to use the self-similar nested lattice scheme described above, we need a base lattice $\Lambda$ with an efficient nearest-neighbor decoder/lattice quantizer $Q_{\Lambda}(x)$ and favorable quantization and coding properties. While it is easy to find (more accurately, to randomly draw) lattices in high-dimensions that are good for coding and quantization (see Section~\ref{subsec:latticebasics}), the task of finding such lattices that also admit an efficient nearest-neighbor decoder is notoriously difficult and is perhaps the holy grail of coding theory for the additive white Gaussian noise (AWGN) channel. A popular compromise between efficiency and ``goodness'', is to use a \emph{product lattice}, with a low-dimensional base lattice that is ``pretty-good'' for coding and quantization~\cite{forney1998modulation,ramiBook}.

Let $d'$ be an integer that divides $d$, and $\Lambda'\subset\RR^{d'}$ be a lattice in $\RR^{d'}$. We construct the lattice $\Lambda\in\RR^d$ as the product of $K=d/d'$ copies of $\Lambda'$. Namely,
\begin{align}
\Lambda=\underbrace{\Lambda'\times \cdots\times \Lambda'}_{K~\text{times}}=\Lambda'^{\otimes K}    
\end{align}
The resulting self-similar nested lattices are also the product of $K$ nested lattice pairs
\begin{align}
\Lambda_c\subset\Lambda_f=(\beta_1\cdot q\Lambda'\subset \beta_1 \Lambda')\times\cdots \times (\beta_K\cdot q\Lambda'\subset \beta_K \Lambda'),
\label{eq:productnestedlattice}
\end{align}
where we allow for different choices of $\beta$ for each product to accommodate for the overload avoidance mechanism described above. Algorithm~\ref{alg:nestedlatenc}, Algorithm~\ref{alg:nestedlatdec} and Algorithm~\ref{alg:dithergeneration} tensorize, and should be applied separately for each $k=1,\ldots,K$ using the base lattice $\Lambda'\subset\RR^{d'}$ with generating matrix $G'\in\RR^{d'\times d'}$. We also have that
\begin{align}
\sigma^2(\Lambda)=\sigma^2(\Lambda')\cdot \frac{1}{K}\sum_{\ell=1}^K \beta^2_k.
\end{align}

Some lattices in small dimensions have excellent quantization and coding properties, as well as efficient nearest neighbor decoding algorithms. In particular, $A_3\cong D_3$ has the highest packing density among all lattices in $\RR^3$~\cite{ConwaySloane}, $A^*_3$ has the smallest NSM among all lattices in $\RR^3$~\cite{ConwaySloane} (only slightly smaller than that of $A_3$), $D_4$ has the highest packing density among all lattices in $\RR^4$ and lowest known NSM among all lattices in $\RR^4$~\cite{ConwaySloane,AgrellAllen22}, and $E_8$ has the highest packing density (even among non-lattice packings)~\cite{viazovska2017sphere} and the smallest known NSM among all lattices in $\RR^8$~\cite{ConwaySloane,AgrellAllen22}. All four lattices listed above, as well as many others from the $A_n$, $D_n$ and $E_n$ families, admit a very fast lattice decoding algorithm~\cite{conway1982fast}. Similarly, among all lattices in $\RR^{24}$, the Leech lattice $\Lambda_{24}$, is the the best known quantizer~\cite{AgrellAllen22}, has the optimal packing density~\cite{cohn2009optimality} (this is true even among all non-lattice packings~\cite{cohn2017sphere}), and admits a pretty fast nearest neighbor decoding  (or approximate nearest neighbor decoding) algorithms~\cite{VardyBeery93,Vardy95,ViterboBoutros99}. In addition, the second moment of all these lattices (and others) is calculated in~\cite{conway1982voronoi} and reported also in~\cite[Table I]{AgrellAllen22}. See also~\cite{ordentlich2025voronoi} for  a study of the NSM of a random lattice. We also note that the optimal lattice quantizer in any dimension has $R(L)=\sigma^2(L)\cdot I_d$, so that $\frac{1}{d}\|R(L)\|_F^2=\sigma^4(L)$ for those lattices.
Any one of those lattices is a good candidate for the base lattice $\Lambda'$. Another important advantage of these lattices is that they are all subsets of $\ZZ^n$ up to scaling. Thus, when these lattices are used for quantization for matrix multiplication, and dithering is not applied, we can use integer multipliers (e.g., int8 tensor core in a GPU), rather than floating point multipliers, for multiplying the quantized matrices.
The lattices of higher dimensions, and in particular the Leech lattice, may yield better rate-distortion tradeoff than the lower-dimensional ones, but there are advantages to using lower-dimensional lattices in terms of efficiency. One of those is described next.

\medskip

\underline{\textbf{Lookup tables:}} Note that we decode $\hat{U}_{k}\in\RR^{d'}$ and $\hat{V}_{k}\in\RR^{d'}$ just to compute their inner product $\hat{U}^\top _{k}\hat{V}_{k}$. If we use the same dither vectors $\tilde{Z}_1,\tilde{Z}_2\in\m{V}_{\Lambda'}$ for all $k=1,\ldots,K$, and the same value of $\beta$, namely, $\beta^U_k=\beta^V_k=\beta$ for all $k=1,\ldots,K$, there are only $q^{d'}$ values of $\hat{U}^\top _{k}$ we can get, and only $q^{d'}$ values of $\hat{V}^\top _{k}$ we can get. Those do not depend on $k$. Thus, we can pre-compute all $q^{2d'}$ possible values of $\hat{U}^\top _{k}\hat{V}_{k}$ and store them in a lookup table (LUT). Then, instead of applying the decoder twice and computing the inner product, we simply fetch the result of the inner product from the LUT. If $\beta^U_k\neq \beta$ or $\beta^V_k\neq \beta$, we simply multiply the value fetched from the LUT by $\frac{\beta^U_k}{\beta}\cdot \frac{\beta^V_k}{\beta}$. On some processors, using LUTs significantly speed up the decoding process, as it completely bypasses all lattice decoding operations, as well as all inner products. For approximate matrix multiplication $A^\top  B$ of $A\in\RR^{n\times a}$ and $B\in\RR^{n\times b}$ using the product nested lattice quantization scheme above, we need to perform $a\cdot b\cdot (n/d')$ such operations, whereas the encoding only involves $a(n/d')+b(n/d')$ lattice encoding operations. Thus,  for $a,b\gg 1$, decoding is the computationally heavy procedure, and speeding it up will result in significant speedup of the total approximate matrix multiplication procedure. Using LUTs is therefore often highly advantageous. However, in order to have a very fast access time to the LUT, we would like it to ``fit'' in the highest levels of the cache, ideally in the L1 cache. This level has small capacity, which restricts the values of $q^{2d'}=2^{2Rd'}$. Thus, we must keep $Rd'$ small. Taking small $R$ will typically not yield satisfactory resolution, so if LUTs are used, we are limited to using lattices $\Lambda'$ of small dimensions. Fortunately, a compelling solution to this shortcoming of the LUT approach was recently developed in~\cite{ko25}. We note that for GPUs the LUT approach may not be attractive since the tensor core computes matrix multiplications extremely fast, while LUT probing is less efficient on this hardware. On CPUs on the other hand, the LUT approach can potentially yield significant speed-up.

\medskip

\underline{\textbf{Hadamard transform:}} Our encoders $f_1, f_2$ for the matrix multiplication problem, as described in the proof of Theorem~\ref{thm:MostgeneralMatMul} and in Figure~\ref{fig:encoders},  multiply each column vector in $A$ as well as each column vector in $B$ (more accurately, in their centered versions $\bar{A},\bar{B}$), by a random projection matrix $S$ drawn from the Haar distribution on $\mathrm{O}_n(\RR)$. In general, the matrix $S$ drawn from this distribution will have no structure, and calculating $SA$ (respectively $SB)$ will require $O(an^2)$ (respectively, $O(bn^2)$) real-valued multiplication and summation operations. To significantly reduce the computational burden of this step, it was proposed in~\cite{tseng2024quip} (see also~\cite{quarot2024}) to restrict $S$ to a certain class of orthogonal projection matrices: The randomized Hadamard transform. Here, we also follow this approach. In particular, we draw a vector $T\sim\Unif(\{-1,1\}^n)$, and set $K=\mathrm{diag}(T)$, that is, $K$ is a diagonal matrix with $K_{i,i}=T_i$. We then set 
\begin{align}
S=\frac{1}{\sqrt{n}}H K, \label{eq:randHad}   
\end{align}
where $H\in\{-1,1\}^{n\times n}$ is the Walsh-Hadamard matrix of dimension $n$. Here, we assumed that $n$ is a power of $2$, such that such a matrix exists. Otherwise, we can add rows of all zeros to both $A$ and $B$, resulting in larger matrices $A\in\RR^{n'\times a}$ and $B\in\RR^{n'\times b}$, with $n'=2^{\lceil \log_2(n)\rceil}$. Note that in~(\ref{eq:RandRotA}-\ref{eq:RandRotB}) we further scale the result by $\sqrt{n}$, so this cancels out the scaling by $\frac{1}{\sqrt{n}}$ in~\eqref{eq:randHad}. The gain for using the randomized Hadamard transform~\eqref{eq:randHad}, is that its special fast-Fourier transform (FFT) structure allows to compute $SA$ (respectively, $SB$) using only $O(an\log n)$ (respectively, $O(bn\log n)$) additions and multiplications. Despite its simple implementation, the result of applying the randomized Hadamard transform on $A$ (or $B$) is quite similar to that of applying a ``pure'' random rotation on $A$ (or $B$) from various statistical perspectives~\cite{ailon2009fast,liberty2009accelerated,tropp2011improved}.

\medskip

\underline{\textbf{Representative numeric example:}} To better illustrate how the building blocks above connect, we provide a numerical example. We have implemented a product nested lattice codebook, with $\Lambda'=D_3$ (such that $d'=3$) as the base lattice. The lattice $D_3$ consists of all vectors in $\ZZ^3$ whose entries sum up to an even integer. In particular, $2\ZZ^3\subset D_3$. The simple structure of $D_3$ also gives rise to a very simple algorithm for computing $Q_{D_3}(x)$~\cite[Algorithm 2]{conway1982fast}. The lattice $D_3$ has the highest packing density among all lattices in $\RR^3$ and its packing radius satisfies~\cite{ConwaySloane} $r_{\mathrm{pack}}(D_3)/r_{\mathrm{eff}}(D_3)\approx (0.74)^{1/3}\approx 0.9045$, such that its Voronoi region is quite close to a ball. We also have that $\sigma^2(D_3)=\frac{3}{24}$, so that $N(D_3)\approx 0.0787$ (since $\mathrm{covol}(D_3)=2$). This NSM is only slightly greater than the smallest NSM attained by any lattice in $\RR^3$, which is $N(A^*_3)\approx 0.0785$.

We have used this base lattice with $q=6$ to construct a product nested lattice code as in~\eqref{eq:productnestedlattice}. We used the same dither vectors $\tilde{Z}_1,\tilde{Z}_2\in\m{V}_{D_3}$ for all $k=1,\ldots,K$ (these vectors were drawn once at the beginning of the experiment).
For this choice of $d'=3$ and $q=6$, we can implement the decoder using a lookup table of size $(q^3)^2=2^{6\log_2 q}<2^{15.6}$. For constructing the LUT, we used the value $\beta=1$. While for this choice of $\beta$ all inner products between vectors in $D_3$ are integer valued, because of the use of dithers, the entries in our LUT are not integer-valued in general. We nevertheless rounded each of them to the nearest integer, and their range allows representing each entry in the LUT using an INT8 variable. Consequently, the total size of the LUT is less than $64$Kbyte, and it can be fully stored in the L1 cache of a modern processing unit. 

For the lattice $D_3$, we have that the right-hand side of~\eqref{eq:gamma1ruleofthumb} evaluates to $\approx 0.6139$. For the experiment, we chose $\gamma_1=0.7$, and set our bank of possible values of $\gamma$ as $\{i\cdot\gamma_1\}_{i=1}^9$. The corresponding value of $\beta$ is given by~\eqref{eq:BetaEq}.

We drew two random matrices $A\in\RR^{n\times n}$, $B\in\RR^{n\times n}$, with all entries iid
$\m{N}(0,1)$, where $n=3\cdot 2^{11}$. We used the product nested lattice codebook
from~\eqref{eq:productnestedlattice} with $K=2^{11}$ for encoding each column of $A$ (using the dither vector $\tilde{Z}_1$) and for encoding each column of $B$ (using the dither vector $\tilde{Z}_2$). Since the iid Gaussian distribution is already rotation invariant, we have not implemented a random rotation. Since the distribution we have used is zero mean, we also did not implement the ``centering'' mechanism. We also used $\alpha=\kappa=1$ (no MMSE scaling and no time-sharing). For each column, we further report the $K$ values of $\beta_i$ (equivalently $\gamma_i$) used for each column. The (empirical) entropy of this random variable (that takes values in $\beta_1\cdot\{i\}_{i=1}^9$) for the choice $\gamma_1=0.7$ was found to be around $\approx 1.3$bits. Since this value is only reported once for every $d'=3$ symbols (using entropy coding), its contribution to the coding rate is about $0.43$ bits per symbol, such that the total rate of the coding scheme is $R_{\text{eff}}\approx \log_2(6)+0.43\approx 3.015$ bits/symbol.

This approximate matrix multiplication algorithm attained $\frac{1}{n^3}\|\widehat{A^\top  B}-A^\top  B\|_F^2\approx 0.0593$. Let $e=\widehat{A^\top  B}-A^\top  B$. The empirical distribution of the normalized approximation error $e/\sqrt{n}$ (among the $n^2$ entries) is plotted in Figure~\ref{fig:practicalschem}. Note that for $R_{\text{eff}}=3.015$, Theorem~\ref{thm:GaussMatrixLowerBound} states that no scheme can attain distortion smaller than of $\Gamma(R_{\text{eff}})=0.0304$ for $A$ and $B$ drawn as above, and Theorem~\ref{thm:generalMatMul} shows that this can be attained using high-dimensional lattices. Thus, our low-complexity implementation is not far of the optimal performance attained using optimal lattice codes. For comparison, we also evaluated the approximation error for a simple $3$-bit scalar quantization scheme where each column $a_i$ is normalized by $\|a_i\|_{\infty}$ such that all its entries are in $[-1,1]$, then each entry $\tilde{a}_{i,t}=\frac{a_{i,t}}{\|a_i\|_{\infty}}$ is quantized to $\frac{1}{4}\mathrm{round}(4\tilde{a}_{i,t})$, and in the end the quantized entries are rescaled again by $\|a_i\|_{\infty}$. The empirical error attained by the $3$-bit scalar quantizer is $\frac{1}{n^3}\|\widehat{A^\top  B}-A^\top  B\|_F^2\approx 0.1668$, about $3$ times greater than the error attained using the $D_3$-based scheme with the same rate. The performance gap between the two scheme grows with $n$, as the random variable $\|a_i\|_{\infty}$ concentrates around $\sqrt{2\ln n}$ for large $n$. Thus, the dynamic range for the scalar quantizer increases with $n$, which results in greater expected squared error.

\begin{figure}
		\centering
		\includegraphics[width=\textwidth]{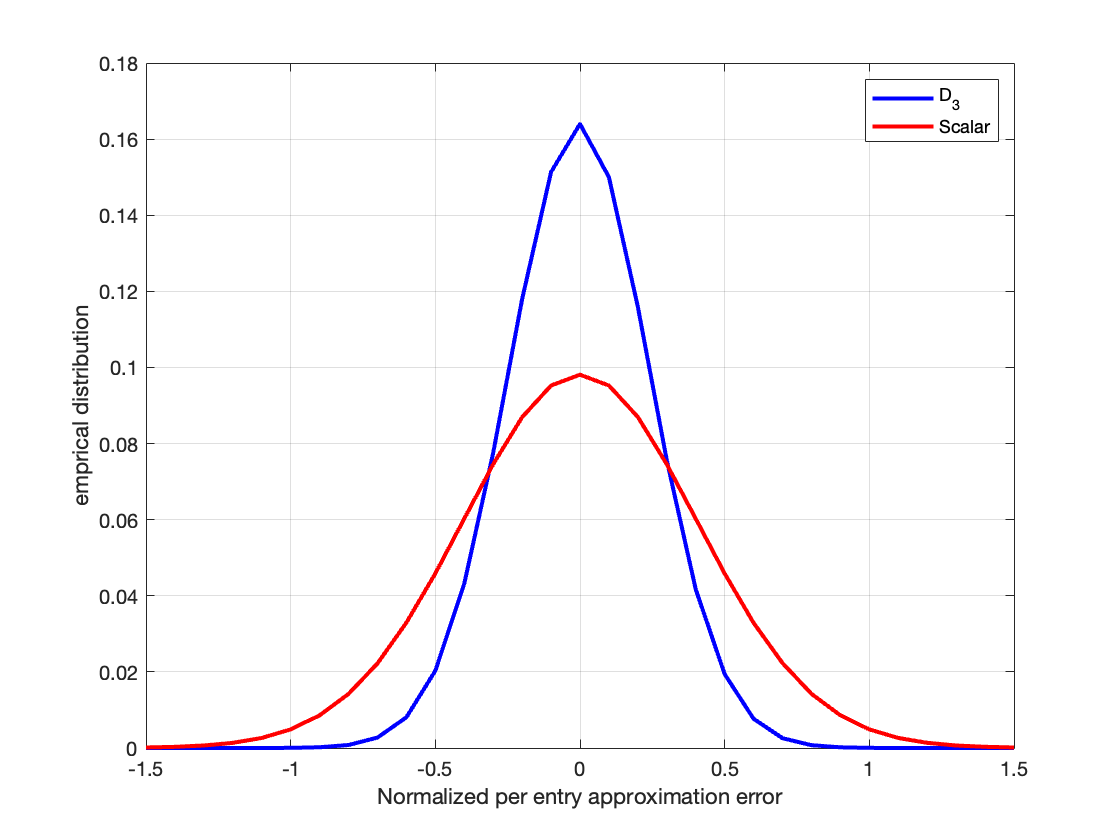}\vspace{-0.1cm}
	\caption{The approximation error of the $D_3$-based product nested lattice coding scheme with $q=6$, for random iid Gaussian matrices $A,B\in \RR^{n\times n}$, $n=3\cdot 2^{11}$. We plot the histogram of the entries of $\frac{1}{\sqrt{n}}(\widehat{A^\top  B}-A^\top  B)$ in blue. For comparison, we also plot the histogram of the entries of $\frac{1}{\sqrt{n}}(\widehat{A^\top  B}-A^\top  B)$ for a $3$-bit scalar quantizer in red.}\vspace{-0.3cm}
	\label{fig:practicalschem}
\end{figure}

\section{Open problems}\label{sec:open}

 One can interpret our Lemma~\ref{lem:SLB} as follows: Let $P=\m{N}(0,1)$ and $U^n\sim P^{\otimes n}$. Then for any random variable $Y$ we have that
\begin{align}
\sum_{i=1}^n R_{P}(\lambda_i)\leq I(U^n;Y),\label{eq:spectralconj} 
\end{align}
where $R_P(D)$ is the quadratic rate-distortion function for a source with distribution $P$ and $(\lambda_1,\ldots,\lambda_n)$ are the eigenvalues of $\Cov(U^n|Y)$. While Lemma~\ref{lem:SLB} establishes~\eqref{eq:spectralconj} for the Gaussian distribution, we were not able to prove~\eqref{eq:spectralconj} for a general distribution, and we could neither find a counterexample. If~\eqref{eq:spectralconj} turns out to hold for any $P$, the proof of Theorem~\ref{thm:SLBbasedDistortionLB} could be easily extended to show that  
\begin{align}
\DIP(R,P)= \text{convex envelope of } (\phi(D_P(R)),
\end{align}
where $D_P(R)$ is the quadratic distortion-rate function for a source with distribution $P$. Thus, proving or disproving that~\eqref{eq:spectralconj} holds for all $P$ is an interesting problem for future research.

In Theorem~\ref{thm:generalMatMul} we have shown the existence of encoders and decoder for quantization for matrix multiplication whose expected approximation error depends only on $\|\bar{A}\|^2_F \cdot\|\bar{B}\|^2_F$ and $\|\bar{A}^\top \bar{B}\|^2_F$, and is optimal for $A$ and $B$ whose entries are iid Gaussian. For iid Gaussian matrices we have that $\frac{\EE[\|A\|_F^2\|B\|_F^2]/n}{\EE[\|A^\top  B\|_F^2]}=1$ so that the two error terms in~\eqref{eq:generalMatMulUB} are well-balanced, and Theorem~\ref{thm:generalMatMul} essentially gives an upper bound of $\|A^\top B\|_F^2\cdot  \Gamma(R)$ on the MSE. Is there a scheme that attains MSE at most $\|A^\top B\|_F^2 \Gamma(R)$ universally (for all matrices $A$, $B$, not just iid ones)?



Another important question is \textit{shared randomness}. Our construction crucially depends on
encoders and decoder sharing randomness (which practically is not a big issue, since the
random seed used by the encoder can be stored along the compressed matrix representation). In the
single-terminal lossy compression shared randomness is not necessary. Indeed, suppose we have some 
compact metric space $E$ with distance $d$ and we proved that there exist a (shared randomness)
encoder-decoder pair $(f,g)$ compressing to $L$ bits and achieving simultaneously for all $x\in E$
guarantee:
$$ \EE[d(x, g(f(x)))] \le \Delta\,.$$
Also suppose that there exists an $\epsilon$-net of size $M_1$ in $E$. Fix $\delta>0$ and average
the previous inequality over all $M_1$ elements of the net. Then there must exist a
choice $\omega_1$ of shared randomness so that at most $M_2 = {M_1\over 1+\delta}$ elements of the
$\epsilon$-net have distortion exceeding $(1+\delta)\Delta$. Now repeat the argument for the
subset $M_2$ to find choice $\omega_2$, etc. After $k\le{\log M_1 \over \log (1+\delta)}$ steps we
get $M_{k+1} = 0$. This shows that there must exist a $k 2^{L}$-sized $(1+\delta)\Delta$-net that
approximates each of $M_1$ elements. Thus, the space $E$ can be covered to within distortion
$(1+\delta)\Delta+\epsilon$ without any shared randomness by compressing down to $L + {\log {\log
M_1\over \log (1+\delta)}}$ bits. By choosing $\delta \to 0$ and $\epsilon \ll \Delta$, one can
thus get rid of shared randomness.

However, this method fails in the case of distributed compression of $(A,B)$.
Indeed, the previous argument breaks down because the choice $\omega$ of shared randomness
affects both quantization grids of $A$ and $B$ simultaneously. Thus, it is not possible for the
compressor who only knows $A$ to decide which of the $k$ choices of $\omega$ to use for quantizing
$A$. It remains an open question to understand fundamental limits of deterministic quantizers.
%

As another extension, we may consider the question of quantization for product of $k$ matrices
$\prod_{t=1}^k A_t$. This paper solves the case of $k=2$, but our methods do not seem to be
immediately extendable to the $k>2$ case. One remark we want to make, however, is regarding the
critical rate. For $k=2$ as we saw quantization below $R<0.906$ bit/entry required additional
dimensionality-reduction (or Johnson-Lindenstrauss) step. This critical point was found by
convexifying the function $\Gamma_1(R) = 1-(1-2^{-2R})^2$. Similarly, if one simply asks the question of
optimal quantization for a product of $k$ diagonal Gaussian matrices, one would need to convexify
the function $\Gamma_1(R) = 1-(1-2^{-2R})^k$. The associated  critical rate grows with $k$ from 
$R\approx 0.906$ for $k=2$ matrices to $R\approx 4$ for $k=46$ matrices etc. This suggests that
quantization for deep LLMs at low rates may benefit from dimensionality reduction steps. 



\section*{Acknowledgements}

The authors thank Omri Weinstein (HUJI) for helping them navigate through the  literature on approximate
matrix multiplication and Yoon Kim (MIT) for explaining hardware and performance limitations of
modern quantization algorithms in LLMs. They also thank the anonymous reviewers for their constructive feedback.

\begin{appendices}

\section{Convex envelope of $\Gamma_1(R)$}
\label{appendix:convexenvelop}

Recall that $\phi(t)=2t-t^2$ and
\begin{align}
\Gamma_1(R)=\phi(2^{-2R}).    
\end{align}
We show that the convex lower envelope of $\Gamma_1(R)$ is $\Gamma(R)$. It is easy to verify that $R\mapsto \Gamma_1(R)$ is decreasing, concave on $[0,1/2)$ and convex on $(1/2,\infty)$. Therefore, its convex envelope consists of a linear segment between $(0,\Gamma_1(0)=1)$ and $(R^*,\Gamma_1(R^*))$ and agrees with $\Gamma_1(R)$ for $R>R^*$. The point $R^*\geq 1/2$ is chosen such that the derivative of $\Gamma(R)$ is smooth and non-decreasing. Thus, the convex envelope of $\Gamma_1(R)$ is given by
\begin{align}
\Gamma(R)=\begin{cases}
\Gamma_1(R^*)+\Gamma'_1(R^*)(R-R^*) & R< R^*\\
\Gamma_1(R) & R\geq R^*
\end{cases}    
\end{align}
where $R^*$ is chosen by requiring that $\Gamma(0)=\Gamma_1(0)=1$, or in other words, that
\begin{align}
\Gamma_1(R^*)-R^*\cdot\Gamma'_1(R^*)=1.
\label{eq:RsDef}
\end{align}
Since $\Gamma'_1(R^*)=-4\ln 2\cdot 2^{-2R^*}(1-2^{-2R^*})$ and we can express $\Gamma_1(R^*)$ as $\Gamma_1(R^*)=2\cdot 2^{-2R^*}(1-2^{-2R^*})+2^{-4R^*}$, we have that~\eqref{eq:RsDef} corresponds to
\begin{align}
&2^{-4R^*}+2\cdot 2^{-2R^*}(1-2^{-2R^*})(1+2\ln 2 R^*)=1\nonumber\\
\Longleftrightarrow&2\cdot 2^{-2R^*}(1-2^{-2R^*})(1+2\ln 2 R^*)=(1-2^{-2R^*})(1+2^{-2R^*})\nonumber\\
\Longleftrightarrow& 2\cdot 2^{-2R^*}(1+2\ln 2 R^*)=1+2^{-2R^*},\\
\Longleftrightarrow& 1+4\ln 2 R^*=2^{2R^*}.
\end{align}

\section{Good Nested Lattices}
\label{appendix:latticeproofs}

The proof of Theorem~\ref{thm:good_lattices} will easily follow from Lemma~\ref{lem:ORW}, Lemma~\ref{lem:coveringdensityimplications} and Lemma~\ref{lem:goodforcoding}, below. We first state these Lemmas, and give the proof of Theorem~\ref{thm:good_lattices}, which uses them. The proofs of Lemma~\ref{lem:coveringdensityimplications} and the proof of Lemma~\ref{lem:goodforcoding} are brought afterwards.

Following the notation from~\cite{orw22}, we denote by $\m{L}_d$ the space of lattices of unit covolume in $\RR^d$, and by $\mu_d$ the natural probability distribution on $\m{L}_d$, which we refer to as th Haar-Siegel measure. Let $\reff(1)=V_d^{-\frac{1}{d}}$ be such that $V_d \reff^d(1)=1$ and therefore $\reff(L)=\reff(1)$ for all $L\in\m{L}_d$. For $c_0,c_1>0$ define the set of lattices
\begin{align}
 E_{c_0,c_1}=\left\{L\in\m{L}_d~:~V_d \rcov^d(L)<c_0 d^{c_1} \right\}.
\end{align}

\begin{lemma}[Corollary of Theorem 1.2 from~\cite{orw22}]
For any $c_1>2$ there exists a universal constant $c_0>0$ such that
\begin{align}
\mu_d(\{L\notin E_{c_0,c_1}\})<\frac{1}{3}.
\label{eq:c0c1Pe}
\end{align}
\label{lem:ORW}
\end{lemma}

\begin{lemma}
For any $c_0,c_1>0$, there are universal constants $C,C'>0$ such that for any $L\in E_{c_0,c_1}$, any $\kappa>0$, $\alpha>0$, and $\beta>0$ the following hold
\begin{enumerate}
\item $\rcov(\kappa L)\leq \left(1+C\frac{\log{d}}{d}\right)\kappa\reff(L)$;\label{itemL:rcov}
\item $\sigma^2(\kappa L)\leq \frac{1}{d}\left( \left(1+C\frac{\log{d}}{d}\right)\kappa\reff(L)\right)^2$;\label{itemL:D}
\item Let $Z~\sim\Unif(\m{V}_{\kappa L})$, and let $U$ be some random variable statistically independent of $Z$. Then for any $\m{A}\subset \RR^d$ we have
\begin{align}
\Pr(\alpha U+\beta Z \notin \m{A})\leq d^{C} \Pr(\alpha U+\beta \tilde{Z} \notin \m{A})
\end{align}
where $\tilde{Z}\sim\Unif\left(\left(1+C\frac{\log{d}}{d}\right)\kappa\reff(L)\m{B}\right)$ is statistically independent of $U$.
\label{item:Lpe}
\item $\frac{1}{d}\|R(\kappa L)\|_F^2\leq \left( \frac{1}{d}\left( \left(1+C\frac{\log{d}}{d}\right)\kappa\reff(L)\right)^2\right)^2\left(1+C'\frac{\log^3 d}{d}\right)$;\label{itemL:FrobCov}
\end{enumerate}
\label{lem:coveringdensityimplications}
\end{lemma}

\begin{lemma}\label{lem:goodforcoding}
Let $U\sim P_U$ be a random variable in $\RR^d$ that satisfies $\|U\|\leq \sqrt{d \cdot r_U}$ with probability $1$, and let $\tilde{Z}\sim\Unif(\sqrt{d\cdot r_b}\m{B})$ be statistically independent of $U$. For $\alpha,\beta,\kappa,\eps>0$ let
\begin{align}
E^{P_U,r_U,r_b}_{\alpha,\beta,\kappa,\eps}=\left\{L\in\m{L}_d~:~\Pr(\alpha U+\beta\tilde{Z}\notin \m{V}_{\kappa L})<6e^{-d\frac{\eps^2}{2}} \right\}.    
\end{align}
Then, for any $0<\eps<\frac{1}{\sqrt{2}}$ and
\begin{align}
\kappa>e^{\frac{\eps^2}{2}}\sqrt{1+\eps}\frac{\sqrt{d(\alpha^2 r_U +\beta^2 r_b)}}{\reff(1)}.\label{eq:kappamin}    
\end{align} 
we have that 
\begin{align}
\mu_d\left(\left\{L\notin E^{P_U,r_u,r_b}_{\alpha,\beta,\kappa,\eps}\right\}\right)<\frac{1}{3}.    
\end{align}
\end{lemma}

\begin{proof}[Proof of Theorem~\ref{thm:good_lattices}]
Let $p$ be a prime number and $k$ be a positive integer, such that $p^k\in [1/2,1]2^{dR}$. Such numbers must exist.
Denote by $\mathrm{Gr}_{d,k}(\mathbb{F}_p)$ the collection of subspaces of dimension $k$ in $\mathbb{F}_p^d$. Let $L$ be some lattice in $\m{L}_d$, $S$ be some subspace in $\mathrm{Gr}_{d,k}(\mathbb{F}_p)$, and let the lattice $L(S)$ be as defined in~\cite[eq. 13]{ordentlich2023bounds}. We have that $L\subset L(S)$ and $|L(S)/L|=p^k$. Fix some $c_1>2$ and $c_0>0$ for which~\eqref{eq:c0c1Pe} holds. Let $C,C'$ be the universal constants from Lemma~\ref{lem:coveringdensityimplications} and let
\begin{align}
\kappa=\frac{\sqrt{dD}}{(1+C\frac{\log{d}}{d})p^{-\frac{k}{d}}\reff(1)}. 
\label{eq:kappadef}
\end{align}
We define
\begin{align}
\Lambda_f=\kappa L(S),~~~\Lambda_c=\kappa L.    
\end{align}
Thus, $|\Lambda_f/\Lambda_c|=|L(S)/L|=p^k$, and Item~\ref{itemL1:rate} holds with probability $1$. Let $P_{\tilde{U}}= P_{U|\|U\|^2\leq r_U}$ and define $\eps'=\sqrt{\eps^2+2C\frac{\log d}{d}}$ such that $\frac{\eps'^2}{2}=\frac{\eps^2}{2}+C\frac{\log d}{d}$. Define the events 
\begin{align}
E^{\text{quant}}_f=\{p^{\frac{k}{d}}L(S)\in E_{c_0,c_1}\}~~,~~E^{\text{quant}}_c=\{L\in E_{c_0,c_1}\},~~E^{\text{code}}_c=\{L\in E^{P_{\tilde{U}},r_U,D}_{\alpha,\beta,\kappa,\eps'}\}. 
\label{eq:all3events}
\end{align}
and assume they all occur (later we will show that if $L\sim\mu_d$ and $S\sim\Unif(\mathrm{Gr}_{d,k}(\mathbb{F}_p)$ are statistically independent, the three events indeed occur simultaneously with positive probability). 

From Lemma~\ref{lem:coveringdensityimplications} we have that
\begin{align}
\rcov(\Lambda_f)&=\rcov\left(\kappa p^{-\frac{k}{d}}p^{\frac{k}{d}}L(S) \right)\leq \left(1+C\frac{\log{d}}{d}\right)\kappa p^{-\frac{k}{d}}\reff(1)=\sqrt{dD},\\
\rcov(\Lambda_c)&=\rcov(\kappa L)\leq \left(1+C\frac{\log{d}}{d}\right)\kappa \reff(1)=p^{k/d}\sqrt{dD}\leq 2^R\sqrt{dD},\\
\sigma^2(\Lambda_f)&=\sigma^2\left(\kappa p^{-\frac{k}{d}}p^{\frac{k}{d}}L(S) \right)\leq\frac{1}{d}\left(\left(1+C\frac{\log{d}}{d}\right)\kappa p^{-\frac{k}{d}}\reff(1) \right)^2=D\\
\frac{1}{d}\|R(\Lambda_f)\|_F^2&=\frac{1}{d}\|R(\kappa p^{-\frac{k}{d}}p^{\frac{k}{d}}L(S))\|_F^2\leq \left( \frac{1}{d}\left( \left(1+C\frac{\log{d}}{d}\right)\kappa p^{-\frac{k}{d}} \reff(1)\right)^2\right)^2\left(1+C'\frac{\log^3 d}{d}\right)\nonumber\\
&=D^2\left(1+C'\frac{\log^3 d}{d}\right)
\end{align}
Thus, $\Lambda_f$ and $\Lambda_c$ satisfy Items~\ref{itemL1:rcov}-\ref{itemL1:FrobCov}, with $C_2=C'$. To show that Item~\ref{item:L1pe} holds, let $\tilde{Z}\sim\Unif(\sqrt{dD}\m{B})$ be statistically independent of $U$ and write
\begin{align}
\Pr(\alpha U+\beta Z\notin \m{V}_{\Lambda_c})&=\Pr(\alpha U+\beta Z\notin \m{V}_{\kappa L})\leq \Pr(\|U\|^2> d r_u)+\Pr(\alpha U+\beta Z\notin \m{V}_{\kappa L}|~\|U\|^2\leq d r_u)\\
&\leq \Pr(\|U\|^2> d r_u)+d^C \Pr(\alpha U+\beta \tilde{Z}\notin \m{V}_{\kappa L}|~\|U\|^2\leq d r_u)\label{eq:boundbyunifonball}\\
&\leq \Pr(\|U\|^2> d r_u)+6e^{-d\frac{\eps^2}{2}},\label{eq:coarsegoodforcoding}
\end{align}
where~\eqref{eq:boundbyunifonball} follows from Item~\ref{item:Lpe} in Lemma~\ref{lem:coveringdensityimplications} and the definition of $\kappa$ in~\eqref{eq:kappadef}, and~\eqref{eq:coarsegoodforcoding} follows since $L\in E^{P_{\tilde{U}},r_U,D}_{\alpha,\beta,\kappa,\eps'}$ and since $d^C e^{-d\frac{\eps'^2}{2}}=e^{-d\frac{\eps^2}{2}}$. 

\medskip

It therefore remains to show that there exist $L\in\m{L}_d$ and $S\in\mathrm{Gr}_{d,k}(\mathbb{F}_p)$ for which $L$ is in $E_{c_0,c_1}$ and $E^{P_{\tilde{U}},r_U,D}_{\alpha,\beta,\kappa,\eps'}$ and $p^{\frac{k}{d}}L(S)$ is in $E_{c_0,c_1}$. To that end, let $L\sim\mu_d$ and let $S\sim\Unif(\mathrm{Gr}_{d,k}(\mathbb{F}_p))$ be statistically independent of $L$. By~\cite[Proposition 2.2]{orw22} we have that $p^{\frac{k}{d}}L(S)\sim\mu_d$. Thus, 
\begin{align}
\Pr&\left(L\in E_{c_0,c_1}, p^{\frac{k}{d}}L(S)\in E_{c_0,c_1}, L\in E^{P_{\tilde{U}},r_U,D}_{\alpha,\beta,\kappa,\eps'}  \right)\nonumber\\
&\geq 1- \Pr\left (L\notin E_{c_0,c_1}\right)-\Pr\left (p^{\frac{k}{d}}L(S)\notin E_{c_0,c_1}\right)- \Pr\left (L\notin E^{P_{\tilde{U}},r_U,D}_{\alpha,\beta,\kappa,\eps'}\right) \\
&=1-2\mu_d(\{L\notin E_{c_0,c_1}\})-\mu_d(\{L\notin E^{P_{\tilde{U}},r_U,D}_{\alpha,\beta,\kappa,\eps'}\})\\
&>\frac{1}{3}-\mu_d(\{L\notin E^{P_{\tilde{U}},r_U,D}_{\alpha,\beta,\kappa,\eps'}\}),
\end{align}
where we have used the union bound in the first inequality and Lemma~\ref{lem:ORW} in the last inequality. We will be able to use Lemma~\ref{lem:goodforcoding} with $r_b=D$ to deduce that $\mu_d(\{L\notin E^{P_{\tilde{U}},r_U,D}_{\alpha,\beta,\kappa,\eps'}\})<\frac{1}{3}$ and complete the proof, once we show that $\kappa$ in~\eqref{eq:kappadef} is greater than the right hand side of~\eqref{eq:kappamin}. To that end, we write 
\begin{align}
\frac{e^{\frac{\eps'^2}{2}}\sqrt{1+\eps'}\frac{\sqrt{d(\alpha^2 r_U +\beta^2 D)}}{\reff(1)}}{\frac{\sqrt{dD}}{(1+C\frac{\log{d}}{d})p^{-\frac{k}{d}}\reff(1)}}&=p^{-\frac{k}{d}}\sqrt{\alpha^2 \frac{r_U}{D}+\beta^2}e^{\frac{\eps'^2}{2}}\sqrt{1+\eps'}  (1+C\frac{\log{d}}{d})\\
&\leq 2^{-\left(R-\frac{1}{2}\log(\beta^2+\alpha^2\frac{r_U^2}{D}\right)}\cdot 2^{\frac{1}{d}}e^{\frac{\eps'^2}{2}}\sqrt{1+\eps'}  (1+C\frac{\log{d}}{d})\\
&\leq 2^{-C_1\left(\eps+\frac{\log d}{d}\right)}\cdot 2^{\frac{1}{d}}e^{\frac{\eps'^2}{2}}\sqrt{1+\eps'}  (1+C\frac{\log{d}}{d})\\
&<1
\end{align}
where the last inequality holds for some universal $C_1$ large enough.
\end{proof}

\begin{proof}[Proof of Lemma~\ref{lem:coveringdensityimplications}]
Let $C_0>0$ be a universal constant satisfying 
\begin{align}
e^{\frac{1}{d}\ln (c_0d^{c_1})}< 1+C_0\frac{\log{d}}{d},~~~\forall d.   
\end{align}
Thus, for $L\in E_{c_0,c_1}$ we have
\begin{align}
\frac{V_d \rcov(L)^d}{V_d \reff(L)^d}< c_0 d^{c_1}\Longleftrightarrow \frac{\rcov(L)}{\reff(L)}< e^{\frac{1}{d}\ln (c_0d^{c_1})}\Longrightarrow \frac{\rcov(L)}{\reff(L)}< 1+C_0\frac{\log{d}}{d}
\end{align}
Since $\rcov(\kappa L)=\kappa\rcov(L)$ and  $\reff(\kappa L)=\kappa\reff(L)$, Item~\ref{itemL:rcov} holds for any $C\geq C_0$. Item~\ref{itemL:D} follows since $\sigma^2(\kappa L)\leq \frac{1}{d}\rcov^2(\kappa L)$ for any lattice $L\subset \RR^d$.

\medskip

To prove Item~\ref{item:Lpe}, let $f_Z$ and $f_{\tilde{Z}}$ be the densities of the random variables $Z$ and $\tilde{Z}$, respectively. By Item~\ref{itemL:rcov} we have that for any $L\in E_{c_0,c_1}$ the support $\m{V}_{\kappa L}$ of $Z$ is contained in the support $\left(1+C_0\frac{\log{d}}{d}\right)\kappa\reff(L)\m{B}$ of $\tilde{Z}$. Thus, for any $z\in \left(1+C_0\frac{\log{d}}{d}\right)\kappa\reff(L)\m{B}$ we have that
\begin{align}
f_Z(z)\leq \frac{\Vol\left(\left(1+C_0\frac{\log{d}}{d}\right)\kappa \reff(L)\m{B} \right)}{\Vol\left(\m{V}_{\kappa L } \right)}f_{\tilde{Z}}(z)\leq  \left(1+C_0\frac{\log d}{d} \right)^{d}  f_{\tilde{Z}}(z)\leq d^C f_{\tilde{Z}}(z).
\end{align}
It therefore follows that
\begin{align}
f_{\alpha U+\beta Z}(x)\leq  d^C f_{\alpha U+\beta\tilde{Z}},~~~\forall x\in\RR^d,
\end{align}
and therefore for any $\m{A}\subset\RR^d$
\begin{align}
\Pr(\alpha U+\beta Z\notin \m{A})\leq    d^C \Pr(\alpha U+\beta \tilde{Z}\notin \m{A}).
\end{align}

\medskip

We move on to proving Item~\ref{itemL:FrobCov}. Let $\eig(R(\kappa L))=(\lambda_1,\ldots,\lambda_d)$ be the eigenvalues of $R(\kappa L)$, such that $\sum_{i=1}^d \lambda_i=d\sigma^2(\kappa L)$, and $\|R(\kappa L)\|_F^2=\sum_{i=1}^d \lambda_i^2$. Let $Z~\sim\Unif(\m{V}_{\kappa L})$. Since $\Vol(\m{V}_{\kappa L})=V_d \reff^d(\kappa L)$, we have that
\begin{align}
\log V_d \reff^d(\kappa L)=h(Z)\leq \frac{1}{2}\log\det \left((2\pi e)R(\kappa L)\right)=\sum_{i=1}^d\frac{1}{2}\log(2\pi e \lambda_i),    
\end{align}
where the upper bound follows since the Gaussian distribution maximizes entropy under covariance constraints. Thus,
\begin{align}
\sum_{i=1}^d\log(\lambda_i)&\geq 2\log V_d \reff^d(\kappa L)-d\log(2\pi e)\\
&=2\log \left(V_d^{2/d} \reff^2(\kappa L)\right)^{d/2}-d\log(2\pi e)\\
&=d\log\left(\frac{V_d^{2/d} \reff^2(\kappa L)}{2\pi e} \right)\\
&=d\log\left(\frac{1}{2\pi e}\frac{V_d^{2/d}  \reff^2(\kappa L)}{\sigma^2(\kappa L)} \sigma^2(\kappa L)\right)\\
&=d\log \sigma^2(\kappa L)-d\log 2\pi e \cdot N(\kappa L),
\end{align}
where $N(\kappa L)=\frac{\sigma^2(\kappa L)}{V^{2/d}_d \reff^2(\kappa L)}$ is the normalized second moment (NSM) of a lattice $\kappa L$ in $\RR^d$. 

Denote $\delta(\kappa L)=2\pi e \cdot G(\kappa L)-1$, and $\|R(\kappa L)\|_2=\max_{i=1,\ldots,n}\lambda_i$. We now show that
\begin{align}\label{eq:grx1}
\frac{1}{d}\sum_{i=1}^d \lambda^2_i\leq 
(\sigma^2(\kappa L))^2+ 2  \delta(\kappa L) \|R(\kappa L)\|^2_2
\end{align}
Indeed, in the range $0 < x \le \|R(\kappa L)\|_2$ the second derivative of $x \mapsto \ln x$ is
upper-bounded by $-{1\over \|R(\kappa L)\|_2^2}$. Thus, from Taylor's expansion around $x=\sigma^2(\kappa L)$ we have 
$$ \ln \lambda_i \le \ln \sigma^2(\kappa L) + {\lambda_i - \sigma^2(\kappa L)\over
\sigma^2(\kappa L)} - {1\over 2\|R(\kappa L)\|_2^2}(\lambda_i - \sigma^2(\kappa L))^2\,.$$
Summing over $i$ and using the facts that a) ${1\over d}\sum_i \lambda_i = \sigma^2(\kappa L)$ and
b)  ${1\over d}\sum_{i=1}^d \ln \lambda_i \ge \ln
\sigma^2(\kappa L) - \ln(1+\delta(\kappa L))$ we get after rearranging terms
$$ 2 \|R(\kappa L)\|_2^2 \ln(1+\delta(\kappa L)) \ge {1\over d} \sum_i(\lambda_i - \sigma^2(\kappa L))^2 = {1\over
d} \sum_i\lambda_i^2 - (\sigma^2(\kappa L))^2\,,$$
completing the proof of~\eqref{eq:grx1}.

To complete our statement, it remains to show that $\delta(\kappa L)\leq c_2\frac{\log d}{d}$ and $\|R(\kappa L)\|_2\leq c_3 \log d\frac{\rcov^2(\kappa L)}{d}$, for some universal constants $c_2,c_3>0$. This will imply, by~\eqref{eq:grx1}, that
\begin{align}
\frac{1}{d}\|R(\kappa L)\|_F^2&\leq (\sigma^2(\kappa L))^2+2c_2 c^2_3 \frac{\log d}{d}\cdot \log^2(d)\cdot \left(\frac{\rcov^2(\kappa L)}{d} \right)^2\\
&\leq \left( \frac{1}{d}\left( \left(1+C\frac{\log{d}}{d}\right)\kappa\reff(L)\right)^2\right)^2\cdot \left(1+2 c_2 c^2_3 \frac{\log^3 d}{d} \right),   
\end{align}
where we have used Items~\ref{itemL:rcov} and~\ref{itemL:D} in the last inequality.

For bounding $\delta(\kappa L)$, we use $\sigma^2(\kappa L)\leq \frac{1}{d}\rcov^2(\kappa L)$ and write
\begin{align}
 2\pi e \cdot N(\kappa L)\leq \frac{2\pi e}{d V_d^{2/d}}\frac{\rcov^2(\kappa L)}{\reff^2(\kappa L)} \leq \frac{2\pi e}{d V_d^{2/d}} \left(c_0 d^{c_1}\right)^{2/d}\leq 1+c_2\frac{\log d}{d},
\end{align}
where in the first inequality follows since $L\in E_{c_0,c_1}$, and in the second inequality we have used the fact that $\frac{2\pi e}{d V_d^{2/d}}=1+O(\frac{\log d}{d})$.

We now upper bound the operator norm $\|R(\kappa L)\|_2=\max_{v\in \mathbb{S}^{n-1}}\EE(v^\top Z)^2$, for $Z\sim\Unif(\m{V}_{\kappa L})$. To that end, 
let $\delta_d=\sqrt{\frac{2(C+1)\log d}{d}}$. For any $v\in\mathbb{S}^{n-1}$ we have that 
\begin{align}
\EE(v^\top Z)^2& =\EE\left[\|v\|^2 \|{\|Z\|^2}\cos^2(\angle(v,Z))\right]\\
&\leq \rcov^2(\kappa L)  \EE\left[\cos^2(\angle(v,Z))\right]\\
&\leq \rcov^2(\kappa L) \left(\delta^2_d+\Pr(|\cos(\angle(v,Z)|>\delta_d)\right)
\end{align}
Let $\m{A}=\{z\in\RR^d~:~|v^\top z|\leq \delta_d\}$, and $\tilde{Z}\sim\Unif\left(\left((1+C\frac{\log{d}}{d}\right)\kappa\reff(L)\m{B}\right)$. From Item~\ref{item:Lpe}, applied with $\alpha=0$ and $\beta=1$, we have
\begin{align}
\Pr(|\cos(\angle(v,Z)|>\delta_d)=\Pr(Z\notin \m{A})&\leq d^C \Pr(\tilde{Z}\notin \m{A})=d^C  \Pr(|\cos(\angle(v,\tilde{Z})|>\delta_d)\\
&=d^C  \Pr(|\cos(\angle(v,Z_B)|>\delta_d),  
\end{align}
where $Z_B\sim\Unif(\m{B})$. Thus,
\begin{align}
\EE(v^\top Z)^2&\leq \rcov^2(\kappa L) \left(\delta^2_d+d^C\Pr(|\cos(\angle(v,Z_B)|>\delta_d)\right)\\
&\leq \rcov^2(\kappa L)\left(\frac{2(C+1)\log d}{d}+d^C \Pr\left(|\cos(\angle(v,Z_B)|)>\sqrt{\frac{2(C+1)\log d}{d}}\right)\right)\\
&\leq c_3 \log d\cdot \frac{\rcov^2(\kappa L)}{d},\label{eq:Pcone}
\end{align}
where~\eqref{eq:Pcone} follows since $\Pr\left(|\cos(\angle(v,Z_B)|\right)>\sqrt{\frac{2(C+1)\log d}{d}})\leq e^{-\frac{2(C+1)\log d}{2}}=d^{-(C+1)}$, which follows from the fact that a spherical cap of height $1-\eps$ has volume (w.r.t. to $\mathbb{S}^{n-1})$ at most $e^{-d\eps^2/2}$ for $0<\eps<\frac{1}{\sqrt{2}}$~\cite[Section 7.2]{boucheron2003concentration}.     
\end{proof}

\begin{proof}[Proof of Lemma~\ref{lem:goodforcoding}]
For $r>0, \kappa>0$ and $x\in\RR^d$, let
\begin{align}
N^*(\kappa L,r\m{B},x)=\left|((\kappa L\setminus \{0\})+x)\cap r\m{B})\right|.    
\end{align}
Note that for any $r>0$  we have the inclusion of events
\begin{align}
\left\{x\in r\m{B},N^*(\kappa L,r\m{B},-x)=0 \right\}\subset\{x\in\m{V}_{\kappa L}\},
\end{align}
which implies that
\begin{align}
\{x\notin\m{V}_{\kappa L}\}\subset\left\{x\notin r\m{B}\right\}\cup\left\{N^*(\Lambda_{\kappa L},r\m{B},-x)>0 \right\}.
\end{align}
Let $X=\alpha U+\beta \tilde{Z}$. For any given lattice $\kappa L$, we therefore have that
\begin{align}
P_e(L)=\Pr(X\notin\m{V}_{\kappa L})&\leq \Pr(X\notin r\m{B})+\Pr\left(N^*(\kappa L,r\m{B},-X)>0\right)\\
&\leq \Pr(X\notin r\m{B})+\EE_X [N^*(\kappa L,r\m{B},-X)],
\end{align}
where the last inequality follows since $\Pr(N>0)\leq \EE[N]$ for a random variable $N$ supported on the non-negative integers. Taking the expectation with respect to $L\sim \mu_n$ gives
\begin{align}
\EE[P_e(L)]\leq \Pr(X\notin r\m{B})+\EE_L\EE_X [N^*(\kappa L,r\m{B},-X)],
\end{align}
Applying Siegel's summation formula, we have
\begin{align}
\EE_{\kappa L}\EE_X [N^*(\kappa L,r\m{B},-X)]   =\EE_X[\EE_{\kappa L} [N^*(\kappa L,r\m{B},-x)|X=x]]=\frac{\Vol(r\m{B})}{\mathrm{covol}(\kappa L)}=\left(\frac{r}{\kappa \reff(1)}\right)^d, 
\end{align}
so that
\begin{align}
\EE[P_e(L)]\leq \Pr(X\notin r\m{B})+\left(\frac{r}{\kappa \reff(1)}\right)^d.
\end{align}
Let $r^2=d(\alpha^2 r_U +\beta^2 r_b)(1+\eps)$, for $0<\eps<\frac{1}{\sqrt{2}}$ and let us upper bound the first term. Recalling that $X=\alpha U+\beta \tilde{Z}$, we have
\begin{align}
\Pr(X\notin r\m{B})&=\Pr((\alpha U+\beta \tilde{Z})^2>r^2)\\
&=\Pr(\alpha^2 \|U\|^2+\beta^2 \|\tilde{Z}\|^2+2\alpha \beta U^\top Z>r^2)\\
&\leq \Pr\left(d(\alpha^2 r_U+\beta^2 r_b)+2\alpha \beta U^\top \tilde{Z}>d(\alpha^2 r_U +\beta^2 r_b)(1+\eps)\right)\\
&=\Pr\left(2\alpha \beta U^\top \tilde{Z}>d(\alpha^2 r_U +\beta^2 r_b)\eps\right)\\
&\leq \Pr\left(2\alpha \beta d \sqrt{r_U r_b} \cos(\angle(U, \tilde{Z}))>d(\alpha^2 r_U +\beta^2 r_b)\eps\right)\\
&=\Pr\left( \cos(\angle(U, \tilde{Z}))>\frac{d(\alpha^2 r_U +\beta^2 r_b)}{2d \alpha \beta \sqrt{r_U r_b}}\eps\right)\\
&\leq\Pr\left( \cos(\angle(U, \tilde{Z}))>\eps\right)\label{eq:amgm}\\
&\leq e^{-d\eps^2/2}
\end{align}
where~\eqref{eq:amgm} follows from $(\alpha\sqrt{r_u}-\beta\sqrt{r_b})^2\geq 0$, and the last inequality follows from the fact that a spherical cap of height $1-\eps$ has volume (w.r.t. to $\mathbb{S}^{n-1})$ at most $e^{-d\eps^2/2}$ for $0<\eps<\frac{1}{\sqrt{2}}$~\cite[Section 7.2]{boucheron2003concentration}.
We have therefore obtained that
\begin{align}
\EE[P_e(L)]\leq e^{-d\eps^2/2}+  e^{-d\log\frac{\kappa \reff(1)}{r}}.
\end{align}
Taking 
\begin{align}
\kappa>e^{\frac{\eps^2}{2}}\frac{r}{\reff(1)}=(1+\eps)e^{\frac{\eps^2}{2}}\frac{\sqrt{d(\alpha^2 r_U +\beta^2 r_b)}}{\reff(1)},
\end{align} 
gives
\begin{align}
\EE[P_e(L)]\leq 2e^{-d\eps^2/2}.
\end{align}
Thus, using Markov's inequality, we obtain the claimed result.
\end{proof}

\section{Bounding the Effect of Overload Events}
\label{appendix:fromidealtoreal}

\begin{proof}[Proof of Proposition~\ref{prop:fromidealtoreal}]
Let
\begin{align}
e_{\OL}=\hat{U}_{[d]}^\top \hat{V}_{[d]}-\Uidd^\top \Vidd,    
\end{align}
such that
\begin{align}
e=\eid+\alpha e_{\OL}.
\end{align}
With probability $1$, we have that 
\begin{align}
|e_{\OL}|&<|\hat{U}_{[d]}^\top\hat{V}_{[d]}|+|\Uidd^\top \Vidd|\leq \|\hat{U}_{[d]}\|\cdot \|\hat{V}_{[d]}\|+\|\Uidd\|\cdot\|\Vidd\|\\
&\leq \rcov^2(\Lambda_c)+(\sqrt{n}+\rcov(\Lambda_f))^2,
\end{align}
where the last inequality follows since $\hat{U}_{[d]},\hat{V}_{[d]}\in\m{V}_{\Lambda_c}$ by definition, and since
\begin{align}
\|\Uidd\|=\|U_{[d]}+Z_1\|\leq \|U_{[d]}\|+\|Z_1\|\leq \|U\|+\|Z_1\|\leq \sqrt{n}+\rcov(\Lambda_f),\label{eq:uiddbound}    
\end{align}
and $\|\Vidd\|$ is bounded similarly. With probability $1$, we also have that
\begin{align}
|\eid|&=|\alpha\Uidd^\top \Vidd -U^\top V|\leq \alpha  |\Uidd^\top \Vidd|+|U^\top V|\leq \alpha \|\Uidd\|\cdot \|\Vidd\|+|\rho| n\\
& \leq \alpha (\sqrt{n}+\rcov(\Lambda_f))^2+|\rho| n\label{eq:uiddnormbound1}\\               
&\leq (\sqrt{n}+\rcov(\Lambda_f))^2+ n,\label{eq:alpharhobound}    
\end{align}
where~\eqref{eq:uiddnormbound1} follows from~\eqref{eq:uiddbound}, and~\eqref{eq:alpharhobound} from $0\leq \alpha,|\rho|\leq 1$.

Note that $\max\{\sqrt{n},\rcov(\Lambda_f)\}\leq M(\Lambda_c)$, since $\Lambda_c\subset\Lambda_f$, and therefore $\m{V}_{\Lambda_f}\subset \m{V}_{\Lambda_c}$. It therefore follows that with probability $1$ we have
\begin{align}
|e_{\OL}|\leq 5 M^2(\Lambda_c),~~|\eid|\leq 5 M^2(\Lambda_c).    
\end{align}
Consequently, 
\begin{align}
\EE(e^2)&=\EE(\eid^2)+\alpha^2\EE(e^2_{\OL})+2\alpha\EE[e_{OL}\eid]\\
&\leq \EE(\eid^2)+75\Pr(\OL)M^4(\Lambda_c),
\end{align}
as claimed, where we have used $0\leq \alpha\leq 1$ again.
\end{proof}

\section{Projections of Random Uniform Orthogonal Vectors}
\label{appendix:projections}

Recall that $S\sim\Unif(\mathrm{O}_n(\RR))$ and we denote $U=\sqrt{n}S_1$, $Z=\sqrt{n}S_2$ and $V=\rho U+\sqrt{1-\rho^2}Z$. Furthermore, $d=\lfloor \kappa n \rfloor$, and we denote $U_{[d]}=(U_1,\ldots,U_d)^\top$ and similarly $V_{[d]}=\sqrt{\rho} U_{[d]}+\sqrt{1-\rho^2} Z_{[d]}$.
We have
\begin{align}
\EE[(U_{[d]}^\top V_{[d]})^2]&=\EE[(\rho\|U_{[d]}\|^2+\sqrt{1-\rho^2} U_{[d]}^\top Z_{[d]})^2]=\rho^2\EE\|U_{[d]}\|^4+(1-\rho^2)\EE[(U_{[d]}^\top Z_{[d]})^2]\label{eq:sphereprojderivation}
\end{align}
where the last equation follows since $\EE[\|U_{[d]}\|^2 U_{[d]}^\top Z_{[d]}]=0$ from symmetry. For $d=n$ we trivially have $\EE\|U\|^4=n^2$ and then $\EE[(U^\top Z)^2]=0$. We proceed to compute $\EE\|U_{[d]}\|^4$ and then $\EE[(U_{[d]}^\top Z_{[d]})^2]$ for general $d\leq n$.

It holds that (using the fact that $\EE(U^2_i)=1$) 
\begin{align}
\EE\|U_{[d]}\|^2&=d.
\end{align}
By symmetry, we also have
\begin{align}
\EE(U_{[d]}^\top  Z_{[d]})&=0.
\end{align}
We will further use the fact that $\EE(U^4_i)=\frac{3n}{n+2}$~\cite{stam1982limit}.
To compute $\EE\|U_{[d]}\|^4$, we first note that, by symmetry
\begin{align}
 n^2=\EE\|U\|^4=\EE\left(\sum_{i=1}^n U_i^2\right)^2=n \EE(U_1^4)+n(n-1)\EE(U_1^2 U_2^2),    
\end{align}
which implies
\begin{align}
\EE(U_1^2 U_2^2)=\frac{n-\EE(U_1^4)}{n-1} =\frac{n}{n+2}.
\end{align}
With this, we can write
\begin{align}
\EE\|U_{[d]}\|^4= \EE\left(\sum_{i=1}^d U_i^2\right)^2=d \EE(U_1^4)+d(d-1)\EE(U_1^2 U_2^2)=\frac{n}{n+2}d(d+2).   \label{eq:Ud4Exp}
\end{align}

We move on to calculate $\EE(U_{[d]}^\top  Z_{[d]})^2$. We have that
\begin{align}
\EE(U_{[d]}^\top  Z_{[d]})^2=\EE\left( \sum_{i=1}^d U_i Z_i \right)^2=\sum_{i=1}^d \EE(U_i^2 Z_i ^2)+\sum_{j\neq i}\EE(U_i U_j Z_i Z_j)=d\xi +d(d-1)\nu,
\label{eq:partialsquaredcorrelation1}
\end{align}
where
\begin{align}
\xi=\EE(U_1^2 Z_1^2),~~~\nu=\EE(U_1 U_2 Z_1 Z_2),   
\end{align}
and the last equality in~\eqref{eq:partialsquaredcorrelation1} follows by symmetry. Taking $d=n$, we get that $U_{[n]}^\top  Z_{[n]}=U^\top  Z=0$ w.p. $1$. Invoking~\eqref{eq:partialsquaredcorrelation1} therefore gives 
\begin{align}
0=n\xi+n(n-1)\nu\Longrightarrow \nu=-\frac{\xi}{n-1}.  
\end{align}
Substituting this into~\eqref{eq:partialsquaredcorrelation1}, we obtain
\begin{align}
\EE(U_{[d]}^\top  Z_{[d]})^2=d\left(1-\frac{d-1}{n-1} \right)\xi=\frac{d(n-d)}{n-1}\xi.  
\label{eq:partialsquaredcorrelation2}
\end{align}
In order to compute $\xi$, define $e=U-Z$. Note that the symmetry and orthogonality of $U$ and $Z$ implies that $e\sim\Unif(\sqrt{2n}\mathbb{S}^{n-1})$, where $\mathbb{S}^{n-1}$ is the unit sphere in $\RR^n$. It therefore follows that
\begin{align}
 \EE(e_1^4)=4\EE(U_1^4).
\end{align}
On the other hand
\begin{align}
\EE(e_1^4)=\EE(U_1-Z_1)^4=\sum_{i=0}^4 {{4}\choose{i}}\EE(U_1^i Z_1^{4-i})=\EE(U_1^4)+\EE(Z_1^4)+6\EE(U_1^2 Z_1^2)+4\EE(U_1 Z_1^3)+4\EE(U_1^3 Z_1).    
\end{align}
By symmetry, we clearly have that $\EE(Z_1^4)=\EE(U_1^4)$. We claim that $\EE(U_1 Z_1^3)=0$. To see this, note that given $Z$, the distribution of $U$ is invariant to negation (in other words $p_{U|Z=z}(u)=p_{U|Z=z}(-u)$). By symmetry, this also implies that $\EE(U_1^3 Z_1)=0$.
We therefore have that
\begin{align}
&4\EE(U_1^4)=\EE(e_1^4)=2\EE(U_1^4)+6\xi\\
\Longrightarrow &\xi=\frac{\EE(U_1^4)}{3}=\frac{n}{n+2}.
\end{align}
Substituting this into~\eqref{eq:partialsquaredcorrelation2}, we obtain
\begin{align}
\EE(U_{[d]}^\top  Z_{[d]})^2=d\cdot\frac{n(n-d)}{(n+2)(n-1)}. \label{eq:UdZd2Exp}   
\end{align}

Consequently, by~\eqref{eq:sphereprojderivation},~\eqref{eq:Ud4Exp} and~\eqref{eq:UdZd2Exp}, we have
\begin{align}
\EE[(U_{[d]}^\top V_{[d]})^2]&=\frac{nd}{n+2}\left(\rho^2(d+2)+(1-\rho^2)\frac{n-d}{n-1}\right)\\
&=\frac{nd}{n+2}\left(\rho^2\frac{n(d+1)-2}{n-1}+\frac{n-d}{n-1}\right)\\
&\leq \rho^2 n\frac{d(d+1)}{n}+\frac{d(n-d)}{n},
\end{align}
where in the last inequality we have used the fact that $\frac{n}{(n-1)(n+2)}\leq \frac{1}{n}$ for all $n\geq 2$. This establishes our claim.


\section{Tail Probability of $U_{[d]}$}
\label{appendix:projUtail}

\begin{proposition}
Let $U$ be uniformly distributed on $\sqrt{n}\mathbb{S}^{n-1}$, and let $U_{[d]}=(U_1,\ldots,U_d)^\top$ be its projection on the first $1\leq d\leq n$ coordinates. Then, for any $0<\eps<1$
\begin{align}
\Pr(\|U_{[d]}\|^2>(1+\eps) d)\leq 2e^{-\left(\frac{\eps}{1+\eps}\right)^2\frac{1}{24}d}\leq 2e^{-\frac{\eps^2}{96}d} 
\end{align}
\label{prop:projUtail}
\end{proposition}

\begin{proof}
Let $Z\sim\m{N}(0,I_n)$ and note that $U$ has the same distribution as $\sqrt{n}\frac{Z}{\|Z\|}$. Let 
\begin{align}
X_1=\sum_{i=1}^dZ_i^2,\\
X_2=\sum_{i=d+1}^n Z_i^2,
\end{align}
and note that $X_1$ and $X_2$ are independent chi-squared random variables with $d$ and $n-d$ degrees of freedom, respectively. We therefore have
\begin{align}
\Pr\left(\|U_{[d]}\|^2>(1+\eps)d\right)&=\Pr\left(n\frac{X_1}{X_1+X_2}\geq (1+\eps)d \right)=\Pr\left((n-(1+\eps)d)X_1\geq (1+\eps)d X_2 \right)    \\
&=\Pr\left(\frac{1}{n-d}X_2\leq \left(\frac{1}{1+\eps}-\frac{d}{n-d}\frac{\eps}{1+\eps} \right)\frac{1}{d}X_1  \right)\\
&=\Pr\left(\frac{1}{n-d}X_2\leq \left(1-\frac{\eps}{1+\eps}\left(1+\frac{d}{n-d}\right) \right)\frac{1}{d}X_1  \right)\\
&=\Pr\left(\frac{1}{n-d}X_2\leq \left(1-\frac{\eps}{1+\eps}\frac{n}{n-d} \right)\frac{1}{d}X_1  \right).\label{eq:UdNormExpression}
\end{align}
Note that for $\frac{\eps}{1+\eps}\cdot \frac{n}{n-d}\geq 1$ the probability above is zero. Thus, for the remainder of the proof, we assume $\frac{\eps}{1+\eps}\cdot \frac{n}{n-d}< 1$.
For any $t>0$ we have that
\begin{align}
\Pr\left(\|U_{[d]}\|^2>(1+\eps)d\right)\leq \Pr\left(\frac{1}{d}X_1>t\right)+\Pr\left(\frac{1}{n-d}X_2\leq\left(1-\frac{\eps}{1+\eps}\frac{n}{n-d} \right)t\right).    
\end{align}
Let $t=1+\delta$ for some $0<\delta<1$, and let
\begin{align}
\delta'=1-\left(1-\frac{\eps}{1+\eps}\frac{n}{n-d} \right)(1+\delta).    
\end{align}
By standard Chenroff bounds on the tail of the chi-squared distribution
\begin{align}
\Pr\left(\frac{1}{d}X_1>t\right)=\Pr\left(\frac{1}{d}X_1>1+\delta\right)\leq \exp\left\{\frac{d}{2}\left(\ln(1+\delta)-\delta \right) \right\}\leq \exp\left\{-\frac{d\delta^2}{16} \right\},    
\end{align}
where we have used $\ln(1+\delta)-\delta<\frac{\delta^2}{8}$ for all $0<\delta<1$. Similarly, if $0<\delta'<1$ we have
\begin{align}
\Pr\left(\frac{1}{n-d}X_2<\left(1-\frac{\eps}{1+\eps}\frac{n}{n-d} \right)t\right)&=\Pr\left(\frac{1}{n-d}X_2<1-\delta'\right)\leq \exp\left\{\frac{n-d}{2}\left(\ln(1-\delta')+\delta' \right) \right\}\\
&\leq \exp\left\{-\frac{(n-d)\delta'^2}{16} \right\}.    
\end{align}
We will choose
\begin{align}
\delta=\eps'\frac{1+\eta^2}{1+\eta-\eps'(1+\eta^2)},~~\eps'=\frac{\eps}{1+\eps},~~\eta=\sqrt{\frac{d}{n-d}},    
\end{align}
such that
\begin{align}
 \frac{d\delta^2}{16}=\frac{(n-d)\delta'^2}{16}   
\end{align}
and $\delta,\delta'>0$. Note further, that
\begin{align}
\delta\geq \eps' \frac{1+\eta^2}{1+\eta}\geq \eps' 2(\sqrt{2}-1),
\end{align}
and we therefore have that
\begin{align}
 \frac{(n-d)\delta'^2}{16}=\frac{d\delta^2}{16}\geq d \eps'^2 \frac{(\sqrt{2}-1)^2}{4}\geq \frac{d \eps'^2}{24},    
\end{align}
yielding the claimed result.
\end{proof}

\end{appendices}

\bibliographystyle{IEEEtran}
\bibliography{IPRD_Bib}

\end{document}